\documentclass[11pt,envcountsame,letter]{llncs}
\usepackage{latexsym,amsmath} 
\usepackage{times}

\def\vec#1{\mathchoice{\mbox{\boldmath$\displaystyle#1$}}
{\mbox{\boldmath$\textstyle#1$}}
{\mbox{\boldmath$\scriptstyle#1$}}
{\mbox{\boldmath$\scriptscriptstyle#1$}}}

\newcommand{\rkNAE}{r_{k\mathrm{-NAE}}}
\newcommand{\rfirst}{r_{\mathrm{first}}}
\newcommand{\rsecond}{r_{\mathrm{second}}}
\newcommand{\rsh}{r_{\mathrm{sh}}}
\newcommand{\rcond}{r_{\mathrm{cond}}}

\newcommand{\blue}{\mathtt{blue}}
\newcommand{\red}{\mathtt{red}}

\usepackage[dvips]{epsfig} 
\graphicspath{{./Fig_eps/}{./Eps/}} 
\DeclareGraphicsExtensions{.ps,.eps} 
\usepackage{color} 
\usepackage{fullpage}

\newcommand\dist{\mbox{dist}}

\newcommand\PHI{\vec\Phi} 
\newcommand\PHId{\PHI_{\vec d}}

\newcommand\cA{\mathcal{A}} 
\newcommand\cB{\mathcal{B}} 
\newcommand\cC{\mathcal{C}} 
\newcommand\cD{\mathcal{D}} 
\newcommand\cF{\mathcal{F}} 
 
\newcommand\cE{\mathcal{E}}

\newcommand\cS{\mathcal{S}} 
\newcommand\cT{\mathcal{T}}

\newcommand\cP{\mathcal{P}} 
\newcommand\cX{\mathcal{X}} 
 
\newcommand\cV{\mathcal{V}} 
 
\newcommand\cZ{\mathcal{Z}} 

\def\cR{{\mathcal R}}
\def\cC{{\mathcal C}}
\def\cE{{\cal E}}

\newcommand\eul{\mathrm{e}} 
\newcommand\eps{\varepsilon} 
\newcommand\ZZ{\mathbf{Z}} 
\newcommand\Var{\mathrm{Var}} 
\newcommand\Erw{\mathrm{E}} 
\newcommand\pr{\mathrm{P}}

\newcommand{\vecone}{\vec{1}}

\newcommand{\Vol}{\mathrm{Vol}}

\newcommand{\Po}{{\rm Po}} 
\newcommand{\Bin}{{\rm Bin}}

\newcommand{\bink}[2] {{{#1}\choose {#2}}}

\newcommand\ra{\rightarrow} 

\newcommand\bc[1]{\left({#1}\right)} 
\newcommand\cbc[1]{\left\{{#1}\right\}} 
\newcommand\bcfr[2]{\bc{\frac{#1}{#2}}} 
 
\newcommand\brk[1]{\left\lbrack{#1}\right\rbrack}

\newcommand\abs[1]{\left|{#1}\right|}

\newcommand{\Whp}{W.h.p.} 
\newcommand{\whp}{w.h.p.}

\newcommand{\Exp}{\Erw}

\newcommand\Lem{Lemma}
\newcommand\Prop{Proposition}
\newcommand\Thm{Theorem}
\newcommand\Cor{Corollary}
\newcommand\Sec{Section}

\numberwithin{equation}{section}
\numberwithin{theorem}{section}
\numberwithin{lemma}{section}
\numberwithin{proposition}{section}
\numberwithin{corollary}{section}

\begin{document} 

\title{Catching the $k$-NAESAT Threshold}

\thispagestyle{empty} 

\author{
Amin Coja-Oghlan\inst 1\thanks{%University of Warwick, Zeeman building, Coventry CV4 7AL, UK,
	%{\tt a.coja-oghlan@warwick.ac.uk}.
	Supported by EPSRC grant EP/G039070/2 and ERC Starting Grant 278857--PTCC (FP7).}
	 and
	Konstantinos Panagiotou\inst 2
} 
\date{\today} 
\institute{University of Warwick, Mathematics and Computer Science, Zeeman building, Coventry CV4~7AL, UK\\
	\email{a.coja-oghlan@warwick.ac.uk}
\and
Max-Planck-Institute for Informatics, 
Campus E1.4,
66123 Saarbr\"ucken,
Germany\\
\email{kpanagio@mpi-inf.mpg.de}} 

\maketitle

\begin{abstract}
\noindent
The best current estimates of  the thresholds for the existence of solutions in random constraint satisfaction problems (`CSPs') mostly derive from the \emph{first} and the \emph{second moment method}. Yet apart from a very few exceptional cases these methods do not quite yield matching upper and lower bounds. According to deep but non-rigorous arguments from statistical mechanics, this discrepancy
is due to a change in the geometry of the set of solutions called \emph{condensation} that occurs shortly before the actual threshold for the existence of solutions (Krzakala, Montanari, Ricci-Tersenghi, Semerjian, Zdeborov\'a: PNAS~2007). To cope with condensation, physicists have developed a sophisticated but non-rigorous formalism called \emph{Survey Propagation} (M\'ezard, Parisi, Zecchina: Science 2002). This formalism yields precise conjectures on the threshold values of many random CSPs. Here we develop a new \emph{Survey Propagation inspired second moment method} for the random $k$-NAESAT problem, which is one of the standard benchmark problems in the theory of random CSPs. This new technique allows us to overcome the barrier posed by condensation rigorously. We prove that the threshold for the existence of solutions in random $k$-NAESAT is $2^{k-1}\ln2-(\frac{\ln2}2+\frac14)+\eps_k$, where $|\eps_k| \le 2^{-(1-o_k(1))k}$, thereby verifying the statistical mechanics conjecture for this problem.

\end{abstract}

\spnewtheorem{fact}[theorem]{Fact}{\bfseries}{\itshape}

\newpage

\setcounter{page}{1}

%What is new/why is it interesting?
%\begin{itemize}
%\item we obtain (pretty much) the precise threshold in a difficult problem for the first time.
%\item difficult means that there is a condensation phase where $\Erw\ln Z<\ln\Erw Z$.
%\item the only other problems where the threshold was known quite accurately are problems such
%		as $k$-XORSAT.
%\item Comment on independent sets?
%\item We obtain this result by a new trick: we apply the second moment method to clusters rather than to individual solutions.
%\item More precisely, we parametrize solutions by their cluster size, thereby following the stat mech prescription.
%\item Since only the very biggest clusters condense, the second moment method can be applied to the rest.
%\item Say that this is precisely the idea behind the Survey Propagation formalism.
%\item This allows us to get a very precise estimate of both the complexity and the partition function.
%\end{itemize}
%
%This is too technical!
%We need to tell the stat mech story at leisure.
%
%Explain the concept of entropy crisis.
%Say that the result is of interest in physics.
%
%Say that we also find a way to study combinatorial properties of the various small clusters
%because we have a planting trick.
%
%We basically verify the stat mech picture that motivates Survey Propagation for the first time:
%	we show that throughout the condensation phase, there are small well-separated clusters.
%Say that this does not mean that Survey Propagation guided decimation works.

\section{Introduction}\label{XSec_Intro}

%\aco{Maybe a grander entrance?? Start with physics!}

Over the past decade, physicists have developed sophisticated but non-rigorous techniques
for the study of random constraint satisfaction problems (`CSPs') such as random $k$-SAT or random graph $k$-coloring~\cite{pnas,MPZ}.
This work has led to a remarkably detailed \emph{conjectured} picture,
according to which various phase transitions affect both the combinatorial and computational nature of random problems.
By now, some of these predictions have been turned into rigorous theorems.
Examples include results on the ``shattering'' of the solution space~\cite{Barriers,fede},
	work on (non-)reconstruction and sampling~\cite{Charis,GM,MRT},
		and even new algorithms for random CSPs~\cite{BetterAlg,WP}.
Many of these contributions have led to the development of new rigorous techniques.
Indeed, it seems fair to say that, combined, these results have advanced our understanding of random CSPs quite significantly.

However, thus far substantial bits of the statistical mechanics picture have eluded all rigorous attempts.
Perhaps most importantly, apart from a very few special cases, the precise thresholds for the existence of solutions
in random CSPs have not been pinned down exactly.
While rigorous upper and lower bounds can be derived via the \emph{first} and the \emph{second moment method}~\cite{ANP},
these bounds do not quite match in most examples, including prominent ones such
	as random $k$-SAT or random graph $k$-coloring.
In fact, the statistical mechanics techniques suggest a striking explanation for this discrepancy,
	namely the existence of a \emph{condensation phase} shortly before the threshold for the existence of solutions.
		In this phase, a crucial necessary condition for the success of the (standard) second moment method is violated.
%In this p condensation phase, a crucial necessary condition
%According to sophisticated but non-rigorous arguments from statistical mechanics,
%the explanation for this discrepancy is a phenomenon called \emph{condensation}.
%which occurs shortly before the threshold for the existence of solutions.
Indeed, in statistical mechanics a deep  formalism called \emph{Survey Propagation} (`SP') has been developed expressly
to deal with condensation.
While SP is primarily an analysis technique, an off-spin has been the \emph{SP guided decimation} algorithm,
which seems highly successful at solving random CSPs experimentally.
%	(see \Sec~\ref{Sec_related}).

%%\noindent
%%{\bf\em Hunting the $k$-NAESAT threshold.}
%Over the past decade, substantial progress has been made in estimating the
%thresholds for the existence of solutions in many random constraint satisfaction
%problems (`CSPs').
%The best current bounds generally derive from the \emph{first} and the \emph{second moment method}.
%However, %in most cases, particularly in prominent examples such as random $k$-SAT or random graph $k$-coloring,
%apart from a very few speical cases,
%the upper and lower bounds obtained via these techniques do not quite match.
%%For example, the best current lower bound on the random $k$-SAT threshold, derived via the second moment method,
%%is $2^k\ln2-O(k)$, while the first moment upper bound is $2^k\ln2$, leaving an additive gap of $O(k)$.
%According to sophisticated but non-rigorous arguments from statistical mechanics,
%the explanation for this discrepancy is a phenomenon called \emph{condensation}.
%%which occurs shortly before the threshold for the existence of solutions.
%Indeed, in statistical mechanics, a deep albeit non-rigorous formalism called \emph{Survey Propagation} has been developed expressly
%to deal with condensation.
%%This formalism has led to conjectures on the precise location of the satisfiability thresholds
%%in a wide range of random CSPs, including and particularly random $k$-SAT and random graph $k$-coloring.

%inspired an algorithm called \emph{Survey Propagation guided decimation} that is experimentally
%highly successful at solving random CSPs.

In this paper we propose a new \emph{SP-inspired second moment method} that
allows us to overcome the barrier posed by condensation.
%The proof technique is directly inspired by the SP formalism.
The specific problem that we work with is random $k$-NAESAT, 
one of the standard benchmark problems in the theory of random CSPs.
Random $k$-NAESAT is technically a bit simpler than
random $k$-SAT due to a certain symmetry property,
but computationally and structurally both problems have strong similarities.
We determine the threshold for the existence of solutions in random $k$-NAESAT
up to an additive error that tends to zero exponentially with  $k$.
This is the first time that the threshold in any random CSP %in which condensation occurs
	of this type
can be calculated with such accuracy.
%\aco{Indeed, we confirm an important statistical mechanics hypothesis %on the solution space geometry
%that underpins %the (non-rigorous) derivation of
%SP.}
%While %from a technical viewpoint 
%$k$-NAESAT is perhaps the simplest example of a random CSP
%that exhibits condensation, we believe that with (substantial) additional technical work
%our methods can be extended to other problems such as random $k$-SAT or random graph $k$-coloring.
While from a technical viewpoint $k$-NAESAT is perhaps the simplest example of a random CSP
that exhibits condensation, our proof technique 
rests on a rather generic approach.
%It should readily generalized to the broad class of ``symmetric'' problems from~\cite{MRT}.
Therefore, we
believe that with %(substantial)
 additional technical work
our approach can be extended to many other problems, including random $k$-SAT or random graph $k$-coloring.

%Over the past decade, substantial progress has been made in determining thresholds.
%The best current bounds are based on the first and the second moment method.
%However, in most problems the upper and lower bounds obtained via these
%techniques do not quite match.
%According to non-rigorous but sophisticated arguments from stat mech,
%this is due to a phenomenon called the \emph{condensation phase}.
%In this paper, we present an \emph{enhanced second moment method} that allows
%us to overcome this obstacle.
%Our results yield (almost) the precise threshold for a problem, namely $k$-NAESAT, that has a condensation
%phase for the first time.
%We believe that, with (substantial) additional technical work, our enhanced second moment method extends to other
%prominent random CSPs such as random $k$-SAT or random graph $k$-coloring.
%In the few examples where the precise threshold is known, the condensation phase
%	does not exist (XORSAT).

%\smallskip\noindent{\bf\em The main result.}
To define random $k$-NAESAT formally, let $k\geq3$ and $n>0$ be integers
and let $V=\cbc{x_1,\ldots,x_n}$ be a set of Boolean variables.
For a fixed real $r>0$ we let $m=m(n)=\lceil rn\rceil$.
Further, let $\PHI=\PHI_k(n,m)$ be a propositional formula obtained by choosing $m$
clauses of length $k$ over $V$ uniformly and independently at random among all $(2n)^k$ possible clauses.
We say that an assignment $\sigma:V\ra\cbc{0,1}$ is an \emph{NAE-solution} (a ``solution'')
if each clause has both a literal that evaluates to `true' under $\sigma$ and one that evaluates to `false'.
In other words,  both $\sigma$ and its inverse $\bar\sigma:x_i\mapsto 1-\sigma(x_i)$ 
are satisfying assignments of the Boolean formula $\PHI$.
We say that an event occurs \emph{with high probability} (``\whp'')
if its probability tends to one as $n\ra\infty$.

Friedgut~\cite{EhudHunting} proved that for any $k$ there exists a \emph{sharp threshold sequence} $\rkNAE=\rkNAE(n)$
such that for any fixed $\eps>0$ \whp\ $\PHI$ has a NAE-solution if $r<\rkNAE-\eps$,
while \whp\ $\PHI$ fails to have one if $r>\rkNAE+\eps$.
It is widely conjectured but as yet unproven that the threshold sequence converges for any $k\geq3$.
The best previous bounds on $\rkNAE$ were derived by Achlioptas and Moore~\cite{nae} and Coja-Oghlan and Zdeborov\'a~\cite{Lenka} via the first/second moment method:
	\begin{equation}\label{eqprevious}
	\textstyle \rsecond=2^{k-1}\ln2-\ln2+o_k(1)\leq \rkNAE\leq \rfirst=2^{k-1}\ln2-\frac{\ln 2}2+o_k(1),
	\end{equation}
where $o_k(1)$ hides a term that tends to $0$ for large $k$.
This left an additive gap of $\frac12\ln2\approx0.347$, which our main result closes.

\begin{theorem}\label{Thm_NAE}
There is a sequence $\eps_k=2^{-(1-o_k(1))k}$ such that %\whp
%\begin{enumerate}
%\item[$\bullet$] $\PHI$ has a NAE-solution if $r<2^{k-1}\ln2-\bc{\frac{\ln2}2+\frac14}-\eps_k$,
%\item[$\bullet$] $\PHI$ does not have a NAE-solution if $r>2^{k-1}\ln2-\bc{\frac{\ln2}2+\frac14}+\eps_k$.
%\end{enumerate}
%Hence, %for large enough $n$ we have%\Thm~\ref{Thm_NAE} shows that
	\begin{equation}\label{eqnew}
	\textstyle 2^{k-1}\ln2-\bc{\frac{\ln2}2+\frac14}-\eps_k\leq \rkNAE\leq2^{k-1}\ln2-\bc{\frac{\ln2}2+\frac14}+\eps_k.
	\end{equation}
\end{theorem}
While the numerical improvement obtained in~\Thm~\ref{Thm_NAE} may seem modest, we are going to argue that the result is conceptually
quite significant for two reasons.
First, we obtain (virtually) matching upper and lower bounds for the first time in a random CSP of this type.
%	\Thm~\ref{Thm_NAE} determines the $k$-NAESAT threshold up to $\eps_k=O_k(k^4/2^k)$,
%		a quantity that tends to zero exponentially for large $k$.
Second, and perhaps even more importantly, we devise a rigorous method for taming the condensation phenomenon.
Indeed,
condensation has been the main obstacle to determining the precise thresholds in random CSPs for the past decade.
To understand why,  we need to discuss the statistical mechanics picture
	and its relation to the second moment method.

%\aco{Mention NP-hardness.}
%
%\aco{Say that the stat mech picture is ``universal''.}
%
%
%\aco{Say that condensation is the key to understanding BP/SP decimation.}

%\section{The Second Moment Method and the Statistical Mechanics Perspective}

%\smallskip\noindent{\bf\em The statistical mechanics perspective.}
%Using insightful but non-rigorous techniques, statistical physicists have put forward striking hypotheses
%on the precise location of phase transitions in random CSPs, and on the combinatorial phenomena associated with them.
%%The general research goal in this area is to put these results on a rigorous foundation.
%To put our results in perspective, we briefly discuss the (non-rigorus) statistical mechanics picture for random $k$-NAESAT
%worked out in~\cite{pnas}.
%particularly the condensation phenomenon.
%

%Let $Z(\PHI)$ denote the number of NAE-solutions of $\PHI$.

\section{Condensation and the second moment method}

\noindent{\bf\em The statistical mechanics perspective.}
We follow~\cite{pnas} to sketch the non-rigorous statistical mechanics
	approach on random $k$-NAESAT.
Let $\cS(\PHI)\subset\cbc{0,1}^n$ denote the set of NAE-solutions of  $\PHI$, and let
$Z(\PHI)=\abs{\cS(\PHI)}$ be the number of solutions.
We turn $\cS(\PHI)$ into a graph by considering two solutions $\sigma,\tau$ adjacent if their Hamming distance is $o(n)$.
According to~\cite{pnas}, the `shape' of  $\cS(\PHI)$ undergoes two substantial changes  \whp\
at certain densities
	$0<\rsh<\rcond<\rkNAE.$

The first transition occurs at $\rsh\sim2^{k-1}\ln(k)/k$,
	%=(1+o_k(1)),
	almost a factor of $k$ below $\rkNAE$.
Namely, for $r<\rsh$, $\cS(\PHI)$ is (essentially) a connected graph.
But in the \emph{shattering phase} $\rsh<r<\rcond$, $\cS(\PHI)$ splits into connected components $S_1,\ldots,S_{N(\PHI)}$
called \emph{clusters} that are mutually separated by a linear Hamming distance $\Omega(n)$.
Each cluster $S_i$ only comprises an exponentially small fraction of $\cS(\PHI)$.
%	(i.e., $|S_i|\leq\exp(-\Omega(n))|\cS(\PHI)|$ for all $1\leq i\leq N$).
In particular, the total number $N(\PHI)$ of clusters, the so-called \emph{complexity}, is exponential in $n$.
This ``shattering'' of $\cS(\PHI)$ was indeed established rigorously in~\cite{Barriers}.

%In the regime $r_{sh}<r<r_{cond}$, 

As the density $r$ increases beyond $\rsh$,
both the overall number $Z(\PHI)$ of solutions and the number and sizes of the clusters shrink.
However, %according to the physics arguments, 
the cluster sizes decrease at a slower rate than $Z(\PHI)$,
until
%, 
at density $\rcond=2^{k-1}\ln2-\ln2+o_k(1)$ the largest cluster has size $\Omega(Z(\PHI))$ \whp\
%Hence, the largest cluster constitutes a \emph{constant} fraction of the entire solution space.
%According to the statistical mechanics picture,
%The second change occurs as $r$ passes the threshold $r_{cond}=2^{k-1}\ln2-\ln2+o_k(1)$. %, i.e., within an additive constant of $\rkNAE$.
%According to the statistical mechanics picture,
In effect, in the \emph{condensation phase} $\rcond<r<\rkNAE$, the set $\cS(\PHI)$ still decomposes into an exponential number of clusters
$S_1,\ldots,S_{N(\PHI)}$, each of tiny diameter and all mutually separated by Hamming distance $\Omega(n)$.
But in contrast to the shattered phase, now the largest cluster contains a \emph{constant} fraction of the entire set $\cS(\PHI)$.
Indeed, \whp\ a \emph{bounded} number of clusters contain a $1-o(1)$-fraction of all solutions.

The %(unproven) 
dominance of a few large clusters in the condensation phase complicates the probabilistic nature of the problem dramatically.
To see why, consider the experiment of
first choosing a random formula $\PHI$, and then picking two solutions $\vec\sigma,\vec\tau\in\cS(\PHI)$ uniformly and independently.
For $\rsh<r<\rcond$, $\vec\sigma,\vec\tau$ likely belong to different clusters, and hence
can be expected to have a ``large'' Hamming distance.
In fact, it is implicit in the previous work on the second moment method that
	%\begin{equation}\label{eqDecorrelate}
	$\dist(\vec\sigma,\vec\tau)\sim n/2\mbox{ \whp}$~\cite{nae,Lenka}.
%	\end{equation}
Intuitively, this %(\ref{eqDecorrelate})
 means that the two random solutions ``decorrelate''. %, which is a necessary condition
%for the success of the type of second moment argument used in~\cite{nae,Lenka}.
%
%As shown in~\cite{Lenka}, a consequence of~(\ref{eqDecorrelate}) is that
%for $r<r_{cond}$ the number $Z(\PHI)$ of solutions is ``concentrated about its expectation''.
%More precisely, taking into account that $Z(\PHI)$ is (typically) exponential in $n$, we have
%	\begin{equation}\label{eqannealed}
%	\ln Z(\PHI)\sim\ln\Erw\brk{Z(\PHI)}\qquad\mbox{\whp\ for $r<r_{cond}$}.
%	\end{equation}
%
%This is a further necessary condition for the success of the ``standard'' second moment argument.
By contrast, 
%according to the statistical mechanics picture,
for $\rcond<r<\rkNAE$ both $\vec\sigma,\vec\tau$
belong to the same large cluster with a non-vanishing probability.
In effect, with a non-vanishing probability their distance $\dist(\vec\sigma,\vec\tau)$ %\leq O_k(2^{-k})n$
is tiny,
%	(and thus~(\ref{eqDecorrelate}) is violated),
reflecting  that solutions in the same cluster are heavily correlated.

The purpose of the physicists' \emph{Survey Propagation} technique is precisely to deal with this type of correlation.
%	that emerge in the condensation phase.
The basic idea is to work with a different, non-uniform probability distribution on $\cS(\PHI)$.
This \emph{SP distribution} is induced by %instead of choosing $\sigma\in\cS(\PHI)$ uniformly,
	first choosing a \emph{cluster} $S_i$ uniformly at random among $S_1,\ldots,S_{N(\PHI)}$, and then
		selecting a solution in that cluster $S_i$ uniformly.
Since the number $N(\PHI)$ of clusters is (thought to be) exponential in $n$
throughout the condensation phase, % right up to the actual threshold $\rkNAE$,
two solutions $\vec\sigma',\vec\tau'$ chosen independently
 from the SP distribution are expected to lie in distinct clusters and thus to decorrelate \whp\
%	(i.e., $\dist(\vec\sigma',\vec\tau')\sim n/2$).
%Indeed, our proof of \Thm~\ref{Thm_NAE} implies that
%	$\dist(\vec\sigma',\vec\tau')=(\frac12+o_k(1))n$  \whp\
%Formally, with $\nu_{\PHI,x_1,\ldots,x_l}$ the joint distribution of $\vec\tau(x_1),\ldots,\vec\tau(x_l)$,
%it should be true that \whp\ for all $r_{cond}<r<\rkNAE$,
%	\begin{equation}\label{eqSPfactorizing}
%	%\Erw_{\PHI}
%	\norm{\nu_{\PHI,x_1,\ldots,x_l}-\nu_{\PHI,x_1}\otimes\cdots\otimes\nu_{\PHI,x_l}}_{tv}=o(1). %\qquad\mbox{for $r<r_{cond}$}.
%	\end{equation}
%%\aco{We will prove this.}	
%Our proof of \Thm~\ref{Thm_NAE} implies that this is indeed the case.
%More precisely, we have
%%
%\begin{corollary}\label{Cor_SP}
%For all $r$ satisfying~(\ref{eqcondInterval}) the correlation decay condition~(\ref{eqSPfactorizing}) holds \whp
%\end{corollary}
%In other words, the SP distribution reduces the `weight' of the few big clusters
%by first choosing a \emph{cluster} uniformly at random.
%This decorrelation hypothesis is the 
%
%asymptotic independence condition~(\ref{eqSPfactorizing}) is 
%physicists' Survey Propagation formalism is based on the hypothesis that~(\ref{eqSPfactorizing}) is true.
%one of the 
%main assumption underpinning the physicists' Survey Propagation formalism.
%More precisely, starting from~(\ref{eqSPfactorizing}), a sequence of delicate (and non-rigorous)

%\aco{If this decorrelation hypothesis is true, 

Starting from this (appropriately formalized) decorrelation assumption,
the SP formalism prescribes a sequence of delicate (non-rigorous)
steps to reduce the computation of the \emph{precise} threshold $\rkNAE$ to the solution
of a continuous variational problem for any $k\geq3$~\cite{DRZ08,Panchenko}.
This variational problem is itself highly non-trivial, but heuristic numerical
techniques yield plausible approximations for small values of $k$~\cite{Mertens}.
%value of $\rkNAE$ have been computed via heuristic methods~\cite{} for a few values of $k$.
Moreover, asymptotically for large $k$ the variational problem can be solved analytically.
This led to the conjecture that
	$\rkNAE=2^{k-1}\ln2-\bc{\frac{\ln2}2+\frac14}+o_k(1)$~\cite{DRZ08},
which  \Thm~\ref{Thm_NAE} resolves.

Is \Thm~\ref{Thm_NAE} ``optimal''?
%The  Survey Propagation formalism  is primarily a (non-rigorous) analysis tool for studying random CSPs~\cite{pnas}.
%According to this formalism, the \emph{precise} threshold $\rkNAE$ can be expressed as the solution to a certain continous problem.
Of course, it would be interesting to prove that for any $k$, the \emph{precise} threshold $\rkNAE$ equals
the solution to the variational problem that the SP formalism spits out.
However, given that this continuous problem itself appears difficult to solve analytically  (to say the very least),
it seems that such a result would merely establish the equivalence of two hard mathematical problems.
%
%although the continuous problem  may be difficult (or impossible) to solve analytically.
%%, it would be interesting to prove that its solution
%%is indeed equal to $\rkNAE$ for any $k\geq3$.
%Such a result would fill the gap left by \Thm~\ref{Thm_NAE}, which yields $\rkNAE$ up to an error that decays exponentially in $k$.
Thus, we believe that \Thm~\ref{Thm_NAE} marks the end of the line as far as an analytic/explicit computation of $\rkNAE$ is concerned.
%apart from perhaps minor improvments of the (already exponentially small) error term $\eps_k$.

%\aco{Can one hope to improve \Thm~\ref{Thm_NAE}? What about small $k$?
%	``this might well amount to merely reducing one difficult problem to another.''}

%In statistical mechanics, highly sophisticated but non-rigorous techniques have been developed
%for analyzing the SP distribution.
%This formalism makes it possible to derive a conjecture on the \emph{precise} value of the $k$-NAE
%threshold $\rkNAE$ for any $k\geq3$.
%However, the result is not an explicit formula from which $\rkNAE$ could be read off,
%but a highly intricate system of interal equations~\cite{}.
%Nonetheless, it is possible to extract an asymptotic expansion  of $\rkNAE$ in $k$:
%we should have
%	\begin{equation}\label{eqconj}
%	\rkNAE=2^{k-1}\ln2-\bc{\frac{\ln2}2+\frac14}+o_k(1).
%	\end{equation}

\smallskip
\noindent{\bf\em The first and the second moment method.}
The above statistical mechanics picture holds the key to understanding why the
previous arguments did not suffice to pin down $\rkNAE$ precisely.
The best previous bounds~(\ref{eqprevious}) were obtained by applying the first/second moment method to the number
$Z(\PHI)$ of solutions, or a closely related random variable.
%Let us recap these techniques briefly.

With respect to the upper bound,
if for some density $r$ the first moment $\Erw\brk{Z(\PHI)}$ tends to $0$ as $n$ gets large,
then $Z(\PHI)=0$ \whp\ by Markov's inequality. Thus, $\rkNAE\leq r$.
Indeed, it is not difficult to verify that $\Erw\brk{Z(\PHI)}=o(1)$ % (using linearity of expecation and independence of the clauses), 
%and to verify that %$\Erw\brk{Z(\PHI)}\ra0$
%it tends to $0$ 
for $r=\rfirst$~\cite{nae}.
This gives the upper bound in~(\ref{eqprevious}).
%This simple argument gives the best previous upper bound~\cite{nae}.
%

The purpose of the second moment method is to bound $\rkNAE$ from below.
%The basic idea behind the second moment method for obtaining a lower bound on $\rkNAE$ is as follows.
The general approach is this:
suppose we can define a random variable $Y=Y(\PHI)\geq0$ %such that $\Erw\brk Y>0$ and 
such that $Y>0$ only if $\PHI$ has a NAE-solution.
Moreover, assume that for some density $r$, the second moment $\Erw[Y^2]$ satisfies
	\begin{equation}\label{eqsmm}
	\textstyle
	\Erw[Y^2]\leq C\cdot\Erw\brk{Y}^2
	\end{equation}
with $C=C(k)\geq1$ dependent on $k$ but not on $n$.
Then the {\em Paley-Zygmund inequality}
	%\begin{equation}\label{eqPaleyZygmund}
	$\pr\brk{Y>0}\geq\Erw\brk{Y}^2/\Erw[Y^2]$
	%\end{equation}
implies that
	\begin{equation}\label{eqboundedAway}
	\pr\brk{\PHI\mbox{ has a NAE-solution}}\geq\pr\brk{Y>0}\geq\Erw[Y^2]/\Erw\brk{Y}^2\geq1/C>0.
	\end{equation}
Because the $k$-NAESAT threshold is sharp, and as %$1/C>0$ with 
$C$ is independent of $n$,
(\ref{eqboundedAway}) implies that $\rkNAE\geq r$.
%Hence, to  lower bound $\rkNAE$ we ``just'' need to construct a random variable $Y$ whose
%second moment $\Erw\brk{Y^2}$ satisfies~(\ref{eqsmm}).
%This is the \emph{second moment method}.

%The lower bound was~(\ref{eqprevious}) is obtained by applying the second moment to the number $Z(\PHI)$ of solutions.

The obvious choice of random variable is the number $Z(\PHI)$ of solutions.
Since $Z(\PHI)^2$ is just the number of \emph{pairs} of NAE-solutions,
	%(as is $Z'(\PHI)^2$ for almost all $\PHI$),
the second moment can be written as
	\begin{equation}\label{eqdecompsSecond}
	\textstyle\Erw[Z(\PHI)^2]=\sum_{\sigma,\tau\in\cbc{0,1}^n}\pr\brk{\mbox{both $\sigma,\tau$ are NAE-solutions}}.
	\end{equation}
Indeed, Achlioptas and Moore~\cite{nae} proved that~(\ref{eqsmm}) is satisfied for $Y=Z(\PHI)$
if $r\leq 2^{k-1}\ln2-\bc{1+\ln2}/{2}$.
 % very closely related to $Z(\PHI)$.
%More precisely, Achlioptas and Moore~\cite{nae} proved that~(\ref{eqsmm}) is satisfied for $Y=Z(\PHI)$
%for densities $r\leq r_{second}=2^{k-1}\ln2-\bc{1+\ln2}/{2}+o_k(1)$, while (\ref{eqsmm})
%violated for $r>r_{second}$.
 %, an additive $\frac12$ below the bound claimed in~(\ref{eqprevious}).
%They also observed that for $r>r_{second}$, $\Erw\brk{Z(\PHI)^2}\geq\exp(\Omega(n))\cdot\Erw\brk{Z(\PHI)}^2$.
%In other words, for $r>r_{second}$ the second moment bound~(\ref{eqsmm}) fails dramatically for $Y=Z(\PHI)$.
%Investigating into the reason for this blow-up of the second moment,
%More precisely,  % improved this to obtain the best previous bound~(\ref{eqprevious})
%by considering a slightly modified random variable.
%Namely,
%	they let $Y(\PHI)=Z(\PHI)\cdot\vecone_{\PHI\in\cA}$, where $\cA$ is a certain event such that $\PHI\in\cA$ \whp\
%In other words, $Y$ is equal to $Z(\PHI)$ for almost all formulas, but a small fraction
%of ``bad'' formulas (that would blow up the second moment) are excluded.
%In effect, the second moment analysis boils down to studying the correlations amongst
%pairs of solutions.
Improving upon~\cite{nae},
% the bound that can be obtained by choosing $Y=Z(\PHI)$, 
Coja-Oghlan and Zdeborov\'a~\cite{Lenka} %improved this result to obtain
obtained the best previous lower bound~(\ref{eqprevious}) %on $\rkNAE$ by applying the second moment method
by considering a slightly modified random variable $Z'(\PHI)$.
%(This random variable is very closely related to $Z$ and indeed $Z'(\PHI)=Z(\PHI)$ for almost all formulas $\PHI$.)
Namely,
	$Z'(\PHI)=Z(\PHI)\cdot\vecone_{\PHI\in\cA}$, where $\cA$ is a certain event such that $\PHI\in\cA$ \whp\
In other words, $Z'(\PHI)$ is equal to $Z(\PHI)$ for almost all formulas, but a small fraction
of ``bad'' formulas (that would blow up the second moment)
are excluded.
Still,  $Z'(\PHI)$ admits a similar decomposition as~(\ref{eqdecompsSecond})
	(one just has to condition on $\cA$).

%The second moment of $Z'(\PHI)$ admits a similar decomposition.

As~(\ref{eqdecompsSecond}) shows, the second moment analysis of either $Z(\PHI)$ or $Z'(\PHI)$
boils down to studying the correlations amongst \emph{pairs} of solutions.
In fact, it was observed in~\cite{nae,Lenka} that a \emph{necessary} condition for the success of this approach
is that two independently and uniformly chosen $\vec\sigma,\vec\tau\in\cS(\PHI)$ satisfy $\dist(\vec\sigma,\vec\tau)\sim n/2$ \whp\
But according to the statistical mechanics picture, this decorrelation condition is violated for $r>\rcond$
%as pairwise correlations become excessive in the condensation phase 
due to the presence of large clusters.
Therefore, it is not surprising that the best previous lower bound~(\ref{eqprevious})
on $\rkNAE$ coincides with the (conjectured) condensation threshold $\rcond$.
% at which,
%according to the physics arguments, the condensation phase commences.
Indeed, it was verified in~\cite{Lenka} that a certain ``weak'' form of condensation
sets in at $r\sim\rcond$.

The statistical mechanics prescription to overcome these correlations is to work
with the Survey Propagation distribution (first select a
cluster uniformly, then choose a random solution from that cluster) rather than the uniform distribution over $\cS(\PHI)$.
This is precisely the key idea behind our new \emph{SP-inspired second moment argument}.
Roughly speaking, we are going to develop a way to apply the second moment method to the number $N(\PHI)$ of \emph{clusters},
rather than the number of solutions.
%
%Our main technical contribution is a \emph{generalized second moment argument} that succeeds in the condensation phase.
%The basic idea is to apply the second moment method to assignments $\vec\tau$ drawn from the {\em Survey Propagation distribution} (first select a
%cluster uniformly, then choose a random solution from that cluster) rather than from the uniform distribution.
More precisely, we introduce a parameter $\beta$ that allows us to work with clusters of a prescribed size.
A specific choice of $\beta$ (namely, $\beta=1/2$) corresponds to the SP distribution
and thus to working with $Y(\PHI)=N(\PHI)$.

This new technique allows us to obtain various further results.
For instance, %while (\ref{eqannealedOff}) merely shows that $Z(\PHI)$ is exponentially smaller than its expectation \whp,
%our new approach allows us to
we can pin down the typical values of both  $Z(\PHI)$ and $N(\PHI)$ throughout the condensation phase (details omitted).
Furthermore, our proof entails the following result that %following asymptotic decorrelation property for solutions drawn from the SP distribution,
confirms the physics conjecture that pairs of solutions drawn from the SP distribution decorrelate throughout the condensation phase.

\begin{corollary}\label{Cor_SP}
Suppose that $r_{cond}\leq r\leq 2^{k-1}\ln2-\bc{\frac{\ln2}2+\frac14}-\eps_k$.
Let $\vec\sigma',\vec\tau'$ be drawn independently from the SP distribution.
Then $\dist(\vec\sigma',\vec\tau')=(\frac12+o_k(1))n$ \whp
\end{corollary}

\section{Related work}\label{Sec_related}

\noindent{\bf\em Rigorous work.}
The $k$-NAESAT problem is well-known to be NP-complete in the worst case for any $k\geq3$.
In fact, the NP-complete problem of $2$-coloring a $k$-uniform hypergraph (with $k\geq3$) simply is the
special case of $k$-NAESAT without negations.
The results in~\cite{Lenka} are actually phrased in terms of hypergraph $2$-coloring but 
carry over to $k$-NAESAT directly.

%Say that $k$-NAESAT is NP-hard in the worst case.

%As we saw in the previous section, \Thm~\ref{Thm_NAE} improves prior
%bounds on $\rkNAE$ from~\cite{nae,Lenka}.
%
%Achlioptas and Moore~\cite{nae} obtained the upper bound on $\rkNAE$ in~(\ref{eqprevious}) via the first moment method.
%In addition, they proved that $\rkNAE\geq 2^{k-1}\ln2-\frac{1+\ln2}2+o_k(1)$ by applying the second moment method to $Z(\PHI)$.
%Their work provided the prototype for the subsequent applications of the second moment method to
%		random $k$-SAT~\cite{yuval}, random graph $k$-coloring~\cite{AchNaor} and several other problems.
%The ``vanilla' second moment argument from~\cite{nae} was recently improved by an additive $\frac12$ by Coja-Oghlan and Zdeborova~\cite{Lenka}
%in order to obtain the best previous lower bound $\rkNAE\geq 2^{k-1}\ln2-\ln2+o_k(1)$ stated in~(\ref{eqprevious}).
%
%In addition, in~\cite{Lenka} we proved that \emph{at} density $r_{cond}=2^{k-1}\ln2-\ln2+o_k(1)$ a weak form of condensation occurs.
%More precisely, for $r=r_{cond}$ the solution space $\cS(\PHI)$ is dominated by a sub-exponential number $\exp(o(n))$ of clusters \whp\
%This is a first step but it falls short of proving the condensation phenomenon as hypothesized by physics arguments, according to which
%a \emph{bounded} number of clusters dominate $\cS(\PHI)$ at \emph{all} densities $r_{cond}<r<\rkNAE$.
%Proving this stronger statement remains an interesting open problem.
%
%\aco{Say something about techniques.}
The main contribution of \Thm~\ref{Thm_NAE} is the improved \emph{lower} bound.
In fact, the upper bound in~(\ref{eqnew}) can be obtained in several different ways.
Achlioptas and Moore~\cite{nae} state without proof that the (quite intricate) enhanced first moment argument from~\cite{DB,KKKS} can be used
to show that $\rkNAE\leq2^{k-1}\ln2-(\frac{\ln2}2+\frac14)+o_k(1)$.
This is indeed plausible as, in terms of the statistical mechanics intuition
	(which was unknown to the authors of~\cite{DB,KKKS}) 
this argument amounts to computing the first moment of the number of \emph{clusters}.
Alternatively, generalizing work of Franz and Leone~\cite{FranzLeone}, Panchenko and Talagrand~\cite{Panchenko} proved that the variational problem
that results from the SP formalism yields a rigorous upper bound on $\rkNAE$, which is conjectured to be tight for any $k\geq3$.
The variational problem can be solved asymptotically in the large-$k$ limit
	(unpublished), yielding the upper bound stated in \Thm~\ref{Thm_NAE}.
In this paper we obtain the upper bound by a relatively simple third argument that has a neat combinatorial interpretation.
%as to why the $k$-NAESAT threshold is $\rkNAE=2^{k-1}\ln2-(\frac{\ln2}2+\frac14)+o_k(1)$.

The proofs of the lower bounds in~\cite{nae,Lenka} and in the present paper are non-constructive in the sense
that they do not entail an efficient algorithm for finding a NAE-solution \whp\
The best current algorithm for random $k$-NAESAT is known to succeed for $r\leq O_k(2^k/k)$,
a factor of $\Omega_k(k)$ below $\rkNAE$~\cite{AKKT}.

From a statistical mechanics point of view, many random CSPs
% problems such as random $k$-SAT or random graph $k$-coloring 
are similar to random $k$-NAESAT.
In particular, the physics methods suggest the existence of a condensation phase in most random CSPs
	(e.g.,  random $k$-SAT/graph $k$-coloring).
%We believe that with (significant) additional technical work, our methods can be extended to these problems.
While~\cite{nae} provided the prototype for the second moment arguments in these and other problems,
the technical details in random graph $k$-coloring~\cite{AchNaor} or random $k$-SAT~\cite{yuval}
are quite a bit more intricate than in random $k$-NAESAT.

For instance, random $k$-NAESAT is simpler than random $k$-SAT because
for any NAE-solution $\sigma$ the inverse $\bar\sigma:x\mapsto1-\sigma(x)$ is a NAE-solution as well.
This symmetry of the solution space under inversion simplifies the second moment calculations significantly.
To cope with the absence of symmetry in random $k$-SAT, Achlioptas and Peres~\cite{yuval}
weighted satisfying assignments  cleverly in order to recover the beneficial analytic properties that symmetry induces.
%in random graph $k$-coloring the second moment argument depends on the optimization of a certain transcendental
%function on the $k$-dimensional Birkhoff polytope, a difficult analytical task.
%Furthermore, in random $k$-SAT the second moment argument is complicated by a certain lack of symmetry.
%Nonetheless, we believe that with (significant) additional technical work, our techniques can be generalized
%	to random $k$-SAT and other problems.
%In effect, in both of these problems the gap between the second moment lower bound and the conjectured threshold value is bigger
%than in $k$-NAESAT.
Our new second moment method is quite different from this weighting approach,
since the asymmetry that called for the weighting scheme in~\cite{yuval} is absent in $k$-NAESAT.
%Nevertheless, we believe that with additional technical work, our improved second moment
%argument can be extended to other problems such as random $k$-SAT or random graph $k$-coloring.
%Instead, %a weighting scheme,
%the basic idea behind our approach is to develop a second moment argument
%for the number $N(\PHI)$ of clusters.

None of the (few) random CSPs in which the threshold for the existence of solutions is known precisely has a condensation phase.
The most prominent example is random $k$-XORSAT (random linear equations mod $2$)~\cite{Dubois,PittelSorkin}.
In this case, the algebraic nature of the problem precludes condensation:
	all clusters are simply translations of the kernel.
Similarly, the condensation phase is empty in the uniquely extendible problem from~\cite{Connamacher}.
Also in random $k$-SAT with $k=k(n)>\log_2 n$ (i.e., the clause length grows as a function of $n$), where
the precise threshold has been determined by Frieze and Wormald~\cite{FriezeWormald} via the second moment method,
 condensation does not occur~\cite{ACOFriezekSAT}.
Nor does it in random 2-SAT~\cite{mick,Goerdt}.

%\aco{Say that our work builds on prior work on the geometry of the solution space.}

Parts of our proof require a precise analysis of geometry of the solution
space $\cS(\PHI)$.
This analysis harnesses some of the ideas that were developed in previous work~\cite{Barriers,fede,Lenka,Daude}
	(e.g., arguments for proving the existence of clusters or of ``rigid variables'').
However, we need to go beyond these previous arguments significantly in two respects.
First, we need to generalize them to accommodate the parameter $\beta$ that controls the cluster sizes.
Second, we need rather precise quantitative information about the cluster structures.
%\aco{A part of this analysis is the solution of a tricky occupancy problem, which may be of independent interest (see \dots).}
%%	whereas the results~\cite{Barriers,fede,Daude} merely was 

\smallskip
\noindent{\bf\em Survey Propagation guided decimation.}
%Condensation is a very important notion in the statistical mechanics of ``disordered systems'' such as glasses.
%In 1948, Kauzmann~\cite{} first hypothesized a condensation phase transition in disordered systems based on experiments.
%Since that time, the analysis of mathematical models of condensation has been an important topic.
%In fact, in statistical mechanics random CSPs are studied as so-called 
%
The SP formalism has given rise to an efficient message passing algorithm
called \emph{Survey Propagation guided decimation} (`SPD')~\cite{MPZ}.
Experimentally, SPD seems spectacularly successful at solving, e.g., random $k$-SAT for small values of $k$.
Unfortunately, no quantitative analysis of this algorithm is currently known (not even a non-rigorous one).
The basic idea behind SPD is to approximate the marginals of the SP distribution
	(i.e., the probability that a given variable is `true' in a  solution drawn from the SP distribution)
via a message passing heuristic.
Then a variable $x$ %(e.g., the one with the most biased marginal)
is selected according to some rule and is assigned a value
based on the 
%For instance, one could assign it a 
%by choosing a random value chosen from its 
(approximate) marginal.
The entire procedure is repeated on the ``decimated'' problem instance where $x$ has been eliminated,
until (hopefully) a solution is found.

The decorrelation of random solutions chosen from the SP distribution
%property~(\ref{eqSPfactorizing}),
%	established rigorously by \Cor~\ref{Cor_SP},
is a crucial assumption behind the message passing computation of the SP marginals.
%By combining 
\Cor~\ref{Cor_SP} establishes such a decorrelation property rigorously.
%with the techniques from~\cite{BPdec}
%it might be possible to prove that SPD computes the correct marginals
%on the \emph{initial} formula $\PHI$. % from which no variable has been eliminated yet.
However, in order to actually analyze SPD, one would have to generalize \Cor~\ref{Cor_SP}
to the situation of a ``decimated'' random formula in which a number of variables have already been eliminated by previous steps of the algorithm.
Still,  we believe that the techniques developed in this paper
are a (necessary) first step towards a  rigorous analysis of SPD.

\section{Heavy solutions and the first moment}\label{Sec_first}

{\em In the rest of the paper we sketch the SP-inspired second moment method on which the proof of \Thm~\ref{Thm_NAE} is based.
Aiming for an asymptotic result, we may assume that $k\geq k_0$ for some (large) constant $k_0>3$.
We also assume $r=2^{k-1}\ln2-\rho$ for 
some $\frac12\ln2\leq \rho\leq\ln2$.
%	(recall that $\rkNAE\leq2^{k-1}\ln2-\ln(2)/2$ by~\cite{nae}).
Let $\PHI_i$ denote the $i$th clause of the random formula $\PHI$ so that $\PHI=\PHI_1\wedge\cdots\wedge\PHI_m$.
Furthermore, let $\PHI_{ij}$ signify the $j$th literal of clause $\PHI_i$;
	thus, $\PHI_i=\PHI_{i1}\vee\cdots\vee\PHI_{ik}$.
For a literal $\ell$ we let $\abs \ell$ denote the underlying variable.}

As we discussed earlier, the demise of the ``standard'' second moment method in the
condensation phase is due to the dominance of few large clusters. %, spoiling the decorrellation property~(\ref{eqBPfactorizing}).
The statistical mechanics prescription for circumventing this issue is to work with a non-uniform distribution over solutions
that favors ``small'' clusters.
To implement this strategy, we are going to exhibit a simple parameter that governs the size of the cluster that a solution belongs to.
Formally, we define the \emph{cluster of $\sigma\in\cS(\PHI)$} as
	$$\cC(\sigma)=\cC_{\PHI}(\sigma)=\cbc{\tau\in\cS(\PHI):\dist(\sigma,\tau)\leq0.01n}.$$
This definition is vindicated by the following observation from~\cite{Lenka}, which shows that
any two solutions either have the same cluster or are well-separated.

\begin{proposition}\label{Prop_noMiddleGround}
Suppose that $2^{k-1}\ln2-\ln2\leq r\leq \rkNAE$.
\Whp\ any two $\sigma,\tau\in\cS(\PHI)$ either satisfy $\dist(\sigma,\tau)\leq0.01n$ or 
	 $\dist(\sigma,\tau)\geq(\frac12-2^{-k/3})n$.
\end{proposition}

To proceed, we need to get an idea of the ``shape'' of the clusters $\cC(\sigma)$.
According to the SP formalism, each cluster has a set $\cR(\sigma)$
of $\Omega(n)$ \emph{rigid variables} on which
\emph{all} assignments in $\cC(\sigma)$ coincide, %take the same values on the rigid variables,
while the values of the non-rigid variables vary.
Formally, we have $\tau(x)=\sigma(x)$ for all $x\in\cR(\sigma)$ and all $\tau\in\cC(\sigma)$,
while for each $x\not\in\cR(\sigma)$ there is $\tau\in\cC(\sigma)$ such that $\tau(x)\neq\sigma(x)$.
This implies an immediate bound on the size of $\cC(\sigma)$, namely
	$\abs{\cC(\sigma)}\leq2^{n-|\cR(\sigma)|}.$
%In fact, in~\cite{Barriers} it was proved that, 
%for a certain regime of densities, a uniformly chosen solution $\sigma\in\cS(\PHI)$ lies
%in a cluster with $\Omega(n)$ rigid variables \whp\
%Building upon their arguments, 
Indeed, we are going to prove that %for any $2^{k-1}\ln2-\ln2\leq r\leq \rkNAE$ 
every cluster has
a rigid set of size $\Omega(n)$ \whp, and that   for all clusters \whp\
	\begin{equation}\label{eqRigidEntropy}
	%(1-o_k(1))(n-|\cR(\sigma)|)\ln 2\leq
	\log_2\abs{\cC(\sigma)}=(1-o_k(1))(n-|\cR(\sigma)|).
	\end{equation}
%\aco{Comment on random subcube model?}

With $|\cC(\sigma)|$ controlled by the number  of rigid variables,
it might seem promising to perform first/second moment arguments for
  the number of solutions with a suitably chosen number of rigid variables.
%and to carry out a first/second moment argument for this random variable.
The problem with this is that there is no simple way to tell whether a given variable is rigid:
deciding this is NP-hard in the worst case.
Intuitively, this is because rigidity emerges from the ``global'' interplay of variables and clauses.
%	given a formula and an assignment, there is not even an obvious way to check whether some variable is rigid.
%		\aco{Is this maybe a hard computational problem?}
In effect, 
parametrizing by the number of rigid variables appears technically infeasible.

Instead, we are going to work with a simple ``local'' parameter that turns out to be a good substitute.
Suppose that $x\in\cR(\sigma)$.
Then $x$ must occur in some clause $\PHI_i$   that would be violated if
$x$ was assigned the opposite value $1-\sigma(x)$
	(with all other variables  unchanged).
By the definition of $k$-NAESAT,
this means that the other $k-1$ literals of $\PHI_i$ take the opposite value of the literal whose underlying variable $x$ is.
In this case we say that $x$ \emph{supports} $\PHI_i$ under $\sigma$, and we call $\PHI_i$ a \emph{critical} clause.
Moreover, we call a variable that supports a clause \emph{blocked}, while all other variables are \emph{free}.
%Let $\cF(\sigma)$ be the set of fee variables.
While every rigid variable is blocked, the converse is not  generally true.
Nonetheless, we will see that the number of variables that are blocked but not rigid is small enough
so that we can control the cluster sizes in terms of blocked variables.

As a first step, we are going to estimate the expected number of solutions with a given number
of blocked variables.
%Set
%	$$\lambda=\frac{kr}{2^{k-1}-1}=k\ln2+O_k(k/2^k).$$
%Moreover, 
Let
	$\lambda=\frac{kr}{2^{k-1}-1}=k\ln2+O_k(k/2^k)$ and let
 us say that $\sigma\in\cS(\PHI)$ is \emph{$\beta$-heavy} if exactly $(1-\beta)\exp(-\lambda)n$ variables are free.
Let $\cS_\beta(\PHI)$ be the set of all $\beta$-heavy solutions and let $Z_{\beta}=\abs{\cS_\beta(\PHI)}$ denote their number.

\begin{proposition}\label{Prop_roughFirst} 
For any $\beta\leq1$ we have %\aco{double-check the error term!!! check the $-1$!!!}
	\begin{equation}%\nonumber
	\textstyle
	\Erw\brk{Z_\beta}=\exp\brk{\frac{n}{2^k}\bc{2\rho-\ln(2)-(1-\beta)\ln(1-\beta)-\beta%+f(\beta)
		+O_k(k\cdot2^{-k})}}.
		%,\qquad\mbox{where}\\
%	f(\beta)&=&-\frac{(1-\beta)\ln(1-\beta)+\beta}{\exp(\lambda)}.
		\label{eqBallsBinsFunction}
	\end{equation}
In particular, $Z_{\beta}=0$ for all $\beta<-3/2$ \whp
\end{proposition}
\begin{proof}
The computation of $\Erw\brk{Z_\beta}$ is instructive because it hinges upon the solution
of an occupancy problem that will play an important role in the second moment computation.
Let $\vecone$ denote the assignment that sets all variables to true.
By the linearity of expectation and by symmetry, we have
	\begin{eqnarray*}
	\Erw\brk{Z_\beta}&=&\sum_{\sigma\in\cbc{0,1}^n}\pr\brk{\sigma\mbox{ is a $\beta$-heavy solution}}
		=2^n\cdot\pr\brk{\vecone\mbox{ is a $\beta$-heavy solution}}\\
		&=&2^n\cdot\pr\brk{\vecone\mbox{ is $\beta$-heavy}|\vecone\mbox{ is a solution}}\cdot\pr\brk{\vecone\mbox{ is a solution}}.
	\end{eqnarray*}
Clearly, $\vecone$ is a solution iff each clause of $\PHI$ contains both a positive and a negative literal.
A random clause has this property with probability $1-2^{1-k}$.
Since the $m\sim rn$ clauses are chosen independently, we get
	$$\textstyle2^n\cdot\pr\brk{\vecone\mbox{ is a solution}}=2^n(1-2^{1-k})^m=\exp\brk{\frac n{2^k}\bc{2\rho-\ln2+O_k(2^{-k})}}.$$

Working out the conditional probability that $\vecone$ is $\beta$-heavy is not so straightforward.
Whether $\vecone$ is $\beta$-heavy depends only on the critical clauses of $\PHI$.
Let $X$ be their number.
Given that $\vecone$ is a solution, each clause $\PHI_i$ is critical with probability $k/(2^{k-1}-1)$ independently
	(as there are $2k$ ways to choose the literal signs to obtain a critical clause).
Hence, $X$ has a binomial distribution $\Bin(m,k/(2^{k-1}-1))$ with mean
	$$\Erw\brk{X|\vecone\in\cS(H)}=\frac{km}{2^{k-1}-1}=\lambda n.$$
%Let $\lambda=\frac{kr}{2^{k-1}-1}=k\ln2+O_k(k/2^k)$.
Since the supporting variable of each critical clause is uniformly distributed,
given $\vecone\in\cS(H)$  the \emph{expected} number of clauses that each variable supports equals $\lambda$.
Thinking of the variables as bins and of the critical clauses as balls, standard results
on the occupancy problem show that the number of free variables is $(1+o(1))\exp(-\lambda)n$ \whp\
Thus, $\Erw\brk{Z_\beta}$ is maximized for $\beta=0$.

By contrast, values $\beta\neq0$ correspond to \emph{atypical} outcomes of the occupancy problem.
Values $\beta<0$ require an excess number of ``empty bins'',
while $\beta>0$ means that fewer bins than expected are empty.
%These atypical outcomes result from an interplay of two rare events:
%	deviations of $X$ from its expectation, and atypical distributions of ``balls'' over ``bins''.
To determine the precise (exponentially small) probability of getting $(1-\beta)\exp(-\lambda)n$
empty bins,
we need to balance large deviations of $X$ against 
the probability that exactly $(1-\beta)\exp(-\lambda)n$
bins remain empty for a given value of $X$.
% when $\mu$ balls are thrown randomly into $n$ bins, and the probability that $X=\mu$.
The result of this combined large deviations analysis is the expression~(\ref{eqBallsBinsFunction}).
The analysis also shows that $\Erw\brk{Z_{\beta}}=\exp(-\Omega(n))$ for $\beta<-3/2$,
	whence $Z_\beta=0$ \whp\ for $\beta<-3/2$.
%
%
%%given $\vecone\in\cS(H)$ the 
%
% first summand is simply the expected number of solutions.
%Let us condition on some assignment being a solution.
%Then we can verify (by Chernoff bounds) that the number of critical clauses is
%	pretty much what it should be, namely $\lambda n$.
%This allows us to compute the probability that the balls-and-bins thing works out, which is given by $f(\beta)$.
\qed\end{proof}
As a next step, we need to estimate the cluster size of a $\beta$-heavy solution.
%\aco{The proof relies on the arguments used previously to show rigidity.}

\begin{proposition}\label{Prop_clusterSize}
\Whp\ for all $-3/2\leq\beta\leq1$ all $\beta$-heavy $\sigma\in\cS(\PHI)$ satisfy
	\begin{equation}\label{eqclusterSize}
	\log_2\abs{\cC(\sigma)}=\frac n{2^k}\brk{1-\beta+o_k(1)
		}.%O_k(2^{-k})}.
	\end{equation}
%
%	\begin{eqnarray*}
%	\frac{\ln\Erw\brk{Z_\beta}}n&=&\frac{\ln2-2c}{2^{k}}+f(\beta)+O_k(2^{-3k/2}),\qquad\mbox{where}\\
%	f(\beta)&=&\aco{\mbox{whatever}}.
%	\end{eqnarray*}
%Furthermore, \whp\ we have $Z_{\beta}=0$ for all $\beta<-1$.
\end{proposition}
\begin{proof}
The crucial thing to show is that all but a very few blocked variables are rigid.
The proof of this builds upon arguments developed in~\cite{Barriers} to establish rigidity.
%The aim is to show that all but a tiny fraction of blocked variables in a $\beta$-heavy assignments $\sigma$ are rigid \whp\
Suppose that $x$ is blocked in  $\sigma\in\cS_\beta(\PHI)$, i.e., $x$ supports some clause, say $\PHI_1$.
In any solution $\tau$ with $\tau(x)\neq\sigma(x)$ there must be another variable $x'$ that occurs in $\PHI_1$ such that $\tau(x')\neq\sigma(x')$.
Given that $x$ supports $\PHI_1$, the other $k-1$ variables of $\PHI_1$ are uniformly distributed.
Since $\sigma$ has no more than $(1-\beta)\exp(-\lambda)n= (1-\beta+o_k(1))2^{-k}n$ free variables,
the probability that $x'$ is free is bounded by $(1-\beta+o_k(1))(k-1)/2^{k}$.
%Assume that it does not.
%
%Since we are assuming that $x'$ supports at least one other clause, say $\PHI_2$, this clause must contain a further variable $x''$ such that
%$\tau(x'')\neq\sigma(x'')$.
In fact, since the \emph{expected} number of clauses that each variable supports is $\lambda=(1+o_k(1))k\ln2$, it is quite likely
that $x'$ supports several clauses and that therefore ``flipping'' $x'$ necessitates \emph{several} further flips.
Continuing this argument, we see that the number of flips follows a branching process with (initial) successor rate $\lambda$.
A detailed analysis shows that for all but $O_k(k4^{-k})n$ blocked initial variables $x$ this process will lead to an avalanche of more than $0.01 n$ flips,
whence $\tau\not\in\cC(\sigma)$. % by \Prop~\ref{Prop_roughFirst}.
This shows that all but $o_k(2^{-k})n$ blocked variables are rigid.
%\aco{Say something about the ``all'' bit!!!}
\qed\end{proof}

%\noindent
%\emph{Proof of the upper bound.}

We are ready to prove that $\rkNAE\leq2^{k-1}\ln2-(\frac{\ln2}2+\frac14)+o_k(1)$,
which  is (almost) the upper bound promised in \Thm~\ref{Thm_NAE}.
	(Some additional technical work is needed to replace the $o_k(1)$ by an error term
		that decays exponentially.)
%To prove the upper bound on $\rkNAE$ stated in \Thm~\ref{Thm_NAE},
Let
	$N_{\beta}=\abs{\cbc{\cC(\sigma):\sigma\in\cS(\PHI)\mbox{ is $\beta$-heavy}}}$
be the number of \emph{clusters} centered around $\beta$-heavy solutions.
%%Suppose that $\beta_*$ be the least value of $\beta$ such that $Z(\beta)>0$.
%Moreover, let us \emph{assume} that $Z_{\beta'}=0$ for all $\beta'\leq\beta$.
%Since the function on the r.h.s.\ of~(\ref{eqBallsBinsFunction}) is strictly decreasing in $\beta$,
%given that $Z_{\beta'}=0$ for all $\beta'\leq\beta$ the \emph{total} number
%of solutions satisfies \whp
%	\begin{equation}\label{eqFirstWeak1}
%	\ln Z\sim\ln Z_\beta\leq(1+o(1))\ln\Erw\brk{Z_\beta}=\frac{n}{2^k}\brk{2c-\ln(2)-(1-\beta)\ln(1-\beta)-\beta%+f(\beta)
%		+O_k(2^{-k})}
%	\end{equation}
%The second step follows from Markov's inequality, and the last step is due to \Prop~\ref{Prop_roughFirst}.
%Furthermore, by \Prop~\ref{Prop_clusterSize} each of the $N_{\beta}$ clusters centered around some $\sigma\in\cS_\beta(\PHI)$
%has size $\log_2\abs{\cC(\sigma)}=\frac n{2^k}\brk{1-\beta+o_k(1)}$ \whp\
%In particular, if $N_\beta>0$, then
%	\begin{equation}\label{eqFirstWeak2}
%	\ln Z\geq \frac n{2^k}\brk{1-\beta+o_k(1)}\ln2.
%	\end{equation}
%Combining~(\ref{eqFirstWeak1}) and~(\ref{eqFirstWeak2}), we see that the event $N_\beta>0$ can have a non-vanishing probability
%\emph{only if}
%	\begin{equation}
%	(1-\beta)\ln2\leq2c-\ln(2)-(1-\beta)\ln(1-\beta)-\beta+\delta_k,
%	\end{equation}
%with $\delta_k=o_k(1)$.
%Let $\sigma\in\cS(\PHI)$ be $\sigma$-heavy.
By \Prop~\ref{Prop_clusterSize}, 
each such cluster has size 
	$\abs{\cC(\sigma)}=2^{n(1-\beta+o_k(1))/2^k}$ \whp\
%the size of $\cC(\sigma)$ is as in~(\ref{eqclusterSize})) \whp\
Hence, once more by \Prop~\ref{Prop_clusterSize}, any solution
$\tau\in\cC(\sigma)$ is $\beta'$-heavy for some $\beta'$ satisfying $\abs{\beta'-\beta}\leq\delta_k=o_k(1)$ \whp\
Letting $Z_{\beta}^*$ be the total number of $\beta'$-heavy solutions with $\abs{\beta'-\beta}\leq\delta_k$,
we conclude that
	\begin{equation}\label{eqFirstWeak3}
	N_\beta\cdot 2^{n(1-\beta+o_k(1))/2^k}\leq Z_{\beta}^*\qquad\mbox\whp
	\end{equation}
Clearly, $Z_{\beta}^*\leq\Erw[Z_{\beta}^*]\cdot\exp(o(n))$ \whp\ by Markov's inequality.
Furthermore, as the total number of free variables in each cluster is an integer between $0$ and $n$,
we have $\Erw[Z_{\beta}^*]\leq (n+1)\cdot\max_{\beta'}\Erw[Z_{\beta'}]$.
Combining these inequalities with the estimate of $\Erw[Z_{\beta'}]$ from
%estimating $\ln\Erw[Z_{\beta}^*]$ via 
\Prop~\ref{Prop_roughFirst}, we find
	\begin{equation}\label{eqFirstWeak4}
	Z_{\beta}^*\leq\exp\brk{o(n)}\Erw[Z_{\beta}^*]\leq
		\exp\bc{\frac{n}{2^k}\brk{2\rho-\ln(2)-(1-\beta)\ln(1-\beta)-\beta+o_k(1)}}
		\qquad\mbox\whp
	\end{equation}
Combining~(\ref{eqFirstWeak3}) and~(\ref{eqFirstWeak4}), we obtain

\begin{fact}\label{Fact_necessary}
%see that if $N_\beta\geq1$, then 
\Whp\ we have $N_\beta\leq\exp\brk{\eta(\beta)\cdot n/2^k}$ for all $\beta$, with
	\begin{equation}\label{eqFirstWeak5}
	\eta(\beta)=2\rho-\ln(2)-(1-\beta)\ln(2-2\beta)-\beta+o_k(1).
	%N_\beta\leq\exp\brk{2c-\ln(2)-(1-\beta)\ln(2-2\beta)-\beta+o_k(1)}\quad\mbox{ for all }\beta.
	\end{equation}
\end{fact}

Finally, it is a mere exercise in calculus to verify that
at density $r^*=2^{k-1}\ln2-(\frac{\ln2}2+\frac14)+o_k(1)$
 the exponent $\eta(\beta)$
	%$2c-\ln(2)-(1-\beta)\ln(1-\beta)-\beta-(1-\beta)\ln2+o_k(1)$
%(\ref{eqFirstWeak5}) 
is negative 
for \emph{all} $\beta$.
Therefore, Fact~\ref{Fact_necessary} implies that $r^*$ is an upper bound on $\rkNAE$.

\begin{remark}
The exponent $\eta(\beta)$ attains its maximum at $\beta=\frac12+o_k(1)$.
Together with our second moment bound below,
this implies that for $\beta=\frac12+o_k(1)$ we have $N(\PHI)=\exp(o_k(1)n)\cdot N_\beta(\PHI)$ \whp,
i.e., setting $\beta=\frac12+o_k(1)$ corresponds to the uniform distribution over clusters
and thus to the SP distribution.
\end{remark}

%	(recall that $Z_\beta=0$ \whp\ for $\beta<-1$ by \Prop~\ref{Prop_roughFirst}).

%(\ref{eqFirstWeak3}) implies that
%	\begin{equation}\label{eqFirstWeak3}
%	N_\beta\cdot 2^{n(1-\beta+o_k(1))/2^k}\leq Z_{\beta}^*.
%	\end{equation}
%
%
%Hence, applying \Prop~\ref{Prop_clusterSize} a second time, we see that there is $\eps=O(2^{-k})$ such that
%\whp\ all solutions $\tau\in\cC(\sigma)$ are $\beta'$-heavy with $\abs{\beta'-\beta}\leq\eps$
%	\aco{error terms?!}.
%Hence, letting $Z_{\beta,\eps}$ be the total number of $\beta'$-heavy solutions with $\abs{\beta'-\beta}\leq\eps$,
%we conclude that \whp\
%	$$\frac1n\ln N_{\beta}+(1-\beta)\exp(-\lambda)\ln2+O_k(4^{-k})\leq\frac1n\ln\Erw\brk{Z_{\beta,\eps}}.$$
%Bounding the expectation of $Z_{\beta,\eps}$ via \Prop~\ref{Prop_roughFirst}, we see that \whp
%	$$\frac1n\ln N_{\beta}\leq\frac{\ln2-2c}{2^{k}}+f(\beta)-(1-\beta)\exp(-\lambda)\ln2+O_k(2^{-3k/2}).$$
%Finally, using basic calculus it is easily verified that the right hand side is
%negative for all $-1\leq \beta\leq 1$ for $c<\frac{\ln2}2+\frac14+O_k(2^{-k})$, whence $N_{\beta}=0$ \whp\ in this case
%for all $-1\leq\beta\leq1$.
%This means that \whp\ the \emph{total} number of clusters is zero, i.e., $\PHI$ does not have a NAE-solution.
%\qed
%%Since \whp\ $\PHI$ does not have a $\beta$-heavy assignment with $\beta<-1$

\section{The second moment}\label{Sec_second}

%\smallskip
\noindent{\bf\em A first attempt.}
%Given the upper bound from \Sec~\ref{Sec_first}, % on heavy solutions and their cluster sizes,
% the obvious choice for
%such a random variable $Y$ seems to be the
The obvious approach to proving a matching lower bound on $\rkNAE$ %that matches the upper bound from \Sec~\ref{Sec_first}
 seems to be a second moment argument
for the number $Z_\beta$ of $\beta$-heavy solutions, for some suitable $\beta$.
There is a subtle issue with this, but exploring it will put us on the right track.

We already computed $\Erw\brk{Z_\beta}$ in \Prop~\ref{Prop_roughFirst}.
%Thus, the remaining task is to verify the second moment bound~(\ref{eqsmm}).
As $\Erw[Z_\beta^2]$ is the expected number of \emph{pairs} of $\beta$-heavy solutions,
the symmetry properties of the random formula $\PHI$ imply that
	$$
	\Erw[Z_\beta^2]=\Erw\brk{Z_\beta}\cdot\Erw\brk{Z_\beta|\sigma\in\cS_\beta(\PHI)}
		\qquad\mbox{for any fixed }\sigma\in\cbc{0,1}^n.
	$$
%where $\vecone$ denotes the all-true assignment.
Thus, the second moment condition~(\ref{eqsmm}) that we would like to establish for $Y=Z_\beta$ becomes
	\begin{equation}\label{eqsmmcond}
	\Erw\brk{Z_\beta|\sigma\in\cS_\beta(\PHI)}\leq C\cdot\Erw\brk{Z_\beta}.
	\end{equation}

What value of $\beta$ should we go for?
By Fact~\ref{Fact_necessary} a necessary condition
for the existence of $\beta$-heavy solutions is that the exponent $\eta(\beta)$ from~(\ref{eqFirstWeak5}) is positive.
Let us call $\beta$ \emph{feasible} for a density $r$ if it is.
An elementary calculation shows that for $r>\rcond=2^{k-1}\ln2-\ln2+o_k(1)$,
any feasible $\beta$ is strictly positive.

However, (\ref{eqsmmcond}) turns out to be false for \emph{any} $\beta>0$, for any density $r>0$.
To understand why, let us define the \emph{degree} $d_x$ of a variable $x\in V$ as the number of times that $x$
occurs in the formula $\PHI$.
Let $\vec d=(d_x)_{x\in V}$ be the degree sequence of $\PHI$.
It is well known that in the ``plain'' random formula $\PHI$ (without conditioning on $\sigma\in\cS_\beta(\PHI)$),
the degree of each variable is asymptotically Poisson with mean $km/n$. On the other hand, if we condition on $\sigma\in\cS_\beta(\PHI)$ for some $\beta>0$, then 
the degrees are \emph{not} asymptotically Poisson anymore.
Indeed, the degree $d_x$ is the sum of the number $s_x$ of clauses that $x$ supports,
and the number $d_x'$ of times that $x$ appears otherwise.
While $d_x'$ is asymptotically Poisson with mean $<km/n$ as the non-critical clauses do not affect the number
of blocked variables at all, $s_x$ is not.
More precisely, we saw in the proof of \Prop~\ref{Prop_roughFirst} that
for $\beta>0$, $s_x$ is the number of ``balls'' that $x$ receives in an \emph{atypical} outcome of the occupancy problem.
The precise distribution of $s_x$ is quite non-trivial,
but it is not difficult to verify that $s_x$ does \emph{not} have a Poisson distribution.
Fleshing this observation out leads to the sobering

\begin{lemma}\label{Lemma_doesNotWork}
For any $\beta>0$ and any $r>0$ we have
%	$Z_\beta\leq\exp\brk{-\Omega(n)}\Erw\brk{Z_\beta}$ \whp
	$\Erw[Z_\beta|\sigma\in\cS_\beta(\PHI)]\geq\exp(\Omega(n))\cdot\Erw\brk{Z_\beta}$.
\end{lemma}
In summary, conditioning on $\sigma\in\cS_\beta(\PHI)$ with $\beta>0$ imposes a skewed
degree distribution that in turn boosts the expected number of $\beta$-heavy solutions beyond the unconditional
expectation.

\smallskip
\noindent{\bf\em Making things work.}
%The first issue (that the size of the cluster of a $\beta$-heavy solution lacks concentration) is fairly easy to fix.
%To make the second moment method succeed, we need to cope with the two issues highlighted above:
%	the lack of concentration of the cluster sizes, and the fluctuations of the degree distribution.
%Among the two issues we just highlighted, the fluctuations of the degree sequence are technically much more difficult to deal with.
We tackle the issue of degree fluctuations by separating the choice of the degree sequence from the choice of the actual formula.
More precisely, for a sequence $\vec d=(d_x)_{x\in V}$ of non-negative integers such that $\sum_{x\in V}d_x=km$
we let $\PHI_{\vec d}$ denote a $k$-CNF with degree sequence $\vec d$ chosen uniformly at random amongst all such formulas.
Fixing a ``typical'' degree sequence $\vec d$, 
we are going to perform a second moment argument for  $\PHI_{\vec d}$,
%with a ``reasonable'' fixed degree sequence $\vec d$, 
thereby preventing fluctuations of the degrees.

How do we define ``typical''?
Ideally,  we would like $\vec d$ to enjoy all the properties that the degree sequence of
the (unconditioned) random formula $\PHI$ is likely to have.
Formally, we let $\vec D=\vec D_k(n,m)$ be the distribution of the degree sequence of $\PHI$.
What we are going to show is that our second moment argument succeeds for a random degree
sequence chosen from the distribution $\vec D$ \whp\

%It is well known that under $\vec D$ %the degree sequence $\vec d$ of an (unconditioned) $\PHI$ can be described as follows:
%	the individual degrees $d_x$ are Poisson variables with mean $km/n$, subject only to the condition that
%		$\sum_{x\in V}d_x=km$.

%The other issue (that the cluster sizes lack concentration) is relatively easy to cope with.
%In fact, we can basically ``define this problem away''.

\begin{definition}\label{XDef_good}
A $\beta$-heavy solution $\sigma\in\cS(\PHId)$ is \emph{good} if %both~(\ref{eqnoMiddleGround}) and~(\ref{eqlikelyClusterSize}) are satisfied.
the following conditions are satisfied.
%\aco{phrasing is a bit sloppy.}
\begin{enumerate}
\item[$\bullet$] We have $\abs{\cC(\sigma)}\leq\Erw\brk{Z_\beta(\PHId)}$.
\item[$\bullet$] There does not exist $\tau\in\cS(\PHId)$ with $0.01n\leq\dist(\sigma,\tau)\leq(\frac12-2^{-k/3})n$.
\item[$\bullet$] No variable supports more than $3k$ clauses under $\sigma$.
%\item $\sigma$ is $\beta$-heavy and the total number of critical clauses is equal to $\lambda n$.
%\item No variable supports more than $k$ clauses.
%\item We have
%		$$\frac1n\ln\abs{\cbc{\tau\in\cS(\PHI_{\vec d}):\dist(\sigma,\tau)/n\leq\frac12-2^{-k/3}}}
%				\leq(1-\beta)\exp(-\lambda)\ln2+O_k(4^{-k}).$$
\end{enumerate}
\end{definition}
The first two items mirror our analysis of the solution space from \Sec~\ref{Sec_first}.
The third one turns out to be useful for a purely technical reason.

Let $\cS_{g,\beta}(\PHId)$ be the set of good $\beta$-heavy solutions and set
$Z_{g,\beta}(\PHId)=\abs{\cS_{g,\beta}(\PHId)}$.
We perform a second moment argument for $Z_{g,\beta}(\PHI_{\vec d})$,
with $\vec d$ chosen randomly from the distribution $\vec D$.
The result  is

\begin{proposition}\label{Prop_smm}
Suppose %that $2^{k-1}\ln2-\ln2\leq r\leq2^{k-1}\ln2-(\frac{\ln 2}2+\frac14)$, and assume
that $\beta>0$ is feasible.
%~(\ref{eqnecessary}).
%	\aco{This is a bit sloppy.}
There is $C=C(k)$ such that for a degree sequence $\vec d$ chosen from the distribution $\vec D$ \whp\
%	\begin{eqnarray}\label{eqsmm1}
%	\frac1n\ln\Erw\brk{Z_{g,\beta}(\PHI_{\vec d})}&=&\frac{\ln2-2c}{2^{k}}+f(\beta)+O_k(2^{-3k/2}),\qquad\mbox{ and}\\
	$\Erw\brk{Z_{g,\beta}(\PHI_{\vec d})^2}\leq C\cdot\Erw\brk{Z_{g,\beta}(\PHI_{\vec d})}^2.$
%		\label{eqsmm2}
%	\end{eqnarray}
\end{proposition}

\Prop~\ref{Prop_smm} shows that the second moment method for $Z_{g,\beta}(\PHI_{\vec d})$ succeeds 
for feasible $\beta$.
As we observed in \Sec~\ref{Sec_first}, a feasible $\beta>0$ exists so long as
	$r\leq2^{k-1}\ln2-(\frac{\ln 2}2+\frac14)-O_k(k^4/2^k)$.
Hence, \Prop~\ref{Prop_smm} and the Paley-Zygmund inequality %~(\ref{eqPaleyZygmund})
 show
that $\PHI_{\vec d}$ is NAE-satisfiable for all such $r$ with a non-vanishing probability
	for $\vec d$ chosen randomly from $\vec D$.
Consequently, the same is true of the unconditioned formula $\PHI$
	(because we could generate $\PHI$ by first choosing $\vec d$ from $\vec D$ and then generating $\PHId$).
	%the degree distribution
	%	$\vec d$ is chosen from the ``correct'' distribution $\vec D$).
Since the $k$-NAESAT threshold is sharp~\cite{EhudHunting},
we obtain the lower bound %on $\rkNAE$ promised 
in \Thm~\ref{Thm_NAE}.

\smallskip
\noindent{\bf\em Proving \Prop~\ref{Prop_smm}.}
As a first step, we need to work out $\Erw\brk{Z_{g,\beta}(\PHI_{\vec d})}$.
Suppose $\beta>0$ is feasible. %are as in \Prop~\ref{Prop_smm}.
Recall that $\rho$ is such that $r=2^{k-1}\ln2-\rho$.

\begin{lemma}\label{Lemma_firstMoment}
\Whp\ the degree sequence $\vec d$ chosen from $\vec D$ is such that 
% over the choice of the the degree sequence $\vec d$ we have \aco{error term?!}
	$$\textstyle
		\Erw\brk{Z_{g,\beta}(\PHI_{\vec d})}\sim\Erw\brk{Z_{\beta}(\PHI_{\vec d})}=
			\exp\brk{\frac n{2^k}\bc{2\rho-\ln2
				-(1-\beta)\ln(1-\beta)-\beta%+f(\beta)
		+O_k(k/2^{k})}}.$$
%				\frac{\ln2-2c}{2^{k}}+f(\beta)+O_k(2^{-3k/2}),\mbox{ with $f$ as in~(\ref{eqBallsBinsFunction}).}$$
\end{lemma}
\begin{proof}
Choose and fix a degree sequence $\vec d$.
We need to compute the probability that some $\sigma\in\cbc{0,1}^V$ is a good $\beta$-heavy solution.
%We begin by computing the probability that $\sigma$ is a $\beta$-heavy solution.
By symmetry, we may assume that $\sigma=\vecone$ is the all-true assignment.
Then $\sigma$ is a solution iff every clause contains both a positive and a negative literal.
Since %this property depends only on 
the signs of the literals are chosen for all $m$ clauses independently,
we see that
	\begin{equation}\label{eqLemma_firstMoment1}
	\pr\brk{\sigma\in\cS(\PHI_{\vec d})}=(1-2^{1-k})^m.
	\end{equation}
Given that $\sigma$ is a solution, the number $X$ of critical clauses has distribution $\Bin(m,k/(2^{k-1}-1))$,
because whether a clause is critical depends on its signs only.
As in the proof of \Prop~\ref{Prop_roughFirst}, to determine the probability that $\sigma$ is $\beta$-heavy we need to solve an occupancy problem:
	 $X$ balls representing the critical clauses are tossed randomly into $n$ bins representing the variables.
However, this time the bins have \emph{capacities}: the bin representing $x\in V$ can hold no more than $\min\cbc{3k,d_x}$ balls in total.
Thus, we need to compute the probability that under these constraints, exactly 
$(1-\beta)2^{-k}n$ bins are empty.
This amounts to a rather non-trivial counting problem, but for a random degree sequence $\vec d$
the probability differs from the formula obtained in \Prop~\ref{Prop_roughFirst} only by an error
term that decays exponentially in $k$.
More precisely, 
	\begin{equation}\label{eqLemma_firstMoment2}
	\textstyle\pr\brk{\sigma\in\cS_\beta(\PHI_{\vec d})|\sigma\in\cS(\PHI_{\vec d})}=\exp\bc{-\frac n{2^k}\brk{(1-\beta)\ln(1-\beta)-\beta-O_k(k/2^k)}}.
	\end{equation}

Let us provide some intuition why this is.
The bin capacities are such that \whp\ most bins can hold about $kr=k2^{k-1}\ln2+O_k(k)$ balls.
By comparison, the total number of balls is $X \sim_k mk/(2^{k-1}-1) \sim_k n\, k\ln 2$ \whp\
In effect, the expected number of balls that a typical bin receives is about $k \ln 2$,  way smaller than the capacity of that bin.
Indeed, since the  number of balls that are received by a typical bin is approximately $\Bin(kr, \frac{nk \ln2}{km}) \approx \Bin(kr, 2^{-k+1})$,
%so that the expected value is far less than the typical capacity of $kr$,
the number of balls can be approximated well by a $\Po(\lambda)$ distribution (with $\lambda=kr/(2^{k-1}-1) \sim_k k\ln 2$).
Thus, the probability that a bin remains empty is close to
$\exp(-\lambda)$, which was the probability of the same event in the experiment without capacities. 
The technical details of this argument are quite delicate, as the fluctuations of the capacities need to be
controlled \emph{very} carefully.

Finally, similar arguments as in the proof of \Prop~\ref{Prop_clusterSize} yield
	%\begin{equation}\label{eqLemma_firstMoment3}
	$\textstyle\pr\brk{\sigma\in\cS_{g,\beta}(\PHI_{\vec d})|\sigma\in\cS_\beta(\PHI_{\vec d})}=1-o(1).$
%	\end{equation}
%the expansion properties of the random formula $\PHI_{\vec d}$ imply that
%shows that given 
%that the occupancy problem works out,
%%$\sigma\in\cS_\beta(\PHI_{\vec d})$, 
%$\sigma$ is in fact good \whp\ (details omitted).
%	\aco{THIS NEEDS WORK!}
Thus, the assertion follows from~(\ref{eqLemma_firstMoment1})--(\ref{eqLemma_firstMoment2}).
%To determine the probability that $\sigma$ is $\beta$-heavy given
%
%7
\qed\end{proof}

We now turn to the second moment.
Fix some $\sigma\in\cbc{0,1}^V$, say $\sigma=\vecone$.
Let $Z_{g,\beta}(t,\sigma)$ denote the number of good $\tau\in\cS(\PHId)$ at distance $t$ from $\sigma$.
Using the linearity of expectation and recalling that  the set of NAE-solutions is symmetric with respect to inversion, we obtain
%	\begin{equation}\label{eqsmmdecomp1}
%%		\textstyle
%	\end{equation}
%Further,% as~(\ref{eqLemma_firstMoment3}) shows that $\pr\brk{\sigma\in\cS_\beta(\PHId)\mbox{ is good}}\sim\pr\brk{\sigma\in\cS_\beta(\PHId)}$,
%%and 
%because
%(\ref{eqsmmdecomp1}) yields
	\begin{eqnarray}%\nonumber
%	\textstyle
	\Erw\brk{Z_{g,\beta}(\PHId)|\sigma\in\cS_{g,\beta}(\PHId)}&\leq&
%		\sum_{t=0}^n\Erw\brk{Z_{g,\beta}(t,\sigma)|\sigma\in\cS_\beta(\PHId)\mbox{ is good}}\\
%	\Erw\brk{Z_{g,\beta}(\PHId)|\sigma\in\cS_\beta(\PHId)\mbox{ is good}}
%		&\leq&
		2\sum_{0\leq t\leq n/2}\Erw\brk{Z_{g,\beta}(t,\sigma)|\sigma\in\cS_{g,\beta}(\PHId)}.
			\label{eqsmmdecomp1}
	\end{eqnarray}
Let $I=\cbc{t\in\ZZ:(\frac12-2^{-k/3})n\leq t\leq n/2}$.
The first two conditions from Definition~\ref{XDef_good} ensure that given that $\sigma$ is good, with certainty we have
	$$%\textstyle
	\sum_{t\leq0.01n} Z_{g,\beta}(t,\sigma)\leq\abs{\cC(\sigma)}\leq\Erw\brk{Z_\beta(\PHId)}\mbox{ and }
		\sum_{0.01n<t<(\frac12-2^{-k/3})n}Z_{g,\beta}(t,\sigma)=0.$$
Hence, \Lem~\ref{Lemma_firstMoment} and (\ref{eqsmmdecomp1}) yield
	\begin{eqnarray}%\nonumber
%	\textstyle
	\Erw\brk{Z_{g,\beta}(\PHId)|\sigma\in\cS_{g,\beta}(\PHId)}&\leq&
%		\sum_{t=0}^n\Erw\brk{Z_{g,\beta}(t,\sigma)|\sigma\in\cS_\beta(\PHId)\mbox{ is good}}\\
%	\Erw\brk{Z_{g,\beta}(\PHId)|\sigma\in\cS_\beta(\PHId)\mbox{ is good}}
%		&\leq&
		(2+o(1))\Erw\brk{Z_{g,\beta}(\PHId)}+
		2%\hspace{-6mm}
		\sum_{t\in I}%\hspace{-6mm}
			\Erw\brk{Z_{g,\beta}(t,\sigma)|\sigma\in\cS_{g,\beta}(\PHId)}.
			\label{eqsmmdecomp2}
	\end{eqnarray}
%	
%\aco{Let $\alpha=2^{-k/3}$.}
%Combining 
%the upper bound imposed on $\sum_{0\leq t\leq(\frac12-\alpha)n}Z_{g,\beta}(t,\sigma)$ imposed by the definition of ``good''
%with the lower bound on
%$\Erw_{\PHI_{\vec d}}\brk{Z_{g,\beta}}$ from \Lem~\ref{Lemma_firstMoment},
%we find
%\begin{corollary}\label{Cor_boundary}
%For a random $\vec d$ we have 
%	$\sum_{0\leq t\leq(\frac12-\alpha)n}\Erw_{\PHI_{\vec d}}\brk{Z_{g,\beta}(t,\sigma)|\sigma\in\cS_{\beta,g}(\PHI_{\vec d})}
%		\leq\Erw_{\PHI_{\vec d}}\brk{Z_{g,\beta}}$
%\whp\
%\end{corollary}
%\begin{proof}
%This follows from the com definition of ``good'' 
%\qed\end{proof}
%
%
%Given~(\ref{eqsmmdecomp2}), reduces the proof 
This reduces the proof to the analysis of the ``central terms'' with $t\in I$.
The result of this is
%
%Finally, with respect to the ``central terms'' we have
%
\begin{lemma}\label{Lemma_central}
There is a constant $C'=C'(k)\geq1$ such that
for a random $\vec d$ we have 
	\begin{equation}\label{eqcentral}
	\textstyle
	\sum_{t\in I}\Erw\brk{Z_{g,\beta}(t,\sigma)|\sigma\in\cS_{\beta,g}(\PHI_{\vec d})
		}\leq C'\cdot\Erw\brk{Z_{g,\beta}(\PHI_{\vec d})}\qquad\mbox{ \whp}
	\end{equation}
\end{lemma}
\begin{proof}[sketch]
This is technically the most challenging bit of this work.
The argument boils down to estimating the probability that two random $\vec\sigma,\vec\tau\in\cbc{0,1}^n$
with $\dist(\vec\sigma,\vec\tau)/n=\alpha\in[\frac12-2^{-k/3},\frac12]$
simultaneously are good $\beta$-heavy solutions.
To compute this probability, we need to analyze the interplay of two occupancy problems
as in the proof of \Lem~\ref{Lemma_firstMoment} with respect to the same degree sequence $\vec d$.

More precisely, let $B=\bigcup_{x\in V}\cbc x\times\cbc{1,\ldots,d_x}$ be a set of $km$ ``balls''. %, with the $(x,l)$, $1\leq l\leq d_x$, representing each variable $x$.
Generating $\PHId$ is equivalent to drawing a random bijection $\vec\pi:\brk{m}\times\brk k\ra B$, with
$\pi(i,j)=(x,l)$ indicating that $x$ is the underlying variable of the $j$th literal of clause $i$, and independently choosing a map $\vec s:\brk m\times\brk k\ra\cbc{\pm1}$ indicating the signs.
Further, we represent the occupancy problems for $\vec\sigma,\vec\tau$ by
two ``colorings'' $g_\sigma,g_\tau:B\ra\cbc{\red,\blue}$, with $g_\sigma(x,l)=\red$ indicating that the $l$th position in bin $x$ is occupied under $\sigma$ (and analogously for $\tau$).
%
% (meant to represent specific solutions to the occupancy problem)
We compute the probability $p(\alpha,g_\sigma,g_\tau)$ that $\vec\pi,\vec s$
induce a formula in which
\begin{enumerate}
\item[$\bullet$] literal $(i,j)$ supports clause $i$ under $\vec\sigma$ iff $g_\sigma\circ\pi(i,j)=\red$, and similarly for $\vec\tau$.
\item[$\bullet$] both $\vec\sigma,\vec\tau$ are good $\beta$-heavy solutions.
\end{enumerate}
%The problem with computing $p(\sigma,\tau,g_\sigma,g_\tau)$ is that the 
%We solve this problem by exhibiting an alternative description of the experiment in which the clauses can be treated independently.
%With two random assignments $\vec\sigma^*,\vec\tau^*$ at distance $\dist(\sigma^*,\tau^*)=n/2$ as a ``reference point'',
The result is that
for any $g_\sigma,g_\tau$ the ``success probability'' is \emph{minimized} at $\alpha=1/2$.
Quantitatively, %we obtain the bound
	\begin{equation}\label{XeqLaplace1}
	\frac{p(\alpha,g_\sigma,g_\tau)}{p(1/2,g_\sigma,g_\tau)}= %\cdot\Erw\brk{Z_{g,\beta}(\PHId)}\cdot
			\exp\brk{O_k(k^4/2^k)(\alpha-1/2)^2n}
			\quad\mbox{ for any $g_\sigma,g_\tau$}.
	\end{equation}
On the other hand, the total \emph{number} of assignment pairs satisfies
	\begin{equation}\label{XeqLaplace2}
	\frac{\abs{\cbc{(\sigma,\tau):\dist(\sigma,\tau)=\alpha n}}}{\abs{\cbc{(\sigma,\tau):\dist(\sigma,\tau)=n/2}}}=\bink n{\alpha n}/\bink n{n/2}
		=\exp(-(4-o_k(1))(\alpha-1/2)^2n),
	\end{equation}
which is \emph{maximized} at $\alpha=1/2$.
Combining~(\ref{XeqLaplace1}) and~(\ref{XeqLaplace2}), we see that for any two colorings $g_\sigma,g_\tau$ the dominant
contribution to the second moment stems from $\alpha=\frac 12+O(1/\sqrt n)$,
	i.e., from ``perfectly decorrelated'' $\vec\sigma,\vec\tau$.
The assertion follows by evaluating the contribution of such $\alpha$ explicitly
and summing over $g_\sigma,g_\tau$.
%%explicitly, we obtain the assertion.
%
%%$p(g_\sigma,g_\tau)$ 
%
%
%In fact, for this part of the argument it is crucial to fix the degree sequence, as otherwise
%the two occupancy problems ``collude'' to boost the overall success probability,
%and thus the left hand side of~(\ref{eqcentral}).
\qed\end{proof}

\noindent{\bf Acknowledgment.}
The first author thanks Dimitris Achlioptas and Lenka Zdeborov\'a for helpful discussions on the second moment
method and the statistical mechanics work on random CSPs.

\newpage
\noindent{\Large\bf Appendix}

\medskip
\noindent
This appendix contains the details omitted from the extended abstract. Section~\ref{Apx_prelim} contains some preliminary facts about random variables that will be used many times.
Appendix~\ref{Apx_upper} contains the full proof of the upper bound claimed in \Thm~\ref{Thm_NAE}
	(with $\eps_k$ exponentially small in $k$).
Finally, in Appendices~\ref{Apx_lower} and~\ref{Apx_centreTerms} we carry out the second moment argument in full.

\begin{appendix}

\section{Preliminaries}\label{Apx_prelim}

The next lemma provides an asymptotically tight bound for the probability that a sum of independent and identically distributed random variables attains a specific value. It will be an important tool in our further analysis, since we will be often interested in the exact probabilities of~\emph{rare} events.
\begin{lemma}
\label{lem:locallimit}
Let $X_1, \dots, X_n$ be independent random variables with support on $\mathbf{N}_0$ with probability generating function $P(z)$. Let $\mu = \Exp[X_1]$ and $\sigma^2 = \Var[X_1]$. Assume that $P(z)$ is an entire and aperiodic function. %Set
%\[
%	p_\alpha = \Pr[X_1 + \dots + X_n = \alpha n].
%\]
Then, uniformly for all $T_0 < \alpha < T_\infty$, where $T_x = \lim_{z\to x}\frac{zP'(z)}{P(z)}$,  as $n \to \infty$
\begin{equation}
\label{eq:locallimit}
	\Pr[X_1 + \dots + X_n = \alpha n] = (1 + o(1))\, \frac{1}{\zeta \sqrt{2\pi n \xi}} \, \left(\frac{P(\zeta)}{\zeta^\alpha}\right)^n, 
\end{equation}
where $\zeta$ and $\xi$ are the solutions to the equations
\begin{equation}
\label{eq:saddlepoints}
	\frac{\zeta P'(\zeta)}{P(\zeta)} = \alpha
	\qquad \text{ and } \qquad
	\xi = \frac{d^2}{dz^2}\left(\ln P(z) - \alpha \ln z\right) \Big|_{z = \zeta}.
\end{equation}
Moreover, there is a $\delta_0 > 0$ such that for all $0 \le |\delta| \le \delta_0$ the following holds. If $\alpha = \Exp[X_1] + \delta\sigma$, then
\begin{equation}
\label{eq:limitclose}
	\Pr[X_1 + \dots + X_n = \alpha n] = (1 + O(\delta)) \, \frac{1}{\sqrt{2\pi n \sigma}} \, e^{(-\delta^2/2 + O(\delta^3))n}.
\end{equation}
\end{lemma}
\begin{proof}
The first statement follows immediately from Theorem VIII.8 and the remark after Example VIII.11 in~\cite{FS}. To see the second statement let us write $\zeta_\delta$ for the solution to the equation $\frac{\zeta_\delta P'(\zeta_\delta)}{P(\zeta_\delta)} = \mu + \delta\sigma$. Since $P(1) = 1$ and $P'(1) = \mu$ we infer that if $\delta = 0$, then $\zeta_\delta = 1$. Moreover, a Taylor series expansion around $z=1$ guarantees for all $\delta$ in a bounded interval around 0 that
\[
\mu + \delta\sigma
= \frac{\zeta_\delta P'(\zeta_\delta)}{P(\zeta_\delta)}
= \frac{P'(1)}{P(1)} + (\zeta_\delta-1) \, \frac{P''(1) + P'(1) - \frac{P'(1)^2}{P(1)}}{P(1)} + O((\zeta_\delta-1)^2).
\]
Since $\sigma^2 = P''(1) + P'(1) - P'(1)^2$, for all $\delta$ in a bounded interval around 0 we have that $\zeta_\delta = 1 + {\delta}/\sigma + O(\delta^2)$. In order to show~\eqref{eq:limitclose} we evaluate the right-hand side of~\eqref{eq:locallimit} at $\zeta = \zeta_\delta$. Again a Taylor series expansion around $z=1$ guarantees that 
\[
\begin{split}
	\frac{P(\zeta_\delta)}{\zeta_\delta^\alpha}
	&\stackrel{\qquad~~~}{=} {P(1)} + (\zeta_\delta-1)(P'(1) - \alpha P(1)) + \frac{(\zeta_\delta-1)^2}2\left(P''(1) + P(1)\alpha^2 + P(1)\alpha - 2P'(1)\alpha\right)  + O(\delta^3) \\
	&\stackrel{(\alpha = \mu + \delta)}{=} 1 - \delta^2 + \frac{\delta^2}{2\sigma^2}(P''(1) + \mu - \mu^2 + O(\delta)) + O(\delta^3) \\
	&\stackrel{\qquad~~~}{=} 1 - \frac{\delta^2}{2} + O(\delta^3).
\end{split}
\]
The exponential term in~\eqref{eq:limitclose} is then obtained by using the fact $1 - x = e^{-x - \Theta(x^2)}$. Finally, note that
\[
	\frac{d^2}{dz^2}\left(\ln P(z) - \alpha \ln z\right) = \frac{P''(z)}{P(z)} - \frac{P'(z)^2}{P(z)^2} + \frac{\alpha}{z^2}.
\]
By applying again Taylor's Theorem to this function we obtain after some elementary algebra (details omitted) that the value of this function at $\zeta = \zeta_\delta$ equals $\sigma + O(\delta)$, and the proof of~\eqref{eq:limitclose} is completed.
\qed
\end{proof}
The next statement provides tight asymptotic bounds for binomial coefficients.
\begin{proposition}
\label{prop:binomials}
Let $0 < \alpha \le 1/2$ and $-1/2 < \eps < 1/2$ be such that $0< \alpha + \eps < 1$. Then, as $N\to\infty$
\[
	\binom{N}{\alpha N} = \frac{1+o(1)}{\sqrt{2\pi f(\alpha) N}} \, e^{H(\alpha)\, N}
	\quad\text{ and }\quad
	\binom{N}{(\alpha + \eps) N} = \frac{1+o(1)}{\sqrt{2\pi f(\alpha + \eps) N}} \, e^{(H(\alpha) + \eps \log(\frac{1-\alpha}{\alpha}) + O(\eps^2/\alpha))N},
\]
where $H(x) = -x\ln x - (1-x)\ln(1-x)$ denotes the entropy function and $f(x) = x(1-x)$.
\end{proposition}
\begin{proof}
The first statement is well-known, see e.g.~\cite{FS}. To see the second statement, note first that that $H'(x) = \ln(\frac{1-x}{x})$ and $H''(x) = (x(x-1))^{-1}$, both valid in $(0,1)$. Then, Taylor's Theorem guarantees that
\[
	H(\alpha + \eps) = H(\alpha) + \eps H'(\alpha) + O(\eps^2/\alpha),
\]
from which the second statement follows immediately.
\qed
\end{proof}

\section{The upper bound on $\rkNAE$ %\Prop~\ref{Prop_clusterSize}}\label{Sec_clusterSize
	}\label{Apx_upper}

To prove the upper bound on $\rkNAE$ we are going to combine the
upper bound on the expectation of $Z_{\beta}$ from \Prop~\ref{Prop_roughFirst} with a \emph{lower} bound
on the cluster sizes of $\beta$-heavy assignments, see Lemma~\ref{Lemma_Zbetagamma}.
Let $\lambda=kr/(2^{k-1}-1)$. First of all, we fill the missing pieces in the proof of \Prop~\ref{Prop_roughFirst}. The next lemma provides the analysis for the balls-into-bins game that was omitted in the proof of \Prop~\ref{Prop_roughFirst}.

\begin{lemma}
Let $X \sim \Bin(m,k/(2^{k-1}-1))$. We throw $X$ balls into $n$ bins uniformly at random. Let $B_i$ denote the number of bins that receive $i$ balls. Then, for any $-3/2\le \beta \le 1$
\begin{equation}
\label{eq:ballsbinssimple}
	n^{-1} \ln {\Pr\left[B_0 = (1-\beta) e^{-\lambda} n\right]} = n^{-1} \ln {\Pr\left[\Bin(n,e^{-\lambda}) = (1-\beta) e^{-\lambda} n\right]} + O_k(k4^{-k}).
\end{equation}
\end{lemma}
\begin{proof}
We shall estimate the desired probability by conditioning on any specific value $x$ of $X$. Let $F_i$ be the number of balls in the $i$th bin, and let $P_1, \dots, P_n$ be independent Poisson distributed random variables with mean $\lambda$. It is well-known and easy to verify that the distribution of $(F_1, \dots, F_n)$ is the same as the distribution of $(P_1, \dots, P_n)$, \emph{conditioned on the event} ${\cal A}(x) = \text{``}\sum_{1\le i\le n} P_i = x\text{''}$. So, if we denote by $N_0$ the number of $P_i$'s that are equal to 0, we infer that
\[
	\Pr\left[B_0 = (1-\beta) e^{-\lambda} n ~|~ X = x\right] = 
	\Pr\left[N_0 = (1-\beta) e^{-\lambda} n ~\middle|~ {\cal A}(x)\right].
\]
By the law of total probability this equals
\begin{equation*}
	\Pr\left[B_0 = (1-\beta) e^{-\lambda} n ~|~ X = x\right] = 
	\Pr\left[N_0 = (1-\beta) e^{-\lambda} n\right] \cdot \frac{\Pr[{\cal A}(x)~|~ N_0 = (1-\beta) e^{-\lambda} n]}{\Pr[{\cal A}(x)]}.
\end{equation*}
Note that $N_0 \sim \Bin(n,e^{-\lambda})$. Furthermore, if we denote by $P_1', \dots , P_{\xi n}'$, where $\xi = 1 - (1-\beta)e^{-\lambda}$, independent Poisson variables that are conditioned on being at least $1$, then the above equation implies that
\begin{equation}
\label{eq:ratioinsimple}
	\frac{\Pr\left[B_0 = (1-\beta) e^{-\lambda} n \right]}{\Pr\left[\Bin(n,e^{-\lambda}) = (1-\beta) e^{-\lambda} n\right]}
	=
	\sum_{x = \xi n}^m \frac{\Pr[\sum_{i = 1}^{\xi n} P_i' = x]}{\Pr[\Po(\lambda n) = x]} \cdot \Pr\left[\Bin(rn, {k}/({2^{k-1}-1})) = x\right].
\end{equation}
In order to complete the proof of~\eqref{eq:ballsbinssimple} we will derive in the sequel appropriate bounds for the right-hand side of the above equation. First, to obtain a lower bound, note that $\xi < \lambda$, since $\xi < 1$ and $\lambda = k \ln 2 + O_k(k2^{-k})$, which is $>1$ for sufficiently large $k$. Thus, we can obtain a lower bound for~\eqref{eq:ratioinsimple} by considering only the term in the sum that corresponds to $x = \lambda n$. Since $\Exp[\Po(\lambda n)] = \Exp[\Bin(rn, {k}/({2^{k-1}-1}))] = \lambda n$, we infer by applying Lemma~\ref{lem:locallimit} that
\[
	\Pr[\Po(\lambda n) = \lambda n] = \Theta(n^{-1/2})
	\qquad \text{and}\qquad
	\Pr[\Bin(rn, {k}/({2^{k-1}-1})) = \lambda n] = \Theta(n^{-1/2})	.
\]
It remains to bound $\Pr[\sum_{i = 1}^{\xi n} P_i' = \lambda n]$. Note that $\Exp[P_1'] = \frac{\lambda}{1 - e^{-\lambda}}$. If we write $N = \xi n$, then 
\[
	\Pr\left[\sum_{i = 1}^{\xi n} P_i' = \lambda n\right] = 
	\Pr\left[\sum_{i = 1}^{N} P_i' = \left(\Exp[X_1] + \frac{\beta \lambda e^{-\lambda}}{\xi (1 - e^{-\lambda}) }\right)N\right],
\]
i.e., we require that the sum of the $P_i'$'s deviates from the expected value by $O_k(k2^{-k}n)$. By applying  Lemma~\ref{lem:locallimit}, where we set $\delta = O_k(k^{1/2}2^{-k})$, we conclude that the right-hand side of~\eqref{eq:ratioinsimple} is at least $\exp\{-O_k(k4^{-k}n)\}$. This shows the lower bound in~\eqref{eq:ballsbinssimple}.

In the remainder of this proof we will show an upper bound for the right-hand side of~\eqref{eq:ratioinsimple}. To this end, we will argue that the ratio $\Pr[\Bin(rn, {k}/({2^{k-1}-1})) = \gamma \lambda n] / \Pr[\Po(\lambda n) = \gamma \lambda n]$ is essentially bounded for all $x$ in the given range, from which the claim immediately follows. More specifically, let us write $x = \gamma \, \lambda n$, where $\xi/\lambda \le \gamma \le r/\lambda$. By applying Stirling's Formula $N! = (1+o(1))\sqrt{2\pi N}(N/e)^N$ we infer that
\begin{equation}
\label{eq:tmpEstimatePo}
	\Pr[\Po(\lambda n) = \gamma \lambda n] = \Theta(1) \, {n^{-1/2}} \, \exp\{\lambda n (-1 + \gamma - \gamma \ln \gamma)\}.
\end{equation}
Moreover, by abbreviating $p = k/(2^{k-1}-1)$ we get
\[
	\Pr[\Bin(rn, {k}/({2^{k-1}-1})) = \gamma \lambda n] = \binom{rn}{(\gamma p)\, rn} p^{(\gamma p) \, rn} (1 - p)^{(1 - \gamma p)rn}.
\]
Since $\binom{N}{\alpha N} \le e^{H(\alpha)\, N}$, where $H$ denotes the entropy function, we obtain after some elementary algebra
\[
	\Pr[\Bin(rn, p) = \gamma \lambda n] \le \exp\left\{\lambda n \left(-\gamma\ln\gamma - \frac{1-\gamma p}p\ln\left(\frac{1-\gamma p}{1-p}\right)\right)\right\}.
\]
By combining this with~\eqref{eq:tmpEstimatePo} we obtain the estimate
\[
	\frac{\Pr[\Bin(rn, {k}/({2^{k-1}-1})) = \gamma \lambda n]}{\Pr[\Po(\lambda n) = \gamma \lambda n]}
	\le \Theta(\sqrt{n}) \, e^{f(\gamma) \, \lambda n},
	\quad \text{ where }\quad
	f(\gamma) = 1 - \gamma - \frac{1-\gamma p}p\ln\left(\frac{1-\gamma p}{1-p}\right).
\]
Recall that $0 < \xi/\lambda \le \gamma \le r/\lambda = 1/p$, and note that both $f(0)$ and $f(1/p)$ are $< 0$. Moreover, $f$ has an extremal point at $\gamma = 1$, where $f(1) = 0$. Thus, for all $\gamma$ in the considered range we have that $f(\gamma) \le 0$, which implies that the right-hand side of~\eqref{eq:ratioinsimple} is bounded from above by at most a polynomial in $n$. This completes the proof of the lemma.
\qed\end{proof}
The proof of \Prop~\ref{Prop_roughFirst} then completes by applying the following statement.
\begin{lemma}
\label{lem:auxf}
There is a $k_0 \ge 3$ such that the following is true. Let $Y \sim \Bin(n,e^{-\lambda})$. For any $-3/2 \le \beta \le 1$
\[
	n^{-1} \ln \Pr\left[Y = \lfloor(1-\beta)e^{-\lambda}n\rfloor\right] = f(\beta) + O_k(4^{-k}).
\]
\end{lemma}
\begin{proof}
Let us abbreviate $\xi = (1-\beta)e^{-\lambda}$. We will assume that $\xi n = \lfloor \xi n \rfloor$, i.e., that $\beta = 1 - N(e^{-\lambda}n)^{-1}$ for some $N \in \mathbf{N}_0$. To see that this is sufficient, note that by Taylor's Theorem, for any $\beta \ge 1$ and any $|\eps_n| \le (e^{-\lambda} n)^{-1}$ such that $\beta + \eps_n \le 1$ there is a $\delta \in [\beta, \beta + \eps_n]$ such that
\[
	f(\beta + \eps_n) = f(\beta) + \eps_n f'(\delta) = f(\beta) + \eps_n e^{-\lambda} \ln(1-\delta) = f(\beta) + O_k(4^{-k}).
\]
With the above assumption we proceed with the proof of the claim.
The definition of the binomial distribution implies
\begin{equation}
\label{eq:locallimitX}
	\Pr[Y = (1-\beta)e^{-\lambda}n] = \binom{n}{\xi n} e^{-\lambda \xi n} (1-e^{-\lambda})^{(1-\xi)n}.
\end{equation}
If $\beta = 1$, then $\xi = 0$ the above expression simplifies to 
\[
	(1-e^{-\lambda})^n = \exp\{n\ln(1 - e^{-\lambda})\} = \exp\{n(- e^{-\lambda} - \Theta(e^{-2\lambda}))\}.
\]
Since $f(1) = e^{-\lambda}$ and $\lambda = k\ln 2 + \Theta(k2^{-k})$, we infer that the statement is true for $\beta = 1$. It remains to treat the case $\beta < 1$. Standard bounds for the binomial coefficients imply
\[
	\binom{n}{\xi n} = \frac{\Theta(1)}{\sqrt{\xi(1-\xi)n}} e^{nH(\xi)},
	\quad \text{ where } \quad
	H(x) = -x\ln x - (1-x)\ln(1-x).
\]
Using the estimate $\ln(1-x) = -x - \Theta(x^2)$, which is valid for $|x|<1$, we infer after some elementary algebra that
\begin{equation}
\label{eq:estbinomlocallimit}
	n^{-1} \ln \binom{n}{\xi n} = e^{-\lambda}((1-\beta)\lambda - (1-\beta)\log(1-\beta) + (1-\beta)) + O_k(4^{-k})
\end{equation}
Similarly, the second and the third term in~\eqref{eq:locallimitX} can be estimated with
\[
	n^{-1}\ln\left(e^{-\lambda \xi n} (1-e^{-\lambda})^{(1-\xi)n}\right) = -e^{-\lambda}((1-\beta)\lambda + 1) + O_k(4^{-k}).
\]
By plugging this fact together with~\eqref{eq:estbinomlocallimit} into~\eqref{eq:locallimitX} we finally obtain the desired statement.
\qed \end{proof}
We proceed with the proof of the upper bound in Theorem~\ref{Thm_NAE}.
Let $Z_{\beta,\gamma}$ denote the number of $\beta$-heavy solutions $\sigma$ such that $\frac1n\log_2\abs{\cC(\sigma)}\le(1-\beta-\gamma)e^{-\lambda}$. The following statement provides an upper bound for the expected number of such solutions.
\begin{lemma}\label{Lemma_Zbetagamma}
For any $-3/2\le \beta \le 1$ and $\gamma>k^{5/2}e^{-\lambda}$ we have for sufficiently large $k$
	$$\frac1n\ln\Erw\brk{Z_{\beta,\gamma}}\leq\frac1n\ln\Erw\brk{Z_{\beta}}-\frac{\ln k}{6}\gamma e^{-\lambda}.$$
\end{lemma}
\begin{proof}
Let $\sigma\in\cbc{0,1}^V$ be an assignment;
	for the sake of concreteness, assume that $\sigma=\vecone$.
	In order to bound $\Erw\brk{Z_{\beta,\gamma}}$ it is sufficient to 
estimate the probability of the event
	$$\cE=\cbc{\frac1n\log_2\abs{\cC(\sigma)}\leq(1-\beta-\gamma)e^{-\lambda}},$$
given that $\sigma\in\cS_\beta(\PHI)$.
Let $\cF(\sigma)$ denote the set of free variables, and denote by $\cX$ be the set of clauses  that do not contain both a positive and a negative literal whose underlying variable is in $V\setminus\cF(\sigma)$.
Then only the clauses in $\cX$ impose constraints on the free variables. We decompose $\cX$ into $k-1$ subsets $\cX_2, \dots,\cX_{k}$, where $\cX_i$ the set of all clauses in $\cX$ that contain $i$ variables from $\cF(\sigma)$. Note that $\cX = \cup_{i=2}^k \cX_i$, as any clause with only one variable from $\cF(\sigma)$ necessarily contains both positive and negative literals whose underlying variables are not free. Let $X_i = |X_i|$. Since only the clauses in $\cX$ impose constraints on  variables from $\cF(\sigma)$ that occur in them, we infer that
	$$\frac1n\log_2\abs{\cC(\sigma)}\geq\abs{\cF(\sigma)}-Y,
	~~\text{where}~~
	Y = \sum_{i=2}^k iX_i.
%	~~\text{and}~~
%	Y_i = iX_i.
	$$
In the remainder we will show that
\begin{equation}
\label{eq:tailboundYTODO}
	\frac1n \ln \Pr\left[Y > \gamma e^{-\lambda}n ~|~ \sigma\in\cS_\beta(\PHI)\right] \le -\frac{\ln k}{6} \gamma e^{-\lambda},
\end{equation}
from which the statement in the lemma follows immediately.

Note that the set $\cF(\sigma)$ is determined by the critical clauses only.
Therefore, given that $\sigma\in\cS_\beta(\PHI)$, the variables that occur in the non-critical clauses are
independent and uniformly distributed over the set of all variables.
Similarly, given that $\sigma\in\cS_\beta(\PHI)$ the $k-1$ variables that contributed the ``majority value'' to each
critical clause are independently uniformly distributed.
Therefore, $X_i$ is stochastically dominated by a binomial random variable
\[
	X_i' \sim \Bin(m, p_i), ~~\text{where}~~ p_i = 2^{-k+1} \cdot 2^i\bink ki((1-\beta)\exp(-\lambda))^i.
\]
Our assumption $-3/2 \le \beta\le 1$ guarantees that $(1-\beta)e^{-\lambda} \le 3e^{-\lambda} \le 3\cdot 2^{-k}$. By using the estimate $\binom{k}{i} \le k^i$ we infer that %the expectation of $X_i/n$ is at most
\begin{equation}
\label{eq:p_iupper}
	%\mu_i \le \frac{m}{2^{k-1}} \cdot 2^i\bink ki ((1-\beta)e^{-\lambda})^i \le \ln2 \, (6k2^{-k})^i.
	p_i \le {2^{-k+1}} \cdot 2^i\bink ki ((1-\beta)e^{-\lambda})^i \le 2^{-k+1} \, (6k2^{-k})^i.
\end{equation}
Moreover, note that the $X_i$ are negatively correlated. Indeed, let $X_{i,j}$ be the indicator for the event that the clause $\PHI_j \in \cX_i$. Then, for all $i \neq i'$ we have $\Erw[X_{i, j} X_{i',j}] = 0 \le \Erw[X_{i, j}]\Erw[X_{i',j}]$, and otherwise, if $(i,j) \neq (i',j')$, then $X_{i,j}$ and $X_{i',j'}$ are independent. Thus, for any $\delta > 0$, Markov's inequality implies with $t = \gamma e^{-\lambda}$
\[
	\Pr\left[Y > t~|~ \sigma\in\cS_\beta(\PHI)\right]
	\le
	e^{-\delta t}\prod_{i=2}^k\Erw[e^{\delta i X_i}]
	\le e^{-\delta t}\prod_{i=2}^k\Erw[e^{\delta i X_i'}]
	\le e^{-\delta t}\prod_{i=2}^k (p_ie^{\delta i} + 1 - p_i)^{m},
\]
Let us fix $\delta = \frac15 \ln k$. By the arithmetic-geometric mean inequality we obtain that the expression in the previous equation is at most
\[
	e^{-\delta t}\left(\frac{\sum_{i=2}^k p_ie^{\delta i} + 1 - p_i}{k}\right)^{km}
	\stackrel{\eqref{eq:p_iupper}}\le
	e^{-\delta t}\left(1 + \frac{2^{-k+1}\sum_{i=2}^k (6k2^{-k})^i e^{\delta i}}{k}\right)^{km}
	=
	e^{-\delta t}\, e^{O_k(k^2 4^{-k} e^{2\delta})n}.
\]
Since $t = \gamma e^{-\lambda} > k^{5/2} 4^{-k}$, for sufficiently large $k$ we get~\eqref{eq:tailboundYTODO}, and the proof is completed.
\qed\end{proof}
%Let $Z_{\geq\beta,\gamma}$ denote the number of $\beta'$-heavy solutions $\sigma$
%such that $\beta'\geq\beta$ and
%	$$\frac1n\log_2\abs{\cC(\sigma)}\leq(1-\beta-\gamma)e^{-\lambda}.$$
Consider the function
\begin{eqnarray*}
	g(\beta) = h(\beta)-(1-\beta)e^{-\lambda}\ln2,
\end{eqnarray*}
where
\begin{eqnarray*}
h(\beta) =\frac{2\rho-\ln2}{2^{k}}+f(\beta) \qquad&\mbox{and}&\qquad f(\beta)=-({(1-\beta)\ln(1-\beta)+\beta})e^{-\lambda}.
\end{eqnarray*}
Let $r_*$ be the least density $r$ such that $g(\beta)<-k^34^{-k+1}$ for all $\beta\geq -1$. Since $g$ is maximized for $\beta = 1/2$, where $g(1/2) = \frac{2\rho - \ln2}{2^k} - \frac12e^{-\lambda}$, it is easily verified that
	$$r_*=2^{k-1}\ln2-\bc{\frac{\ln2}2+\frac14}+O_k(k^32^{-k}).$$

\begin{proposition}\label{Prop_betaScan}
With $r=r^*$ the random formula $\PHI$ does not have a NAE-solution \whp
\end{proposition}
\begin{proof}
Let $Z_{\leq\beta}$ be the number of solutions that are $\beta'$-heavy for some $\beta'\leq\beta$.
%, and define $Z_{\geq\beta}$ analogously. \Prop~\ref{Prop_roughFirst} implies that \whp\ $Z_{\leq-3/2}=0$.
 %Moreover, again by \Prop~\ref{Prop_roughFirst} we have that $\frac1n \ln \Erw[Z_\beta] = h(\beta)+ O_k(k4^{-k})$. Since $h(\beta) < 0$ for all, say, $\beta \ge 9/10$, we also infer that \whp\ $Z_{\ge 9/10} = 0$.
In order to prove that $Z_{\leq\beta}=0$ \whp\ for all $\beta$ we proceed as follows.
Let $-3/2=\beta_0<\cdots<\beta_\ell = 1$ be a sequence such that $|\beta_i-\beta_{i+1}|\leq\delta$ for all $i$, where $\delta=2^{-2^k}$. We are going to show inductively that $Z_{\leq\beta_i}=0$ \whp;	by the previous discussion we may assume that this is true for $i=0$.

Let us assume for the induction step that $i$ is such that \whp\ $Z_{\leq\beta_i}=0$. Let $\gamma_0=  k^3e^{-\lambda}$, and let $Z'$ be the number of solutions that are $\beta'$-heavy for some $\beta' > \beta_i$ and such that $\frac1n\log_2\abs{\cC(\sigma)}\geq(1-\beta_{i}-\gamma_0)e^{-\lambda}$. Then, by applying \Prop~\ref{Prop_roughFirst} and using that $h(x)$ is monotone increasing for $x \le 0$ and monotone decreasing for $x \ge 0$ we obtain
\[	
	\frac1n \ln \Erw[Z']
	%\le \frac1n \ln \Erw[Z_{\le \beta_{i+1}} - Z_{\le \beta_i}]
	\le \max_{\beta > \beta_i} h(\beta) + O_k(\delta + k4^{-k})
	= O_k(k4^{-k}) + 
	\begin{cases}
		h(0), &\text{ if } \beta_i \le 0, \\
		h(\beta_i), &\text{ if } \beta_i > 0
	\end{cases}.
\]
Let us first consider the case $\beta_i \le 0$. The choice of $r^*$ guarantees that $g(0) = h(0) - e^{-\lambda} \ln2 < -k^3 4^{-k+1}$. Since $Z'>0$ implies $Z'\ge\exp\{n(1-\beta_i-\gamma_0)e^{-\lambda}\ln2\} \ge \exp\{n(1-\gamma_0)e^{-\lambda}\ln2\}$ or otherwise $Z_{\le \beta_i}>0$ we infer for sufficiently large $k$ that
\[
	\Pr[Z' > 0] \le \Pr[Z_{\le \beta_i}>0] + \Erw[Z'] \exp\{-n(1-\gamma_0)e^{-\lambda}\ln2\} = o(1).
\]
On the other hand, if $\beta_i > 0$, then again the choice of $r^*$ is such that $g(\beta_i) = h(\beta_i) - (1-\beta_i)e^{-\lambda}\ln 2 < -k^3 4^{-k+1}$. Thus, for sufficiently large $k$
\[
\begin{split}
	\frac1n \ln \Erw[Z']
	&< -k^3 4^{-k+1} + (1-\beta_i)e^{-\lambda}\ln2 + O_k(k4^{-k})
	< -k^{7}4^{-k} + (1-\beta_i-\gamma_0)e^{-\lambda}\ln2.
	%\le -k^{7}4^{-k+1} + \frac1n\log_2\abs{\cC(\sigma)}.
\end{split}
\]
So, since $Z'>0$ implies $Z'\ge\exp\{n(1-\beta_i-\gamma_0)e^{-\lambda}\ln2\}$ or otherwise $Z_{\le \beta_i}>0$ we infer that
\[
	\Pr[Z' > 0] \le \Pr[Z_{\le \beta_i}>0] + \Erw[Z'] \exp\{-n(1-\beta_i-\gamma_0)e^{-\lambda}\ln2\} = o(1).
\]
Thus, in both cases we have that $\Pr[Z' > 0] = o(1)$. In remains to consider all satisfying assignments such that $\frac1n\log_2\abs{\cC(\sigma)}\leq(1-\beta_{i}-\gamma_0)e^{-\lambda}$. More specifically, let $Z_j'$ be the number of solutions that are $\beta'$-heavy for some $\beta_i < \beta' \le \beta_{i+1}$ and such that
\[
	(1-\beta_{i}-\gamma_{j+1})e^{-\lambda}
	\le \frac1n\log_2\abs{\cC(\sigma)}
	\le (1-\beta_{i}-\gamma_j)e^{-\lambda},
\]
where $\gamma_{j+1} = 2\gamma_j$.
Choose $\beta'$ be such that $\cS_{\beta'}(\PHI)\cap\cC(\sigma)$ is maximized.
Then
	\begin{equation}\label{eqcontradiction}
	|\cS_{\beta'}(\PHI)\cap\cC(\sigma)|\geq\frac{|\cC(\sigma)|}n.
	\end{equation}
Since $Z_{\leq\beta_i}=0$ \whp, we may assume that $\beta'>\beta_i$.
There are two cases to consider.

~\\
\emph{Case 1: $1-\beta_{i}-\gamma_{j+1} > 1-\beta'$.}
We will show that in this case the number of $\beta'$-heavy assignments is larger than the expected value by at least an exponential factor. Indeed, our assumption on $g$ implies for sufficiently large $k$ that
	\[
	\begin{split}
		\frac1n\ln\Erw\brk{Z_{\beta'}}
		&= h(\beta') + O_k(k4^{-k})
		< -k^3 4^{-k} + (1-\beta')e^{-\lambda}\ln2
		< -k^3 4^{-k} + (1-\beta_i-\gamma_{j+1})e^{-\lambda}\ln2.
	\end{split}
	\]
	However, if~(\ref{eqcontradiction}) holds then
		$$\frac1n\ln Z_{\beta'}\geq\frac1n\ln|\cC(\sigma)|-o(1)=(1-\beta_i-\gamma_{j+1})e^{-\lambda}\ln2-o(1).$$
	By Markov's inequality, the probability of this event is $\exp(-\Omega(n))$.

~\\
\emph{Case 2: $1-\beta_i-\gamma_{j+1} \leq 1-\beta'$.} The assumption guarantees the existence of a $\gamma'>0$ such that
\[
	1-\beta_i-\gamma_{j+1} = 1-\beta'-\gamma'.
\]
In this case we will show that the number of solutions in $\cS_{\beta',\gamma'}(\Phi)$ is larger than the expected value by at least an exponential factor. Equation~(\ref{eqcontradiction}) implies that
\begin{equation}
\label{eq:Zbeta'gamma'lower}
	\frac1n\ln Z_{\beta',\gamma'}
	\geq \frac1n\ln\abs{\cC(\sigma)} - o(1)
	= (1-\beta'-\gamma')e^{-\lambda}\ln2 - o(1).
\end{equation}
If $\gamma' > k^{5/2}e^{-\lambda}$, then by \Lem~\ref{Lemma_Zbetagamma} and our assumption on $g$
		$$\frac1n\ln\Erw\brk{Z_{\beta',\gamma'}}\leq h(\beta') + O_k(k4^{-k})-\frac{\ln k}{6}\gamma'e^{-\lambda}
		\le (1-\beta')e^{-\lambda}\ln 2 - \frac{\ln k}{6}\gamma'e^{-\lambda},$$
	Thus, by applying~\eqref{eq:Zbeta'gamma'lower}, we infer that $Z_{\beta',\gamma'}>\exp(\Omega(n))\Erw\brk{Z_{\beta',\gamma'}}.$
	By Markov's inequality, the probability of this event is $\exp(-\Omega(n))$.
	On the other hand, if $\gamma' < k^{5/2}e^{-\lambda}$, then for sufficiently large $k$
	\[
		\frac1n\ln\Erw\brk{Z_{\beta',\gamma'}}
		\le h(\beta') + O_k(k4^{-k})
		\le -k^3 4^{-k+1} + (1-\beta')e^{-\lambda}\ln2
		< -k^3 4^{-k} + (1-\beta'-\gamma')e^{-\lambda}\ln2.
	\]
	Thus, again by applying~\eqref{eq:Zbeta'gamma'lower}, we infer that also in this case $Z_{\beta',\gamma'}>\exp(\Omega(n))\Erw\brk{Z_{\beta',\gamma'}}$, and Markov's inequality asserts that the probability of this event is $\exp(-\Omega(n))$.
	
~\\
Since the probability that either case occurs is $\exp(-\Omega(n))$, we conclude that
	the same is true of the event ``$Z_j'>0$''. Taking the union bound over $j$ then completes the induction step, i.e., $Z_{\leq\beta_{i+1}}=0$ \whp
\qed
\end{proof}

\noindent
Finally, the upper bound on $\rkNAE$ claimed in \Thm~\ref{Thm_NAE}
follows directly from \Prop~\ref{Prop_betaScan}.

\section{Proof of the lower bound} %Proof of \Lem~\ref{Lemma_firstMoment}}\label{Sec_firstMoment}
\label{Apx_lower}

\subsection{Outline}

%\aco{explain a little that $\vec d$ is a ``quenched disorder''}

Let $\vec d,\vec D$ be as in \Sec~\ref{Sec_second}.
In the extended abstract, we presented a slightly streamlined definition of ``good''.
%	(for which, nonetheless, the second moment argument works out).
Technically it will be more convenient to work with the following definition.
(It will emerge later that the two definitions are equivalent.)
Recall that $\lambda=kr/(2^{k-1}-1)$.

\begin{definition}\label{Def_good}
We call a solution $\sigma\in\cbc{0,1}^V$ of $\PHI_{\vec d}$ \emph{$\beta$-good} if it satisfies the following conditions.
\begin{enumerate}
\item $\sigma$ is $\beta$-heavy and the total number of critical clauses is equal to $\lambda n$.
\item No variable supports more than $3k$ clauses.
\item We have
		$$\frac1n\ln\abs{\cbc{\tau\in\cS(\PHI_{\vec d}):\dist(\sigma,\tau)/n\leq\frac12-2^{-k/3}}}
				\leq(1-\beta)\exp(-\lambda)\ln2+O_k(k^{13} 4^{-k}).$$
\end{enumerate}
\end{definition}
Let $\cZ_{\beta}$ be the number of $\beta$-good solutions.
As a first step, we determine the expectation of $\cZ_\beta$.

\begin{proposition}\label{Prop_realFirst}
Suppose that $\vec d$ is chosen from the distribution $\vec D$.
Then \whp\
	$$\frac1n\ln\Erw\brk{\cZ_{\beta}}\geq\frac{2\rho-\ln2}{2^{k}}+f(\beta) - O_k(k^{13} 4^{-k}).$$
\end{proposition}

Let us fix an assignment $\sigma\in\cbc{0,1}^V$, say $\sigma=\vecone$.
Moreover, let $\Sigma$ be the event that $\sigma$ is a $\beta$-good solution.
Let $\cZ_{\beta}(t)$ be the number of $\beta$-good solutions $\tau\in\cS(\PHId)$ such that $\dist(\sigma,\tau)=t$.
Then the symmetry properties of $\PHId$ imply the following.

\begin{fact}
For any $\vec d$ we have
	$\Erw\brk{\cZ_\beta^2}=\Erw\brk{\cZ_\beta|\Sigma}\cdot\Erw\brk{\cZ_\beta}.$
\end{fact}
Thus, we need to compare $\Erw\brk{\cZ_\beta|\Sigma}$ with $\Erw\brk{\cZ_\beta}$. 
Let $\delta=2^{-k/3}$.
By the linearity of expectation and by the symmetry of $\cS(\PHI)$ with respect to inversion, for any $\vec d$
	\begin{eqnarray}\nonumber
	\Erw\brk{\cZ_\beta|\Sigma}&=&
		\sum_{t=0}^n\Erw\brk{\cZ_\beta(t)|\Sigma}\leq2\sum_{t=0}^{n/2}\Erw\brk{\cZ_\beta(t)|\Sigma}\\
		&=&2\sum_{t\leq(\frac12-\delta)n}\Erw\brk{\cZ_\beta(t)|\Sigma}+2\sum_{(\frac12-\delta)n<t\leq\frac12n}\Erw\brk{\cZ_\beta(t)|\Sigma}\nonumber\\
		&\leq&2\abs{\cbc{\tau\in\cS(\PHI_{\vec d}):\dist(\sigma,\tau)/n\leq\frac12-2^{-k/3}}}+2\sum_{(\frac12-\delta)n<t\leq\frac12n}\Erw\brk{\cZ_\beta(t)|\Sigma}
			\nonumber\\
		&\leq&2\exp\brk{n((1-\beta)\exp(-\lambda)\ln2 - O_k(k^{13}4^{-k}))}+2\sum_{(\frac12-\delta)n<t\leq\frac12n}\Erw\brk{\cZ_\beta(t)|\Sigma}
			\label{eqsecondTrick}
	\end{eqnarray}
by the definition of $\beta$-good. Let
\[
	r^*=2^{k-1}\ln2-\bc{\frac{\ln2}2+\frac14}- k^{14} 2^{-k}.
\]

\begin{lemma}\label{Lemma_betaExists}
For any $r<r^*$ there exists $0<\beta\leq\frac12$ such that
for $\vec d$ chosen from $\vec D$ \whp
	$$\Erw\brk{\cZ_{\beta}}\geq\exp\brk{n((1-\beta)e^{-\lambda}\ln2 ) + k^{14}2^{-k+1}}.$$
\end{lemma}
\begin{proof}
This follows from \Prop~\ref{Prop_realFirst} and a little bit of calculus.
\qed\end{proof}

As a next step, we are going to bound the second summand in~(\ref{eqsecondTrick}).
This is technically the most demanding part of this work.
In Appendix~\ref{Apx_centreTerms} we are going to prove the following.

\begin{lemma}\label{Lemma_centreTerms}
Let $\delta = 2^{-k/3}$. There is a number $C=C(k)$ such that for a degree sequence $\vec d$ chosen
from $\vec D$ we have \whp
	$$\sum_{(\frac12-\delta)n<t\leq\frac12n}\Erw\brk{\cZ_\beta(t)|\Sigma}\leq C\cdot\Erw\brk{\cZ_\beta}.$$
\end{lemma}

\begin{corollary}\label{Cor_centreTerms}
For any $r<r^*$ there is $0<\beta\leq\frac12$ such that
	$\Erw\brk{\cZ_\beta|\Sigma}\leq C\cdot\Erw\brk{\cZ_\beta}$
for some constant $C=C(k)>1$.
\end{corollary}
\begin{proof}
This follows directly from~(\ref{eqsecondTrick}), \Lem~\ref{Lemma_betaExists}, and \Lem~\ref{Lemma_centreTerms}.
\qed\end{proof}

\noindent{\em Proof of \Thm~\ref{Thm_NAE} (lower bound).}
By \Cor~\ref{Cor_centreTerms} and the Paley-Zygmund inequality, for any $r<r^*$ for a random
$\vec d$ chosen from the distribution $\vec D$ we have \whp
	\begin{equation}\label{eqlowerboundfinal1}
	\pr\brk{\PHId\mbox{ has an NAE-solution}}\geq\pr\brk{\cZ_\beta>0}\geq1/C.
	\end{equation}
Since $\vec D$ is precisely the distribution of the degree sequence of the uniformly random formula $\PHI$,
we have
	$$\Erw_{\vec d}\brk{\pr\brk{\PHId\mbox{ has an NAE-solution}}}=\pr\brk{\PHI\mbox{ has an NAE-solution}},$$
where the expectation on the left hand side ranges over $\vec d$ chosen from $\vec D$.
Therefore, (\ref{eqlowerboundfinal1}) implies that
	\begin{equation}\label{eqlowerboundfinal2}
	\pr\brk{\PHI\mbox{ has an NAE-solution}}\geq\frac1C-o(1),
	\end{equation}
which remains bounded away from $0$ as $n\ra\infty$.
Hence, (\ref{eqlowerboundfinal2}) implies that $\rkNAE\geq r^*$,
	as the $k$-NAESAT threshold is sharp.
\qed

\subsection{Proof of \Prop~\ref{Prop_realFirst}}\label{Sec_realFirst}

We begin with the following simple observation.

\begin{lemma}\label{Lemma_NAEsolution}
For any $\vec d$ and any $\sigma\cbc{0,1}^V$ we have
	$\pr\brk{\sigma\in\cS(\PHId)}=(1-2^{1-k})^m.$
\end{lemma}
\begin{proof}
We may assume without loss that $\sigma=\vecone$.
Then $\sigma$ is a solution iff each clause has both a positive and a negative literal.
Since the \emph{signs} of the literals are chosen uniformly and independently, the assertion follows.
\qed\end{proof}
We defer the proof of the following result to \Sec~\ref{Sec_occ}.

\begin{proposition}\label{Prop_occ}
Let $\vec d$ be a chosen from $\vec D$.
Then \whp\ we have
	$$\frac1n\ln\pr\brk{\sigma\mbox{ has Properties 1.\ and 2.\ from Definition~\ref{Def_good}}~|~\sigma\in\cS(\PHId)}=
		f(\beta)+O_k(k 4^{-k}).$$
\end{proposition}
To continue, we need the following basic fact about the random degree distribution $\vec d$.
For a set $S\subset V$ we let $\Vol(S)=\sum_{x\in S}d_x$.

\begin{lemma}\label{Lemma_degs}
Let $\vec d$ be chosen from $\vec D$.
Then \whp\ the following is true.
	\begin{equation}\label{eqdegs}
	\parbox[t]{13cm}{For any set $S\subset V$ we have $\Vol(S)\leq10\max\cbc{kr|S|,|S|\ln(n/|S|)}.$}
	\end{equation}
\end{lemma}
\begin{proof}
For any fixed $S\subset V$ the volume $\Vol(S)=\sum_{x\in S}d_x$ is a sum of independent Poisson variables $\Po(kr)$.
Hence, $\Vol(S)=\Po(|S|kr)$, and the lemma follows from a straight first moment argument.
\qed\end{proof}
%
%The following lemma deals with the distribution of the number of clauses that each variable supports.

%\begin{lemma}\label{Lemma_reds}
%Let $B$ be the number of solutions $\sigma\in\cS(\PHI)$ for which there exists a set $S\subset V$ such that
%the total number of clauses supported by variables in $S$ exceeds
%	$
%\end{lemma}
%\begin{proof}
%\qed\end{proof}

Let us call $S\subset V$ \emph{dense} if each variable in $S$ supports at least two clauses
that each feature another variable from $S$.

\begin{lemma}\label{Lemma_dense}
Let $\vec d$ be chosen from $\vec D$ and let $\sigma\in\cbc{0,1}^V$.
Let $\cA$ be the event that $\sigma\in\cS(\PHId)$ and that $\sigma$ satisfies
	conditions 1.--2.\ in Definition~\ref{Def_good}.
Then \whp
	\begin{eqnarray*}
	\pr\brk{\mbox{there is a dense $S\subset V$, $|S|\leq n/k^5$}~|~\cA}=o(1).
	\end{eqnarray*}
\end{lemma}
\begin{proof}
We may assume that $\vec d$ satisfies~(\ref{eqdegs}).
Let $\cD(S)$ be the event that $S\subset V$ is dense.
We claim that
	\begin{eqnarray*}
	\pr_{\PHId}\brk{\cD(S)}&\leq&k^{2|S|}\cdot\frac{(k-1)^{2|S|}\Vol(S)^{2|S|}}{(krn/2)^{2|S|}}
		\leq\bcfr{2k^2\Vol(S)}{krn}^{2|S|}.
	\end{eqnarray*}
Indeed, the factor $k^{2|S|}$ accounts for the number of ways to choose the two relevant
clauses supported by each variable, and the second factor bounds the probability that each of these
clauses contains another occurrence of a variable from $S$.
Now, (\ref{eqdegs}) yields
	\begin{eqnarray*}
	\pr_{\PHId}\brk{\cD(S)}&\leq&\bcfr{2k^2|S|\ln(n/|S|)}{n}^{2|S|}.
	\end{eqnarray*}
For $0<s\leq 1/k^5$ let $X_s$ be the number of sets $S$ of size $|S|=sn$ for which $\cD(S)$ occurs.
Then
	\begin{eqnarray*}
	\Erw\brk{X_s}&\leq&\bink{n}{sn}\brk{2k^2s\ln(1/s)}^{2sn}
		\leq\brk{\frac{\eul}s\cdot\bc{2k^2s\ln(1/s)}^2}^{sn}
			\leq\bc{4\eul k^4 s\ln^2(1/s) }^{sn}=o(1).
	\end{eqnarray*}
Summing over all possible $s$ and using Markov's inequality completes the proof.
\qed\end{proof}

\begin{lemma}\label{Lem_supporting}
The expected number of solutions $\sigma\in\cS(\PHI)$ in which
more than $k^42^{-k}n$ variables support at most four clauses is 
	$\leq\exp(-nk^3/2^k)$.
\end{lemma}
\begin{proof}
Fix an assignment $\sigma\in\cbc{0,1}^V$, say $\sigma=\vecone$.
Then number of clauses supported by each $x\in V$ is asymptotically Poisson with mean $\lambda$.
Let $\cE_x$ be the event that $x$ supports no more than three clauses.
Then
	$$\pr\brk{\cE_x}\leq\lambda^3\exp(-\lambda)\leq k^42^{-k-1}.$$
The events $(\cE_x)_{x\in V}$ are negatively correlated.
Therefore, the total number $X$ of variables $x\in V$ for which $\cE_x$ occurs is stochastically
dominated by a binomial variable $\Bin(n,k^42^{-k-1})$.
Hence, the assertion follows from Chernoff bounds.
\qed\end{proof}

%\begin{corollary}\label{Cor_supporting}
%Let $\vec d$ be chosen from $\vec D$.
%Then the expected number of solutions $\sigma\in\cS(\PHId)$ in which
%more than $k^42^{-k}n$ variables support at most three clauses is 
%	$\leq\exp(-nk^2/2^k)$ \whp
%\end{corollary}

Let us call a set $S\subset V$ \emph{self-contained} if each variable in $S$
supports at least two clauses that consist of variables in $S$ only.
There is a simple process that yields a (possibly empty) self-contained set $S$.
\begin{enumerate}
\item[$\bullet$] For each variable $x$ that supports at least one clause, choose such a clause $C_x$ randomly.
\item[$\bullet$] Let $R$ be the set of all variables that support at least four clauses.
\item[$\bullet$] While there is a variable $x\in R$ that supports fewer than two clauses $\PHI_i\neq C_x$ 
		that consist of variables of $R$ only, remove $x$ from $R$.
\end{enumerate}
The clauses $C_x$ will play a special role later.

\begin{lemma}\label{Lemma_selfcontained}
The expected number of solutions $\sigma\in\cS(\PHI)$ for which the above process
yields a set $R$ of size $|R|\leq(1-k^5/2^k)n$ is bounded by $\exp(-\Omega(n))$.
\end{lemma}
\begin{proof}
Let $\sigma\in\cbc{0,1}^V$ be an assignment, say $\sigma=\vecone$.
Let $Q$ be the set of all variables that support fewer than three clauses.
By \Lem~\ref{Lem_supporting} we may condition on $|Q|\leq k^42^{-k}n$.
%Now, construct a set $R\subset V\setminus Q$ as follows.
%	\begin{quote}
%	While there is a variable $x\in R$ that supports fewer than three clauses
%	that consist of variables from $R$ only, remove $x$ from $R$.
%	\end{quote}
%The final outcome of this process is a self-contained set.
Assume that its size is $|R|\leq(1-k^5/2^k)n$.
Then there exists a set $S\subset V\setminus(R\cup Q)$ of size $\frac12k^5 n/2^k\leq S\leq k^5 n/2^k$
such that each variable in $S$ supports two clauses that contain another variable from $S\cup Q$.
With $s=|S|/n$ the probability of this event is bounded by
	\begin{eqnarray*}
	\bink{m}{2sn}\brk{\frac{2^{1-k}}{1-2^{1-k}}\cdot\frac{k^2|S\cup Q|^2}{n^2}}^{2sn}
		&\leq&\brk{4\eul k^2s}^{2sn}.
	\end{eqnarray*}
Hence, the expected number of set $S$ for which the aforementioned event occurs is bounded by
	\begin{eqnarray*}
	\bink ns\brk{4\eul k^2s}^{2sn}&\leq&\brk{\frac\eul s\cdot(4\eul k^2s)^2}^{sn}
		\leq\exp(-sn).
	\end{eqnarray*}
Since $\Erw\brk{Z(\PHI)}\leq\exp(O_k(2^{-k}n))$, the assertion follows.
\qed\end{proof}

\begin{corollary}\label{Cor_selfcontained}
Let $\vec d$ be chosen from $\vec D$.
Then the expected number of solutions $\sigma\in\cS(\PHId)$ 
for which the above process
yields a set $R$ of size $|R|\leq(1-k^5/2^k)n$ is bounded by $\exp(-\Omega(n))$.
\end{corollary}
\begin{proof}
Since the random formula $\PHI$ can be generated by first choosing $\vec d$ from $\vec D$ and
then generating $\PHId$, the assertion follows from \Lem~\ref{Lemma_selfcontained}.
\qed\end{proof}

Let us call a variable $x$ is \emph{attached} if $x$ supports a clause whose other $k-1$ variables belong to $R$.

\begin{corollary}\label{Cor_attached}
\Whp\ a degree sequence $\vec d$ chosen from $\vec D$ has the following property.
Let $\sigma\in\cbc{0,1}^V$ and  let $\cA$ be the event that $\sigma\in\cS(\PHId)$ and that $\sigma$ satisfies
	Conditions 1.\ and 2.\ in Definition~\ref{Def_good}.
Moreover, let $Y$ be the number variables that support a clause but that are not attached.
Then
	$$\pr_{\PHId}\brk{Y\leq n k^{13}4^{-k}~|~\cA}=1-o(1).$$
\end{corollary}
\begin{proof}
We may assume that $\vec d$ satisfies~(\ref{eqdegs}).
Let $F=V\setminus R$.
Then~(\ref{eqdegs}) ensures that
	$\frac{\Vol(F)}{krn}\leq\frac{2k^6}{2^k}.$
Therefore, for each of the ``special'' clause $\cC_x$ that we reserved for each $x$ that supports at least one clause
the probability of containing a variable from $F\setminus\cbc x$ is bounded by
	$$(1+o_k(1))k\cdot\frac{\Vol(F)}{krn}\leq\frac{3k^7}{2^k}.$$
Furthermore, these events are negatively correlated (due to the bound on $\Vol(F)$).
Since $|V\setminus R|\leq k^5n/2^k$ \whp\ by \Cor~\ref{Cor_selfcontained}, the assertion thus follows from Chernoff bounds.
\qed\end{proof}

Let us call a variable $x\in V$ \emph{$\xi$-rigid} in a solution $\sigma\in\cS(\PHId)$
if for any solution $\tau\in\cS(\PHId)$ with $\tau(x)\neq\sigma(x)$ we have $\dist(\sigma,\tau)\geq\xi n$.

\begin{corollary}\label{Cor_rigid}
\Whp\ a degree sequence $\vec d$ chosen from $\vec D$ has the following property.
Let $\sigma\in\cbc{0,1}^V$ and  let $\cA$ be the event that $\sigma\in\cS(\PHId)$ and that $\sigma$ satisfies
	Conditions 1.\ and 2.\ in Definition~\ref{Def_good}.
Moreover, let $Y$ be the number of variables that support a clause but that are $k^{-5}$-rigid.
%
%
%Let $\vec d$ be chosen from $\vec D$ and let $\sigma\in\cbc{0,1}^V$.
%Let $\cA$ be the event that $\sigma\in\cS(\PHId)$ and that $\sigma$ satisfies
%	conditions 1.--2.\ in Definition~\ref{Def_good}.
Then %\whp\ 
	$$\pr_{\PHId}\brk{Y\leq n k^{13}4^{-k}~|~\cA}=1-o(1).$$
\end{corollary}
\begin{proof}
We condition on the event $\cA$.
By \Cor~\ref{Cor_selfcontained}, we may assume that the self-contained set $R$  has size $|R|\geq(1-k^5/2^k)n$.
Assume that there is $\tau\in\cS(\PHId)$, $\dist(\sigma,\tau)< n/k^5$, such that
	$$\Delta=\cbc{x\in R:\tau(x)\neq\sigma(x)}$$
is non-empty.
Then $\Delta$ is dense.
Indeed, every $x\in \Delta$ supports at least two clauses, and thus $\Delta$ must contain another variable from each of them.
Thus, \Lem~\ref{Lemma_dense} shows that $\abs{\Delta}\geq n/k^5$, which is a contradiction.

Hence, \whp\ all variables $x\in R$ are $k^{-5}$-rigid.
Furthermore, if a variable $y$ is attached, then for any solution $\tau$ with $\tau(y)\neq\sigma(y)$ there is $x\in R$
such that $\tau(x)\neq\sigma(x)$.
Consequently, all attached variables are $k^{-5}$-rigid \whp\
Therefore, the assertion follows from \Cor~\ref{Cor_attached}.
\qed\end{proof}

To complete the proof, we need the following fairly simple lemma.

\begin{lemma}\label{Lemma_psi}
The expected number of pairs of solutions $\sigma,\tau\in\cS(\PHI)$ such that
	$\frac{n}{k^6}\leq\dist(\sigma,\tau)\leq(\frac12-2^{-k/2})n$
is $\leq\exp(-\Omega(n))$.
\end{lemma}
\begin{proof}
For a given $0\leq\alpha\leq 1$ let $P_\alpha$ denote the
 number of pairs $\sigma,\tau\in\cS(\PHI)$ with $\dist(\sigma,\tau)=\alpha n$
As worked out in~\cite{nae}, we have
	\begin{eqnarray*}
	\frac1n\ln\Erw\brk{P_\alpha}&\leq&
		\ln2-\alpha\ln\alpha-(1-\alpha)\ln(1-\alpha)+r\ln\bc{1-2^{2-k}+2^{1-k}(\alpha^k+(1-\alpha)^k)}.
	\end{eqnarray*}
It is a mere exercise in calculus to verify that the r.h.s.\ is strictly negative for all 
	$k^{-6}\leq\alpha\leq \frac12-2^{-k/2}$.
\qed\end{proof}

\begin{corollary}\label{Cor_psi}
\Whp\ a degree sequence $\vec d$ chosen from $\vec D$ has the following property.
The expected number of pairs of solutions $\sigma,\tau\in\cS(\PHId)$ such that
	$\frac{n}{k^6}\leq\dist(\sigma,\tau)\leq(\frac12-2^{-k/2})n$
is $\leq\exp(-\Omega(n))$.
\end{corollary}

%Finally, %\Prop~\ref{Prop_realFirst} 
Combining %is a direct consequence of 
	\Lem~\ref{Lemma_NAEsolution}, 
	\Prop~\ref{Prop_occ},
	\Cor~\ref{Cor_rigid}, and
	\Cor~\ref{Cor_psi},
we obtain

\begin{corollary}\label{Cor_nochange}
\Whp\ a degree sequence $\vec d$ chosen from $\vec D$ has the following property.
Let $\sigma\in\cbc{0,1}^V$ and  let $\cA$ be the event that $\sigma\in\cS(\PHId)$ and that $\sigma$ satisfies
	Conditions 1.\ and 2.\ in Definition~\ref{Def_good}.
Then
	$$\pr_{\PHId}\brk{\mbox{3.\ in Definition~\ref{Def_good} is satisfied}~|~\cA}=1-o(1).$$
\end{corollary}
Finally, \Prop~\ref{Prop_realFirst} 
is a direct consequence of 
	\Lem~\ref{Lemma_NAEsolution}, 
	\Prop~\ref{Prop_occ}, and
	\Cor~\ref{Cor_nochange}.

\subsection{Proof of \Prop~\ref{Prop_occ}}\label{Sec_occ}

Let us begin with establishing the probable properties of $\mathbf{d}$ that we will need.
\begin{lemma}
\label{lem:propD}
Let $\mathbf{d} = (d_1, \dots, d_n)$ be from the distribution $\mathbf{D}=\mathbf{D}(k, r, n)$. Then, with high probability, for any $0 \le \alpha \le (kr)^{1/2}$, the sequence $\mathbf d$ has the following properties. First, for all $i$ such that $|i - kr| \le \alpha \sqrt{kr}$
\begin{equation}
\label{eq:DmiddleDegs}
	D_{i} = \left|\left\{j:~ d_j = i\right\}\right| = (1+o(1)) \Pr[\Po(kr) = i] \, n.
\end{equation}
Moreover, the remaining variables satisfy
\begin{equation}
\label{eq:DExceptional}
	D^{\ge \alpha} = \left|\left\{j:~ |d_j - kr| \ge \alpha \sqrt{kr}\right\}\right| \le 2e^{-\alpha^2/2}n
	\quad\text{and}\quad
	\sum_{j \in D^{\ge \alpha}} d_j \le 2e^{-\alpha^2/2}(kr)^2n.
\end{equation}
\end{lemma}
\begin{proof}
Let $P_1, \dots, P_n$ be independent $\Po(kr)$ random variables, and note that the joint distribution of $(d_1, \dots, d_n)$ and $(P_1, \dots, P_n)$, conditional on $\sum_{1 \le i \le n}P_i = krn$, coincide. Since the expectation of the sum of the $P_i$'s equals $krn$, Lemma~\ref{lem:locallimit} applied with $\delta = 0$ implies that for any event $\cal E$ we have that
\[
	\Pr[\mathbf{d} \in {\cal E}] = \Pr\left[(P_1, \dots, P_n) \in {\cal E} ~\big|~ \sum_{1 \le i \le n}P_i = krn \right] = O(n^{1/2})\,\Pr[(P_1, \dots, P_n) \in {\cal E}].
\]
In other words, it sufficient to show that the statements in the lemma hold with probability $1 - o(n^{-1/2})$ for a sequence of \emph{independent} Poisson random variables. The statements the follow from the Chernoff bounds and the fact that for any $\lambda = kr$ and $\alpha$ as assumed
\[
	\Pr[\Po(\lambda) \ge \alpha \sqrt{\lambda}] \le 2e^{-\alpha^2}
	~~\text{ and }~~
	\sum_{j :~ |j - \lambda| \ge \alpha\sqrt{\lambda}} j \Pr[\Po(\lambda) = j] \le 2e^{-\alpha^2/2} \lambda^2.
\]
\qed
\end{proof}
The aim of this section is to show that for any $\mathbf d$ satisfying the conclusions of Lemma~\ref{lem:propD}
\begin{equation}
\label{eq:expDegSeq}
\frac1n\ln \Pr[\sigma\mbox{ has Properties 1.\ and 2.\ from Definition~\ref{Def_good}} ~|~ \sigma \in {\cal S}(\Phi_{\mathbf d})] = f(\beta) + O_k(4^{-k}),
\end{equation}
i.e., \Prop~\ref{Prop_occ} holds. We will assume that $\sigma = \mathbf{1}$ throughout.

First of all, let $C$ denote the number of critical clauses. Given that $\mathbf{1}$ is a NAE-satisfying assignement, then there are for each clause in total $2^k - 2$ ways to choose the signs of the variables, each one of them being equally likely. Since the number of ways to choose the signs so as to obtain a critical clause is $2k$, the probability that a given clause is critical is $k / (2^{k-1}-1)$. Moreover, the events that different clauses are critical are independent, implying that $C$ is distributed like $\Bin(m, k/(2^{k-1}-1))$.

Note that
$
	\Exp[C ~|~ {\mathbf 1} \in {\cal S}(\Phi_{\mathbf{d}})] = m \cdot {k}/({2^{k-1}-1}) = \lambda n
$. By applying Lemma~\ref{lem:locallimit} with $\delta = 0$ we thus obtain that
\[
	\Pr[C = \lambda n ~|~ {\mathbf 1} \in {\cal S}(\Phi_{\mathbf{d}})] = \Theta(n^{-1/2}).
\]
It follows that the probability in~\eqref{eq:expDegSeq} equals
\begin{equation}
\label{eq:condoncriticalclauses}
	\Theta(n^{-1/2}) \cdot \Pr[{\mathbf 1} \text{ is $\beta$-heavy and no variable supports $\ge 3k$ clauses} ~|~ C = \lambda n \text{ and } {\mathbf 1} \in {\cal S}(\Phi_{\mathbf{d}})].
\end{equation}
In the sequel we adopt a different formulation of this probabilistic question that is based on the classical occupancy problem. Let us think of the variables as bins, such that the i$th$ bin has capacity $d_i$, where ${\mathbf{d}} = (d_1, \dots, d_n)$. In other words, we assume that the $i$th bin contains $d_i$ distinguished ``slots''. Then we throw randomly $\lambda n$ balls into the bins, i.e., the $j$th ball chooses uniformly at random one of the remaining $\sum_{1 \le i \le n} d_i - (j-1) = krn - j + 1$ available slots, for each $1 \le j \le \lambda n$. In this setting, the probability in~\eqref{eq:condoncriticalclauses} is equal to the probability that in the balls-into-bins game with the given capacity constraints the number of empty bins equals $(1-\beta)e^{-\lambda} n$, and no bin contains more than $2k$ balls. More precisely, let $R_i$, where $1 \le i \le n$, denote the number of balls selected from the $i$th bin. Then, the probability in~\eqref{eq:condoncriticalclauses} equals
\begin{equation*}
\label{eq:bbprob}
	\Pr[{\cal A} \text{ and } {\cal B}],
	\quad 
	\text{where}
	\quad
	{\cal A} = \text{``}|\{i~:~ R_i = 0\}| = (1 - \beta)e^{-\lambda}n\text{''}
	~\text{ and }~
	{\cal B} = \text{``}\forall 1\le i \le n:~R_i \le 3k\text{''}.
\end{equation*}
We will show that the probability above is $\exp\{(f(\beta) + O_k(4^{-k}))n\}$, which together with~\eqref{eq:condoncriticalclauses} completes the proof of~\eqref{eq:expDegSeq}.

In order to compute the probability of the event ``$\cal A$ and $\cal B$'' we resort to the following experiment. Instead of throwing $\lambda n$ balls into the available slots, we decide for each slot \emph{independently} with probability $\lambda / kr$ whether if receives a ball or not. Let $T$ be the total number of balls that are thrown in this setting, and let $B_i \sim \Bin(d_i, \lambda/kr)$ be the number of balls that the $i$th bin received. Since the total number of slots is $krn$, we have that $\Exp[T] = \lambda n$. Moreover, conditional on any value of $T$, the $T$ slots that receive a ball are a random subset of size $T$ of all available slots. Thus, conditional on ``$T = \lambda n$'' the joint distributions of $(R_1, \dots, R_n)$ and $(B_1, \dots, B_n)$ coincide, and by abbreviating $X_i = |\{j : B_j = i\}|$ we obtain that
\begin{equation}
\label{eq:tbd}
	\Pr[{\cal A} \text{ and } {\cal B}] = \Pr\left[X_0 = (1 - \beta)e^{-\lambda}n \text{ and } X_{>3k} = 0 ~|~ T = \lambda n\right].
\end{equation}
Before we estimate the latter probability, let us give some intuitive explanation why this should be equal to $e^{(f(\beta) + O_k(4^{-k}))n}$, i.e., why the conclusion of the proposition is true. Our assumption on the bin capacities~\eqref{eq:DmiddleDegs} guarantees that most bins have a capacity very close to $kr \approx k2^{k-1}\ln2$. Recall also that the probability that any slot receives a ball is $\lambda/kr \approx 2^{-k+1}$. This means that the expected number of balls that a typical bin receives is $\approx k$, which is far smaller than the capacity of that bin. But we can say even more: since the  number of balls that are received by a typical bin is $ \approx\Bin(kr, \lambda/kr)$, and the expected value is far less than $kr$, it is reasonable to assume that this number can be approximated well by a $\Po(\lambda)$ distribution. So, the probability that a bin remains empty is close to $e^{-\lambda}$, and then the probability that the number of empty bins is exactly $(1-\beta)e^{-\lambda}n$ should be close to $\Pr[\Bin(n, e^{-\lambda}) = (1-\beta)e^{-\lambda}n]$.
The argument then completes by applying Lemma~\ref{lem:auxf}.

Let us now put the above intuitive reasoning on a rigorous ground. First of all, note that in the right-hand side of~\eqref{eq:tbd} the condition ``$T = \lambda n$'' is \emph{global}, in the sense that it binds the values of all variables $B_1, \dots, B_n$. We can get rid of this global restriction by applying the law of total probability. We obtain that
\begin{equation}
\label{eq:binomapprox}
\begin{split}
	\Pr[{\cal A} \text{ and } {\cal B}]
	&= \frac{\Pr\left[T = \lambda n \text{ and } X_0 = (1 - \beta)e^{-\lambda}n \text{ and } X_{>3k} = 0\right]}{\Pr[T = \lambda n]} \\
	&= \Pr\left[T = \lambda n \text{ and } X_0 = (1 - \beta)e^{-\lambda}n ~|~ X_{>3k} = 0\right] \, \frac{\Pr[X_{>3k} = 0]}{\Pr[T = \lambda n]}.
\end{split}
\end{equation}
The remainder of the proof is devoted to showing the following bounds.
\begin{eqnarray}
 \Pr[T = \lambda n] &  = & \Theta(n^{-1/2}), \label{eq:upperbound3} \\
\label{eq:lowerbound2}
 \Pr[X_{>3k} = 0] & \ge & e^{-O_k(4^{-k})n},  \\
\label{eq:lowerbound1}
 \Pr\left[T = \lambda n \text{ and } X_0 = (1 - \beta)e^{-\lambda}n ~|~ X_{>3k} = 0\right] & \ge & \Pr[\Bin(n,e^{-\lambda}) = (1 - \beta)e^{-\lambda}n] \cdot e^{-O_k(k4^{-k})n}. ~~~
\end{eqnarray}
The three inequalities together with~\eqref{eq:binomapprox} imply that
\[
	\Pr\left[{\cal A} \text{ and } {\cal B}\right] 
	\ge \Pr[\Bin(n, e^{-\lambda}) = (1 - \beta)e^{-\lambda}n] \cdot e^{-O_k(k4^{-k})},
\]
and the proof of the proposition is completed after applying Lemma~\ref{lem:auxf}.

In the remainder of the proof we will write $\mathcal{D}_i$ for the set of bins with capacity $i$ and $\mathcal{D}^{\ge \alpha}$ for the set of bins with capacity smaller than $kr - \alpha\sqrt{kr}$ or larger than $kr + \alpha\sqrt{kr}$, and note that $|\mathcal{D}_i| = D_i$ and $|\mathcal{D}^{\ge \alpha}| = D^{\ge \alpha}$.

~\\
\noindent
\emph{Proof of~\eqref{eq:upperbound3}.}
Since $T$ is distributed like $\Bin(krn, \lambda/kr)$ we have that $\Exp[T] = \lambda n$. The result then follows by applying Lemma~\ref{lem:locallimit} with $\delta = 0$ to $T$.

~\\
\noindent
\emph{Proof of~\eqref{eq:lowerbound2}.}
Recall that the number of bins with capacity $i$ is denoted by $D_i$. Since the number of balls in a bin with capacity $i$ is distributed like $\Bin(i, \lambda/kr)$, and these variables are all independent, we obtain that
\begin{equation}
\label{eq:prallsmaller3k}
\begin{split}
	\Pr[X_{>3k} = 0] &= \prod_{i\ge 0} \Pr\left[\Bin\left(i, {\lambda}/{kr}\right) \le 3k\right]^{D_i} \\
	&\ge \prod_{i:~ |i-kr|< k\sqrt{kr}} \Pr\left[\Bin\left(i, {\lambda}/{kr}\right) = 0\right]^{D_i} \cdot \prod_{i:~ |i-kr|\ge k\sqrt{kr}}\Pr\left[\Bin\left(i, {\lambda}/{kr}\right) \le 3k\right]^{D_i}.
\end{split}
\end{equation}
Our assumption~\eqref{eq:DExceptional} guarantees that $\mathbf d$ is such that $$\sum_{i:~|i-kr| \ge k\sqrt{kr}} iD_i = \sum_{j \in \mathcal{D}^{\ge k}} d_j \le 2e^{-k^2/2}(kr)^2n.$$ Thus, if $k$ is sufficiently large, the last term in~\eqref{eq:prallsmaller3k} can be bounded with
\begin{equation}
\label{eq:exceptionalinegligible}
	\prod_{i:~|i-kr| \ge k\sqrt{kr}} \Pr\left[\Bin\left(i, {\lambda}/{kr}\right) = 0\right]^{D_i}
	= \prod_{j \in\mathcal{D}^{\ge k}} \left(1- \frac{\lambda}{kr}\right)^{d_j}
	= \left(1- \frac{\lambda}{kr}\right)^{\sum_{j \in \mathcal{D}^{\ge k}} d_j}
	\ge e^{-e^{-k^2/3}n}.
\end{equation}
Let us now consider the terms involving all $i$ such that $|i-kr|< k\sqrt{kr}$ in~\eqref{eq:exceptionalinegligible}. By using the estimate $\binom{a}{b} \le (ea/b)^b$ we infer that for any such $i$ and sufficiently large $k$ we have
\begin{equation}
\label{eq:tailPo}
	\Pr\left[\Bin\left(i, \frac{\lambda}{kr}\right) > 3k\right] \le \binom{i}{3k}\left(\frac{\lambda}{kr}\right)^{3k} \le \left(\frac{e i }{3k} \frac{\lambda}{kr}\right)^{3k} \le \left(\frac{e kr(1+o_k(1)) }{3k} \frac{k \ln 2 (1+o_k(1))}{kr}\right)^{3k} \le 4^{-k}.
\end{equation}
Thus, since $\sum_{i \ge 0} D_i = n$, by using the fact $1 - x = e^{-x - \Theta(x^2)}$, valid for all $|x| \le 1$,
\begin{equation*}
	\prod_{i:~|i-kr| < k\sqrt{kr}} \Pr\left[\Bin\left(i, {\lambda}/{kr}\right) \le 3k \right]^{D_i}
	\ge \prod_{i:~|i-kr| < k\sqrt{kr}} (1-4^{-k})^{D_i} = e^{-4^{-k}n - \Theta_k(4^{-2k})n}.
\end{equation*}
This result, together with~\eqref{eq:exceptionalinegligible} and~\eqref{eq:prallsmaller3k} finally prove~\eqref{eq:lowerbound2}.

~\\
\noindent
\emph{Proof of~\eqref{eq:lowerbound1}.} Note that
\begin{equation}
\label{eq:refIntuition}
	\Pr[\Bin(n, e^{-\lambda}) = (1-\beta)e^{-\lambda}n] = \binom{n}{(1-\beta)e^{-\lambda}n} (e^{-\lambda})^{(1-\beta)e^{-\lambda}n} (1 - e^{-\lambda})^{(1 -(1-\beta)e^{-\lambda})n}.
\end{equation}
In the following proof we will approximate the probability of the event ``$T = \lambda n \text{ and } X_0 = (1 - \beta)e^{-\lambda}n$'', conditional on $X_{>3k} = 0$, by the right-hand side of the above equation times an error term, which is of order $\exp\{-O_k(k4^{-k})n\}$. In particular, we will identify the most relevant objects that contribute precisely these terms to the desired probability.

In order to prove a lower bound for the probability of the event ``$T = \lambda n \text{ and } X_0 = (1 - \beta)e^{-\lambda}n$'' we will consider only specific configurations of balls that lead to the desired outcome. More precisely, let $\mathbf{b} = (b_1, \dots, b_n)$ denote a possible outcome of the random experiment that we study, where $b_i$ denotes the number of balls in the $i$th bin. We will call $\mathbf{b}$ \emph{balanced} if it has the following properties:
\begin{enumerate}
	\item Let $j \in \mathcal{D}_i$, where $|i - kr| \ge k\sqrt{kr}$. Then $b_j = 0$. Informally, the $D^{\ge k}$ bins with ``too small'' or ``too big'' capacities are empty.
	\item Let $\mathcal{D}_i'$ denote the set of bins in $\mathcal{D}_i$ that do not receive a ball. For all $i$ such that $|i - kr| < k\sqrt{kr}$
	\[
		D_i' = |\mathcal{D}_i'| = \frac{D_i\big((1-\beta)e^{-\lambda} - D^{\ge k}/n\big)}{1 - D^{\ge k}/n}.
	\]
	Informally, the fraction of empty bins among those in $\mathcal{D}_i$ is the same (and approximately equal to $(1-\beta) e^{-\lambda}$) for all relevant $i$.
	\item Let $T_i$ denote the total number of balls in all bins in $\mathcal{D}_i$. Then, for all $i$ such that $|i - kr| < k\sqrt{kr}$
	\[
		T_i = t_i = \frac{D_i-D_i'}{1 - (1-\beta)e^{-\lambda}} \, \frac{\lambda i}{kr} \cdot x,
	\]
	where $x$ is chosen such that the sum of all $t_i$ is $\lambda n$. As we shall see later, see~\eqref{eq:xapprox}, $x$ is very close to 1. Then again, informally this requires that the fraction of balls in the bins in $\mathcal{D}_i$ is approximately $\lambda$ for all relevant $i$.
	\item For all $1 \le i \le n$ we have $b_i \le 3k$, i.e., $X_{>3k}(\mathbf{b}) = 0$.
\end{enumerate}
By our construction, note that if $\mathbf{b}$ is balanced, then $X_0(\mathbf{b}) = (1-\beta)e^{-\lambda} n$ and $T(\mathbf{b}) = \lambda n$. Thus,
\begin{equation}
\label{eq:redBalanced}
\Pr[T = \lambda n \text{ and } X_0 = (1 - \beta)e^{-\lambda}n ~|~ X_{>3k} = 0] \ge 
\Pr[(B_1, \dots, B_n) \text{ is balanced}].
\end{equation}
In the sequel we will estimate the latter probability. First of all, note that the number of ways to choose the empty bins in a balanced $\mathbf b$ is
\begin{equation}
\label{eq:choicesbalanced}
	\prod_{i:~|i-kr| < k\sqrt{kr}} \binom{D_i}{D_i'}.
\end{equation}
Note that bins contained in $\mathcal{D}^{\ge k}$ do not have to be counted explicitly, since they are contained in the set of empty bins per definition.  Let us write $\Bin_{i,j}(N,p)$ for a binomially distributed random variable that is conditioned on being in the interval $[i,j]$. Then,
after having fixed the locations of the empty bins, the probability that $(B_1, \dots, B_n)$ is balanced with precisely the chosen set of empty bins is
\begin{equation}
\label{eq:probbalanced}
	\prod_{i: |i-kr| \ge k\sqrt{kr}} \Pr\left[\Bin_{0,3k}\Big(i, \frac{\lambda}{kr}\Big) = 0\right]^{D_i}
	\cdot
	\prod_{i: |i-kr| < k\sqrt{kr}} \Pr\left[\Bin_{0,3k}\Big(i, \frac{\lambda}{kr}\Big) = 0\right]^{D_i'}
	\Pr\left[\mathcal{T}_i ~|~ X_{>3k} = 0\right],
\end{equation}
where $\mathcal{T}_i$ is the event ``$T_i = t_i$ and $\forall j\in \mathcal{D}\setminus \mathcal{D}_i':~ B_j \ge 1$''. Let $T_i'$ be a sum of $D_i - D_i'$ independent variables, which are distributed like $\Bin_{1,3k}(i, \lambda/rk)$. Then
\begin{equation}
\label{eq:ti'}
	\Pr[\mathcal{T}_i ~|~ X_{>3k} = 0] = \Pr\left[T_i' = \frac{D_i-D_i'}{1 - (1-\beta)e^{-\lambda}} \, \frac{\lambda i}{kr} \cdot x\right] \, \Pr[\Bin_{0,3k}(i, \lambda/rk) \ge 1]^{D_i - D_i'}.
\end{equation}
The probability that $(B_1, \dots, B_n)$ is balanced is then the product of the terms in~\eqref{eq:choicesbalanced} and~\eqref{eq:probbalanced}. In the remaining proof we will estimate the five terms in~\eqref{eq:choicesbalanced}--\eqref{eq:ti'}.

We begin with estimating the product in~\eqref{eq:choicesbalanced}. Let $\alpha$ be such that $D_i' = \alpha D_i$, and note that $\alpha$ is independent of $i$. Since $0 \le D^{\ge k} \le 2e^{-k^2/2} n$, see~\eqref{eq:DExceptional}, we obtain that
\begin{equation}
\label{eq:boundAlpha}
	\alpha = \frac{(1-\beta)e^{-\lambda} - D^{\ge k}/n}{1 - D^{\ge k}/n} = (1 - \beta)e^{-\lambda} + \Theta(1) \, e^{-k^2/2}.
\end{equation}
By applying Proposition~\ref{prop:binomials} with $\alpha = (1-\beta)e^{-\lambda}$ and $\eps = \Theta(1) \, e^{-k^2/2}$ we infer that
\[
	\prod_{i: |i-kr| < k\sqrt{kr}} \binom{D_i}{D_i'}
	= \prod_{i: |i-kr| < k\sqrt{kr}} \frac{\Theta(1)}{\sqrt{\alpha(1-\alpha) D_i}}e^{(H(\alpha) + O_k(ke^{-k^2/2}))D_i} = e^{H(\alpha)(n - D^{\ge k}) + O_k(ke^{-k^2/2})n} .
\]
By using once more the fact $0 \le D^{\ge k} \le 2e^{-k^2/2} n$ and by applying Proposition~\ref{prop:binomials} we infer that
\begin{equation}
\label{eq:choicesBalFinal}
	\prod_{i: |i-kr| < k\sqrt{kr}} \binom{D_i}{D_i'} = \binom{n}{(1-\beta)e^{-\lambda} n} \cdot e^{O_k(e^{-k^2/3})n}.
\end{equation}
This estimate contributes the binomial coefficient in~\eqref{eq:refIntuition} to our lower bound for the probability in~\eqref{eq:redBalanced}. It remains to bound the expression in~\eqref{eq:probbalanced}. Let us begin with considering the first product, which accounts for all $i$ that deviate significantly from $kr$. Since $\Pr[\Bin_{i,j}(N,p) = \ell] \ge \Pr[\Bin(N,p) = \ell]$ for all $N,p,i,j$ and $i \le \ell \le j$
%Recall that $\mathbf d$ is such that $\sum_{j \in \mathcal{D}^{\ge k}} d_j \le 2e^{-k^2/2}(kr)^2n$, see~\eqref{eq:DExceptional}. Thus
we have
\begin{equation}
\label{eq:exceptionalpfinal}
	\prod_{i:~|i-kr| \ge k\sqrt{kr}} \Pr\left[\Bin_{0,3k}\left(i, \frac{\lambda}{kr}\right) = 0\right]^{D_i}
	\stackrel{\eqref{eq:exceptionalinegligible}}{\ge}
	e^{-e^{-k^2/3}n}.
	%\prod_{i: |i-kr| \ge k\sqrt{kr}} \Pr\left[\Bin\left(i, \frac{\lambda}{kr}\right) \le 3k\right]^{-D_i}
	%\ge e^{-e^{-k^2/3}n}.
\end{equation}
Let us consider the middle term in~\eqref{eq:probbalanced}. Using again the property $\Pr[\Bin_{i,j}(N,p) = \ell] \ge \Pr[\Bin(N,p) = \ell]$ and the facts $1-x = e^{-x - \Theta(x^2)}$ and $\lambda = k \ln 2 + O_k(k2^{-k})$ and $r = 2^{k-1}\ln 2 - c$ we obtain
\[
	\prod_{i: |i-kr| < k\sqrt{kr}} \Pr\left[\Bin_{0,3k}\left(i, \frac{\lambda}{kr}\right) = 0\right]^{D_i'}
	\ge \exp\left\{-\left(\frac{\lambda}{kr} + O_k(4^{-k})\right)\alpha \sum_{i:~|i-kr| < k\sqrt{kr}} i D_i\right\}.
\]
By using again the property of $\mathbf d$ in~\eqref{eq:DExceptional} we infer that
\[
	\sum_{i: |i-kr| < k\sqrt{kr}} i D'_i = krn - \sum_{j \in \mathcal{D}^{\ge k}} d_j
	= krn - O_k(e^{-k^2/2})n.
\]
Recall that $\alpha = (1-\beta)e^{-\lambda} + \Theta(1) \, e^{-k^2/2}$. Thus the middle term in~\eqref{eq:probbalanced} is at least
\begin{equation}
\label{eq:pFinal}
	\prod_{i: |i-kr| < k\sqrt{kr}} \Pr\left[\Bin_{0,3k}\left(i, \frac{\lambda}{kr}\right) = 0\right]^{D_i'}
	%= \exp\left\{-\left(({\lambda} + O_k(k2^{-k}))\alpha + O_k(e^{-k^2/2})\right)\right\}
	\ge (e^{-\lambda})^{(1-\beta)e^{-\lambda}n} \cdot e^{-O_k(k4^{-k})}.
\end{equation}
This estimate contributes the $(e^{-\lambda})^{(1-\beta)e^{-\lambda}n}$ term in~\eqref{eq:refIntuition} to our lower bound for the probability in~\eqref{eq:redBalanced}. We finally consider the probability of the event $\mathcal{T}_i$ in~\eqref{eq:probbalanced}, c.f.\ also~\eqref{eq:ti'}. The last term in~\eqref{eq:ti'} can be bounded as follows. First, note that
\[
\prod_{i: |i-kr| < k\sqrt{kr}} \Pr[\Bin_{0,3k}(i, \lambda/rk) \ge 1]^{D_i - D_i'}
\ge
\prod_{i: |i-kr| < k\sqrt{kr}} \Pr[1 \le \Bin(i, \lambda/rk) \le 3k]^{D_i - D_i'}
\]
By using~\eqref{eq:tailPo} and the fact $1 - x = e^{-x - \Theta(x^2)}$, where $0 \le x \le 1$, we obtain
\[
\begin{split}
	\Pr[1 \le \Bin(i, \lambda/kr) \le 3k] \ge 1 - (1-\lambda/kr)^i - 4^{-k}
	=\exp\{-(1 - \lambda/kr)^i + O_k(4^{-k})\}.
\end{split}
\]
With this estimate at hand we can bound the last term in~\eqref{eq:ti'}. We get that
\[
	\prod_{i: |i-kr| < k\sqrt{kr}} \Pr\left[\Bin_{0,3k}\Big(i, \frac{\lambda}{kr}\Big) \ge 1\right]^{D_i - D_i'}
	\ge \exp\left\{-(1-\alpha) \sum_{i: |i-kr| < k\sqrt{kr}} (1 - \lambda/kr)^i D_i + O_k(4^{-k})n\right\}.
\]
Our assumption~\eqref{eq:DmiddleDegs} on $\mathbf{d}$ guarantees that $D_i = (1+o(1))\Pr[\Po(kr)=i]n$. Thus, the sum in the previous equation is at most
\[
	(1+o(1)) n \, \sum_{i\ge 0} (1 - \lambda/kr)^i \, \Pr[\Po(kr) = i] = (1+o(1)) e^{-\lambda} n,
\]
from which we get that, by applying again the fact $1 - x = e^{-x - \Theta(x^2)}$,
\begin{equation}
\label{eq:1-pfinal}
	\prod_{i: |i-kr| < k\sqrt{kr}} \Pr\left[\Bin_{0,3k}\Big(i, \frac{\lambda}{kr}\Big) \ge 1\right]^{D_i - D_i'}
	\ge (1 - e^{-\lambda})^{(1 - (1-\beta)e^{-\lambda})n} \cdot e^{-O_k(k4^{-k})}.
\end{equation}
This estimate contributes the last missing term in~\eqref{eq:refIntuition} to our lower bound for the probability in~\eqref{eq:redBalanced}.

It remains to bound the probability for the event ``$T_i' = t_i$'' in~\eqref{eq:ti'}, for all $i$ with the property $|i - kr|<k\sqrt{kr}$. Recall that $t_i = \frac{D_i-D_i'}{1 - (1-\beta)e^{-\lambda}} \, \frac{\lambda i}{kr} \cdot x$, where $x$ is such that the sum of the $t_i$'s is $\lambda n$. Let us begin with estimating the value of $x$. Note that
\[
	\lambda n = x \sum_{i:~|i - kr|<k\sqrt{kr}}t_i
	= \frac{x\lambda(1-\alpha)}{kr(1-(1-\beta)e^{-\lambda})}\sum_{i:~|i - kr|<k\sqrt{kr}}iD_i.
\]
Recall~\eqref{eq:boundAlpha}, which guarantees that $\alpha = (1-\beta)e^{-\lambda} + \Theta(1)e^{-k^2/2}$. Moreover, the property~\eqref{eq:DExceptional} allows us to assume for large $k$ that $\sum_{j \in \mathcal{D}^{\ge k}} d_j \le e^{-k^2/3}n$. Thus, the above equation simplifies to
\begin{equation}
\label{eq:xapprox}
	\lambda n
	= \frac{x\lambda(1-(1-\beta)e^{-\lambda} + O_k(e^{-k^2/2}))}{kr(1-(1-\beta)e^{-\lambda})} (1 - O_k(e^{-k^2/3}))krn
	\implies
	x = 1 + O_k(e^{-k^2/3}).
\end{equation}
Let us now return to our original goal of estimating the probability for the event ``$T_i' = t_i$'' in~\eqref{eq:ti'}. Recall that $T_i'$ is the sum of $D_i-D_i'$ independent variables, all distributed like $\Bin_{1,3k}(i,\lambda/kr)$. We will apply Lemma~\ref{lem:locallimit}. First of all, note that
\[
	\Exp[\Bin_{1,3k}(i,\lambda/kr)] = \frac{\frac{i\lambda}{kr} - \sum_{j > 3k}j\Pr[ \Bin(i,\lambda/kr) =j]}{\Pr[1\le \Bin(i,\lambda/kr) \le 3k]} = \frac{i\lambda}{kr} + \Theta_k(k2^{-k}),
\]
and similarly, since $i = \Theta(1) kr$, that
\[
	\sigma^2 = \Var[\Bin_{1,3k}(i,\lambda/kr)] = \Theta(1) \, \frac{i\lambda}{kr} = \Theta(\lambda).
\]
Thus, the event ``$T_i' = t_i$'' is equivalent to $\text{``} T_i' = (D_i - D_i')(\Exp[\Bin_{1,3k}(i,\lambda/kr)] + \Theta_k(k^{1/2}2^{-k})\sigma) \text{''}$. By applying Lemma~\ref{lem:locallimit} we arrive at
\[
	\prod_{i:~|i - kr|<k\sqrt{kr}} \Pr[T_i' = t_i]
	= \exp\left\{\sum_{i:~|i - kr|<k\sqrt{kr}}(-\delta^2/2 + O(c\delta^3))(D_i - D_i')\right\} = \exp\{-O_k(k4^{-k})n\}.
\]
Combining this result with Equations~\eqref{eq:redBalanced}--\eqref{eq:ti'} and~\eqref{eq:choicesBalFinal}--\eqref{eq:1-pfinal} yields~\eqref{eq:lowerbound1}, as desired.

\section{Proof of \Lem~\ref{Lemma_centreTerms}}\label{Apx_centreTerms}

\subsection{Outline}

Let $\sigma=\vecone$ be the all-true assignment and let $\vec d$ be a degree sequence chosen from the distribution $\vec D$.
Let $\Sigma$ be the event that $\sigma$ is a $\beta$-good solution.
Furthermore, let $\Sigma'$ be the event that $\sigma$ is a solution that satisfies conditions 1.\ and 2.\ in Definition~\ref{Def_good}.

\begin{fact}\label{Fact_noBudge}
Let $\vec d$ be a degree sequence chosen from the distribution $\vec D$.
Then 
	$\pr\brk{\Sigma}\sim\pr\brk{\Sigma'}$ \whp
\end{fact}
\begin{proof}
This is a direct consequence of \Cor~\ref{Cor_nochange}.
\qed\end{proof}

Let $\cZ_\beta'(t)$ be the number of  solutions $\tau$ such that $\dist(\sigma,\tau)=t$ that satisfy conditions 1.\ and 2.\ in Definition~\ref{Def_good}.
Moreover' let $\cZ_\beta'$ be the number of all solutions $\tau$ that satisfy conditions 1.\ and 2.\ in Definition~\ref{Def_good}.
For $0\leq t\leq n/2$ we let
	$$\mu(t)=\Erw\brk{\cZ_\beta'(t)~|~\Sigma'}.$$
The main step of the proof lies in establishing the following proposition. %, whose proof is the content of the following subsections.

\begin{proposition}\label{Prop_half}
There is a constant $c=c(k)>0$ such that
for $\vec d$ chosen from $\vec D$  the following two statements hold \whp
\begin{enumerate}
\item We have $\mu\bc{n/2}\leq \frac{c}{\sqrt n}\cdot \Erw\brk{\cZ_\beta'}$.
\item For any $\alpha\in\brk{\frac12-2^{-k/3},\frac12}$ we have
		$\mu\bc{\alpha n}\leq\exp\brk{-c\bc{\alpha-\frac12}^2n}\mu\bc{n/2}.$
\end{enumerate}
\end{proposition}

\noindent{\em Proof of \Lem~\ref{Lemma_centreTerms} (assuming \Prop~\ref{Prop_half}).}
By Fact~\ref{Fact_noBudge} we have \whp
	\begin{eqnarray*}
	\sum_{(\frac12-2^{-k/3})n\leq t\leq n/2}\Erw\brk{\cZ_\beta(t)|\Sigma}&\sim&
			\sum_{(\frac12-2^{-k/3})n\leq t\leq n/2}\Erw\brk{\cZ_\beta(t)|\Sigma'}\\
		&\leq&\sum_{(\frac12-2^{-k/3})n\leq t\leq n/2}\Erw\brk{\cZ_\beta'(t)|\Sigma'}\\
		&=&\sum_{(\frac12-2^{-k/3})n\leq t\leq n/2}\mu(t)\\
		&\leq&c'\sqrt n\cdot\mu(n/2)\qquad\mbox{[by \Prop~\ref{Prop_half}, part~2, with $c'=c'(k)>1$]}\\
		&\leq&cc'\cdot\Erw\brk{\cZ_\beta'}\qquad\quad\mbox{[by \Prop~\ref{Prop_half}, part~1]}\\
		&\leq&(1+o(1))cc'\Erw\brk{\cZ_\beta}\qquad\quad\mbox{[by Fact~\ref{Fact_noBudge}]},
	\end{eqnarray*}
as desired.
\qed

\smallskip\noindent
The following subsections are devoted to the proof of \Prop~\ref{Prop_half}.
%We will begin by setting up the probabilistic framework.

\subsection{The probabilistic framework}

%\aco{Say what $\PHI_{ij}$ means.}

Recall that we denote the clauses of a $k$-CNF formula $\Phi$ by $\Phi_1,\ldots,\Phi_m$, i.e., $\Phi=\Phi_1\wedge\cdots\wedge\Phi_m$.
Furthermore, for each clause $\Phi_i$ we let $\Phi_{i1},\ldots,\Phi_{ik}$ signify the literals that the clause consists of, i.e.,
	$\Phi_i=\Phi_{i1}\vee\cdots\vee\Phi_{ik}$.

We are going to break down $\mu(t)$ into a sum of different terms of various types.
This requires a few definitions and a bit of notation.
Given the sequence $\vec d=(d_x)_{x\in V}$ chosen from the distribution $\vec D$, we let
	$$B=\bigcup_{x\in V}\cbc{x}\times\brk{d_x},$$
where $\brk{d_v}=\cbc{1,2,\ldots,d_v}$.
We think of the elements of $B$ as ``balls'', so that $B$ contains $d_x$ balls $(x,j)$, $j\in\brk{d_x}$, associated with each variable $x$.
A \emph{configuration} is a bijection $\pi:B\ra\brk m\times\brk k$.
Furthermore, a \emph{signature} is a map $s:\brk m\times\brk k\ra\cbc{\pm1}$.

A configuration $\pi$ and a signature $s$ give rise to a formula $\Phi(\pi,s)$ as follows: for each $(i,j)\in\brk m\times\brk k$
\begin{itemize}
\item $\Phi(s,\pi)_{ij}$ is a positive literal if $s(i,j)=1$ and a negative literal if $s(i,j)=-1$,
\item the variable underlying $\Phi(s,\pi)_{ij}$ is the variable $x$ such that $(i,j)\in\pi(x,\brk{d_x})$.
\end{itemize}
We let $\vec \pi$ denote a configuration chosen uniformly at random, and we let $\vec s$ denote a signature chosen uniformly at random
	and independently of $\pi$.

\begin{fact}\label{Fact_pis}
For any event $\cE$ we have
	$\pr\brk{\PHId\in\cE}=\pr\brk{\Phi(\vec\pi,\vec s)\in\cE}$.
\end{fact}
\begin{proof}
For each formula $\Phi$ with degree sequence $\vec d$ there are precisely
$\prod_{x\in V}d_x!$ pairs $(s,\pi)$ such that $\Phi=\Phi(s,\pi)$.
\qed\end{proof}

Thus, from now on we may work with the random formula $\Phi(\vec\pi,\vec s)$ that emerges from
choosing a random configuration and independently a signature.
This will be useful because some properties depend only on the signature, and thus we will
be able to treat them independently of the choice of the configuration.

Let $g:B\ra\cbc{\red,\blue}$ be a map that assigns a color to each ball.
For each variable $x$ we let
	$$\red_x(g)=\abs{\cbc{j\in\brk{d_x}:g(x,j)=\red}},\qquad\blue_x(g)=\abs{\cbc{j\in\brk{d_x}:g(x,j)=\blue}}.$$
Furthermore, for a pair $(g_\sigma,g_\tau)$ of maps $B\ra\cbc{\red,\blue}$
and $\tau\in\cbc{0,1}^V$ we say that $(\sigma,\tau)$ is  \emph{$(g_\sigma,g_\tau)$-valid} for a formula $\Phi$ if the following conditions
are satisfied.
\begin{enumerate}
\item[$\bullet$] Under $\sigma$ each variable $x$ supports precisely $\red_x(g_\sigma)$ clauses.
\item[$\bullet$] Under $\tau$ each variable $x$ supports precisely $\red_x(g_\tau)$ clauses.
\item[$\bullet$] The number of clauses that any $x$ supports under \emph{both} $\sigma,\tau$
	is $\abs{\cbc{j\in\brk{d_x}:g_\sigma(x,j)=g_\tau(x,j)=\red}}.$
\end{enumerate}

%An assignment $\tau\in\cbc{0,1}^V$ is \emph{$g$-valid} for a formula $\Phi$ if $\tau\in\cS(\Phi)$ and if each variable $x$
%supports precisely $\red_x(g)$ clauses.

Let $s$ be a signature and let $\pi$ be a configuration.
We call an assignment $\tau\in\cbc{0,1}^V$ \emph{$g$-valid for $(s,\pi)$} if the following two conditions are satisfied.
\begin{enumerate}
\item[$\bullet$] $\tau\in\cS(\Phi(s,\pi))$.
\item[$\bullet$] For any $(i,j)\in\brk m\times\brk k$ the following is true.
		Let $(u,v)=\pi(i,j)$.
			Then $g(i,j)=\red$ iff $|\Phi(s,\pi)_{uv}|$ supports $|\Phi(s,\pi)_{u}|$.
\end{enumerate}
In words, $\tau$ is $g$-valid for $(s,\pi)$ if $\tau$ is a solution of the formula $\Phi(s,\pi)$ induced by $s,\pi$,
and if each ball $(i,j)$ that is colored red under $g$ supports the clause that it is mapped to under $\pi$, and vice versa.

\begin{fact}\label{Fact_valid}
Let $g_\sigma,g_\tau:B\ra\cbc{\blue,\red}$.
Then
	$$\pr\brk{\mbox{$(\sigma,\tau)$ is $(g_\sigma,g_\tau)$-valid  for $\Phi(\vec s,\vec \pi)$}}=
			\pr\brk{\mbox{$\sigma$ is $g_\sigma$-valid and $\tau$ is $g_\tau$-valid for $(\vec s,\vec\pi)$}}.$$
\end{fact}
\begin{proof}
%\aco{This requires a proof.}
%Clearly, if $\sigma,\tau$ are $g_{\sigma,\tau}$-valid for $(s,\pi)$, then they are for $\Phi(s,\pi)$, too.
%Conversely, 
Let $\Phi$ be a formula such that $(\sigma,\tau)$ is $(g_{\sigma},g_\tau)$-valid for $\Phi$.
Then the total number of pairs $(s,\pi)$ with $\Phi=\Phi(s,\pi)$ such that $\sigma$ is $g_{\sigma}$-valid
and $\tau$ is $g_{\tau}$-valid for $(s,\pi)$ equals
	$$\prod_{x\in V}\prod_{c,c'\in\cbc{\red,\blue}}\abs{(\cbc{x}\times\brk{d_x})\cap g_\sigma^{-1}(c)\cap g_\tau^{-1}(c'))}!,$$
a term that is independent of $\Phi$.
\qed\end{proof}

%\begin{definition}
A \emph{profile} $\cC$ consists of two maps $g_{\sigma},g_{\tau}:B\ra\cbc{\blue,\red}$
and a set $\Gamma\subset g_\sigma^{-1}(\blue)\cap g_\tau^{-1}(\red)$ 
such that 
%with the following properties.
%\begin{enumerate}
%\item 
$\abs{g_\sigma^{-1}(\red)}=\abs{g_\tau^{-1}(\red)}=\lambda n$ and
such that $\red_x(g_\sigma),\red_x(g_\tau)\leq 3k$ for all $x\in V$.
%\item \aco{something about $\Gamma$. We have to make sure that profiles are ``disjoint''.}
%\end{enumerate}
%\end{definition}

%\begin{definition}
Let $\cC$ be a profile.
Moreover, let $\tau\in\cbc{0,1}^V$, let $s$ be a signature, and let $\pi$ be a configuration.
We say that $(\sigma,\tau,s,\pi)$ is \emph{$\cC$-valid} if the following conditions are satisfied.
\begin{enumerate}
\item $\sigma,\tau$ are $g_\sigma,g_\tau$-valid for $(s,\pi)$.
\item Let $(x,l)\in  g_\sigma^{-1}(\blue)\cap g_\tau^{-1}(\red)$.
	Let $(i,j)=\pi(x,l)$.
	Then $(x,l)\in\Gamma$ iff $\Phi(s,\pi)_i$ is $\sigma$-critical.
\end{enumerate}

In words, this means that $(\sigma,\tau,s,\pi)$ is $\cC$-valid
if $\sigma,\tau$ are solutions of the formula $\Phi(s,\pi)$  under which the
colors assigned to the literals by $g_{\sigma}$,$g_{\tau}$ ``work out'' (i.e., a ball is red iff $\pi$ puts it in a place such that
it supports the clause it occurs in), and if a ball $(x,j)$ belongs to $\Gamma$ if it supports a clause under $\tau$
that is supported by another ball under $\sigma$.

Let $\cP$ be the set of all profiles.
For any $\cC\in\cP$ and any $t$ let
	$$\mu_{\cC}(t)=\Erw\left[\abs{\cbc{\tau\in\cbc{0,1}^V:
		\dist(\sigma,\tau)=t\mbox{ and $(\sigma,\tau,\vec s,\vec\pi)$ is $\cC$-valid}}}\right],$$
where the expectation is taken over $\vec s,\vec\pi$.
%	\aco{PROBLEM: Shouldn't this be some kind of conditional probability? Or are we counting pairs here??}

\begin{fact}\label{Fact_mudecomp}
We have
	\begin{equation}\label{eqmudecomp}
	\mu(t)=\frac{\sum_{\cC\in\cP}\mu_{\cC}(t)}{2^{-n}\Erw\brk{\cZ_\beta'}}.
	\end{equation}
\end{fact}
\begin{proof}
The denominator equals the probability that $\sigma$ is a NAE-solution that
satisfies the first two conditions in Definition~\ref{Def_good}.
Furthermore, $\mu_{\cC}(t)$ accounts for the probability that the \emph{pair} $(\sigma,\tau)$ is $\cC$-valid,
	because for any $s,\pi$ and any $\tau$ there is no more than one profile $\cC\in\cP$ such that $(\sigma,\tau,s,\pi)$ is $\cC$-valid.
Hence, (\ref{eqmudecomp}) follows from Facts~\ref{Fact_pis} and~\ref{Fact_valid}.
\qed
\end{proof}
We call a profile $\cC=(g_\sigma,g_\tau,\Gamma)$ \emph{good} if
	\begin{eqnarray*}
	\frac1n \abs{g_\sigma^{-1}(\red)\cap g_\tau^{-1}(\red)}&\in&\brk{\frac{k}{3\cdot2^k},\frac{3k}{2^k}}
	~~\mbox{ and }~~
	\frac1n \abs{\Gamma}\in\brk{\frac{k^2}{3\cdot2^k},\frac{3k^2}{2^k}}.
	\end{eqnarray*}
Let $\cP_g$ be the set of all good profiles, and let $\cP_b=\cP\setminus\cP_g$.
Furthermore, let
	$$\mu_b(t)=\sum_{\cC\in\cP_b}\mu_{\cC}(t).$$
In Appendix~\ref{Sec_badProfiles} we are going to show the following.

\begin{proposition}\label{Prop_badProfiles}
\Whp\ the degree sequence $\vec d$ chosen from $\vec D$ is such that
	$$\sum_{(\frac12-2^{-k/3})n\leq t\leq\frac n2}\mu_b(t)=o(1).$$
\end{proposition}
Furthermore, in Appendix~\ref{Sec_goodProfiles} we are going to prove

\begin{proposition}\label{Prop_goodProfiles}
\Whp\ the degree sequence $\vec d$ chosen from $\vec D$ has the following property.
Let  $\cC\in\cP_g$ and let $\frac12-2^{-k/3}\leq\alpha\leq\frac12$.
Then
	$$\mu_{\cC}(\alpha n)\leq\exp\brk{-c\bc{\alpha-\frac12}^2n}\mu_{\cC}(n/2)+\exp(-\Omega(n)).$$
for a certain  $c=c(k)>0$.
\end{proposition}
%
%Finally, in Appendix~\ref{Sec_exactHalf} we will prove
We will also need the following fact.
\begin{proposition}\label{Prop_exactHalf}
\Whp\ the degree sequence $\vec d$ chosen from $\vec D$ is such that
	$\mu(n/2)\leq\frac c{\sqrt n}\Erw[\cZ_\beta']$
for a certain $c=c(k)>0$.
\end{proposition}
\begin{proof}
Note that by~\eqref{eqmudecomp} the claim is equivalent to showing
\[
	\sum_{\cC\in\cP} \mu_\cC(n/2) \le cn^{-1/2}2^{-n} \Exp[\cZ_\beta']^2.
\]
However, since $\Exp[\cZ_\beta']$ is the sum of the expectations of indicator random variables over all possible assignments, by expanding $\Exp[\cZ_\beta']^2$ we arrive at an expression that is a sum over all profiles $\cC\in\cP$. Then the results follows essentially by performing a term-by-term comparison with the left-hand side of the above inequality.
\qed
\end{proof}
\Prop~\ref{Prop_half} is an immediate consequence of~(\ref{eqmudecomp}) and \Prop s~\ref{Prop_badProfiles}, \ref{Prop_goodProfiles}, and~\ref{Prop_exactHalf}.

\subsection{Proof of \Prop~\ref{Prop_badProfiles}}\label{Sec_badProfiles}

Let $\Phi$ be a $k$-CNF and let $\sigma,\tau\in\cbc{0,1}^V$.
We say that $(i,j)\in\brk m\times\brk k$ is \emph{$\sigma$-red} if $\Phi_{ij}$ supports $\Phi_i$ under $\sigma$.
Let $\red(\sigma,\Phi)$ be the set of all $\sigma$-red pairs $(i,j)$.
We define the term \emph{$\sigma$-blue} and the set $\blue(\sigma,\Phi)$ analogously.
Furthermore, let $\Gamma(\sigma,\tau,\Phi)$ be the set of all $(i,j)$ such
that $(i,j)\in\blue(\sigma,\Phi)\cap\red(\sigma,\Phi)$ while $\Phi_i$ is critical under $\sigma$.

Finally, we call the pair $(\sigma,\tau)\in\cS(\Phi)^2$ \emph{bad} if $(\frac12-2^{-k/3})n\leq\dist(\sigma,\tau)\leq n/2$ and one of the following conditions holds:
\begin{enumerate}
\item[$\bullet$] $|\red(\sigma,\Phi)\cap\red(\tau,\Phi)|\not\in\brk{\frac{kn}{3\cdot 2^k},\frac{3\cdot kn}{2^k}}$, or
\item[$\bullet$] $|\Gamma(\sigma,\tau,\Phi)|\not\in\brk{\frac{k^2n}{3\cdot 2^k},\frac{3\cdot k^2n}{2^k}}$.
\end{enumerate}

\begin{lemma}\label{Lemma_firstBadProfile}
Let $B$ be the number of bad pairs $(\sigma,\tau)\in\cS(\PHI)^2$.
Then $\Erw\brk B=\exp(-\Omega(n))$.
\end{lemma}
\begin{proof}
Let $\sigma=\vecone$ and let $\alpha\in\brk{\frac12-2^{-k/3},\frac12}$.
Let $S(\alpha)$ be the event that $\sigma,\tau\in\cS(\PHI)$.
As shown in~\cite{nae}, we have
	$$\pr\brk{\cS}=(1-2^{2-k}+2^{1-k}(\alpha^k+(1-\alpha)^k))^m.$$
Let $R=\abs{\red(\sigma,\PHI)\cap\red(\tau,\PHI)}$.
Given that $\cS$ occurs, $R$ has a binomial distribution
	$$\Bin\bc{m,\frac{k(\alpha^k+(1-\alpha)^k)}{(2^{k-1}-1)(1-\alpha(1-\alpha)^{k-1}-(1-\alpha)\alpha^{k-1})}}.$$
For given that $\sigma$ is a solution, there are a total of $2^{k}-2$ ways to choose the signs of the $k$ literals in any clause,
	and precisely $2k$ ways to choose the signs so that the clause is critical under $\sigma$.
Given that it is, there are $n^k(1-\alpha(1-\alpha)^{k-1}-(1-\alpha)\alpha^{k-1})$ ways to choose the actual variables
that occur in the clause so as to ensure that $\tau$ is a solution, too.
(Namely, we have to avoid that either $\tau$ and $\sigma$ differ on the $\sigma$-supporting variable only,
	or that they agree on the $\sigma$-supporting variable only;
		furthermore, the probability that $\sigma$, $\tau$ differ on a randomly chosen variable is equal to $\alpha$.)
Finally, given that a given clause is $\sigma$-critical,
the probability that the clause is critical under $\tau$ and supported by the same variable as under $\sigma$ is equal to $\alpha^k+(1-\alpha)^k$
	(for $\sigma,\tau$ would either have to agree or disagree on all the $k$ variables).
	
Further, let $G=\abs{\Gamma(\sigma,\tau,\PHI)}$.
Given that $\cS$ occurs, $G$ is a binomial variable
	$$\Bin\bc{m,\frac{k(k-1)(\alpha^2(1-\alpha)^{k-2}+\alpha^{k-2}(1-\alpha)^2)}{(2^{k-1}-1)(1-\alpha(1-\alpha)^{k-1}-(1-\alpha)\alpha^{k-1})}}.$$
For in each $\sigma$-critical clause there are $k-1$ ways to choose another literal $j$ to support that clause under $\tau$,
	and to materialize this choice, $\tau$ has to either disagree with $\sigma$ on the $\sigma$-supporting literal and on literal $j$
		and agree on all other literals, or the inverse configuration must occur.

It is easily verified that for any $\alpha\in\brk{\frac12-2^{-k/3},\frac12}$ we have
	\begin{eqnarray*}
	\Erw\brk{R|\cS}&=&(1+o_k(1))\frac{krn}{2^{2k-2}}\in\brk{\frac{kn}{2^{k}},\frac{kn}{2^{k-1}}},\\
	\Erw\brk{G|\cS}&=&(1+o_k(1))\frac{k^2rn}{2^{2k-2}}\in\brk{\frac{k^2n}{2^{k}},\frac{k^2n}{2^{k-1}}}.
	\end{eqnarray*}
As $R,G|\cS$ are binomially distributed, Chernoff bounds yield
	\begin{eqnarray}\label{eqtailsR}
	\Pr\brk{R\not\in \brk{\frac{kn}{3\cdot2^{{k-1}}},\frac{3kn}{2^{k-1}}}}\leq\exp\brk{-\Omega_k\bcfr{k}{2^k}n},\\
	\Pr\brk{G\not\in \brk{\frac{k^2n}{3\cdot2^{{k-1}}},\frac{3k^2n}{2^{k-1}}}}\leq\exp\brk{-\Omega_k\bcfr{k^2}{2^k}n}.
		\label{eqtailsG}
	\end{eqnarray}
Since the \emph{total} expected number of pairs of solutions is
	$$\Erw\brk{Z^2}\leq\exp\brk{O_k(2^{-k})n},$$
the bounds (\ref{eqtailsR}) and~(\ref{eqtailsG}) imply that $\Erw\brk{B}\leq\exp(-\Omega(n))$, as claimed.
\qed\end{proof}

\Prop~\ref{Prop_badProfiles} is an immediate consequence of \Lem~\ref{Lemma_firstBadProfile},
	because the experiment of first choosing $\vec d$ from the distribution $\vec D$ and then generating $\PHId$
		yields precisely the uniform distribution $\PHI$.

\subsection{Proof of \Prop~\ref{Prop_goodProfiles}}\label{Sec_goodProfiles}

Let $\cC=(g_\sigma,g_\tau,\Gamma)\in\cP_g$. % and let $\frac12-2^{-k/3}\leq\alpha\leq\frac12$.
For $c,c'\in\cbc{\red,\blue}$ let
	\begin{eqnarray*}
	g_{c,c'}&=&g_{c,c'}(\cC)=\abs{g_\tau^{-1}(c)\cap g_\sigma^{-1}(c')}/n,\mbox{ and let}\\
	\gamma&=&\gamma(\cC)=|\Gamma|/n.
%	g_{\blue,\blue}&=&g_{\blue,\blue}(\cC)=\frac1n\abs{g_\tau^{-1}(\blue)\cap g_\sigma^{-1}(\blue)},\\
%	g_{\blue,\red}&=&g_{\blue,\red}(\cC)=\frac1n\abs{g_\tau^{-1}(\blue)\cap g_\sigma^{-1}(\red)},\\
%	g_{\red,\blue}&=&g_{\red,\blue}(\cC)=\frac1n\abs{g_\tau^{-1}(\red)\cap g_\sigma^{-1}(\red)}.
	\end{eqnarray*}
Furthermore, for any $\sigma,\tau\in\cbc{0,1}^V$ we define
	\begin{eqnarray*}
	\alpha_{c,c'}&=&\alpha_{c,c'}(\sigma,\tau,\cC)=\frac{\abs{\cbc{x\in g_\tau^{-1}(c)\cap g_\sigma^{-1}(c'):\sigma(x)=\tau(x)}}}{g_{c,c'}n},\\
	\alpha_{\Gamma}&=&\alpha_{\Gamma}(\sigma,\tau,\cC)=\abs{\cbc{(x,i)\in\Gamma:\tau(x)=\sigma(x)}}{\abs\Gamma},\\
	\vec\alpha&=&\vec\alpha(\sigma,\tau,\cC)=(\alpha_{\red,\red},\alpha_{\red,\blue},\alpha_{\blue,\red},\alpha_{\blue,\blue},\alpha_\Gamma)\in\brk{0,1}^5.
	\end{eqnarray*}
An important observation is that by symmetry, the probability for a pair $(\sigma,\tau)$ to be $\cC$-valid is governed by their
``overlap vector'' $\vec\alpha$.
More precisely, we have

\begin{fact}\label{Fact_overlap}
Let $\cC=(g_\sigma,g_\tau,\Gamma)\in\cP_g$. % and let $\frac12-2^{-k/3}\leq\alpha\leq\frac12$.
Let $\sigma,\tau,\tau'\in\cbc{0,1}^V$ be such that $\vec\alpha(\sigma,\tau,\cC)=\vec\alpha(\sigma,\tau',\cC)$.
Then
	$$\pr\brk{\mbox{$(\sigma,\tau,\vec s,\vec\pi)$ is $\cC$-valid}}=\pr\brk{\mbox{$(\sigma,\tau',\vec s,\vec\pi)$ is $\cC$-valid}}.$$
\end{fact}
Fact~\ref{Fact_overlap} motivates the following definition:
	for $\vec\alpha=\vec\alpha(\sigma,\tau,\cC)$ we let
		$$p_{\cC}(\vec\alpha)=\pr\brk{\mbox{$(\sigma,\tau,\vec s,\vec\pi)$ is $\cC$-valid}}.$$

For a real $\alpha\in(0,1)$ we  call a vector $\vec\alpha=(\alpha_{\red,\red},\ldots)$ \emph{$\alpha$-tame} if
	\begin{eqnarray*}
	\abs{\alpha_{\red,\red}-\alpha}&\leq&10/\sqrt k,\\
	\abs{\alpha_{\red,\blue}-\alpha}&\leq&2^{-k/3},\\
	\abs{\alpha_{\blue,\red}-\alpha}&\leq&2^{-k/3},\\
	\abs{\alpha_{\blue,\blue}-\alpha}&\leq&2^{-k/2},\qquad\mbox{and}\\
	\abs{\alpha_{\Gamma}-\alpha}&\leq&100/k.
	\end{eqnarray*}
Let $\cT(\alpha)$ be the set of all $\alpha$-tame vectors.
The following lemma shows that we can neglect ``overlap vectors'' $\vec\alpha$ that are not tame.

\begin{lemma}\label{Lemma_wild}
Let $\cC=(g_\sigma,g_\tau,\Gamma)\in\cP_g$. 
Let $W$ be the number of pairs $(\sigma,\tau)\in\cS(\PHI)^2$ with
	$1-\alpha=\dist(\sigma,\tau)/n\in\brk{\frac12-2^{-k/3},\frac12}$ and such that there is a profile $\cC$
		such that $\vec\alpha(\sigma,\tau,\cC)\not\in\cT(\alpha)$.
Then
	$\Erw\brk W=\exp(-\Omega(n))$.
\end{lemma}
The proof of \Lem~\ref{Lemma_wild} is based on a similar first moment argument as in the proof of \Lem~\ref{Lemma_firstBadProfile}.
%We defer the proof of \Lem~\ref{Lemma_wild} to Appendix~\ref{Sec_wild}.
Furthermore, in \Sec~\ref{Sec_Taylor} we will establish the following.

\begin{lemma}\label{Lemma_Taylor}
Let $\cC=(g_\sigma,g_\tau,\Gamma)\in\cP_g$. 
Let $\vec\alpha\in\cT(\alpha)$ for some $\alpha\in\brk{\frac12-2^{-k/3},\frac12}$.
Letting $\vec\delta=\vec\alpha-\frac12\vecone$, we have
	\begin{eqnarray*}
	\frac1n\ln\bcfr{p_{\cC}(\vec\alpha)}{p_{\cC}(\frac12\vecone)}&\leq&
		%O_k(k)\cdot\brk{g_{\red,\red}}
		O_k\bc{k}\cdot\brk{g_{\red,\red}(\delta_{\red,\red}\delta_{\blue,\blue}+\delta_{\blue,\blue}^2)
				+\gamma(\delta_{\Gamma}\delta_{\blue,\blue}+\delta_{\blue,\blue}^2)}\nonumber\\
		&&+O_k\bcfr{k^4}{2^k}\brk{\delta_{\red,\blue}\delta_{\blue,\blue}+\delta_{\blue,\red}\delta_{\blue,\blue}+\delta_{\blue,\blue}^2}.
	\end{eqnarray*}
%\aco{what are $\delta_{\xi'}$, $\delta_{\zeta'}$???}
\end{lemma}
For a number $\alpha\in\brk{\frac12-2^{-k/3},\frac12}$ 
let $p_{\cC}(\alpha)$ be the probability that for a random $\tau\in\cbc{0,1}^V$ with $\dist(\sigma,\tau)=\alpha n$
we have $\alpha(\sigma,\tau,\cC)\in\cT(\alpha)$ and $(\sigma,\tau,\vec s,\vec\pi)$ is $\cC$-valid.
We will derive the following consequence of \Lem~\ref{Lemma_Taylor} in \Sec~\ref{Sec_CorTaylor}.

\begin{corollary}\label{Cor_Taylor}
Suppose that $\alpha\in\brk{\frac12-2^{-k/3},\frac12}$ and let $\cC$ be a good profile.
Then
%	$$\frac1n\ln\bcfr{p_{\cC}(\alpha)}{p_{\cC}(1/2)}\leq O_k(k^4/2^k)\cdot\bc{\alpha-\frac12}^2.$$
	$$p_{\cC}(\alpha)\leq p_{\cC}(1/2)\cdot\exp\brk{O_k(k^4/2^k)\cdot\bc{\alpha-\frac12}^2\cdot n}+\exp(-\Omega(n)).$$
%	\frac1n\ln\bcfr{p_{\cC}(\alpha)}{p_{\cC}(1/2)}\leq O_k(k^4/2^k)\cdot\bc{\alpha-\frac12}^2.$$
%\aco{problem: we need to discount for $\alpha$-tame profiles.}
\end{corollary}

\noindent{\em Proof of \Prop~\ref{Prop_goodProfiles}.}
By \Prop~\ref{Prop_badProfiles} and \Lem~\ref{Lemma_wild}, 
for a random $\vec d$ chosen from $\vec D$ we have  \whp\
	$$\mu_{\cC}(\alpha)\leq \bink{n}{\alpha n}p_{\cC}(\alpha)+o(1).$$
Thus, it suffices to estimate $\bink{n}{\alpha n}p_{\cC}(\alpha)$.
By Stirling's formula and \Cor~\ref{Cor_Taylor},
	\begin{eqnarray*}
	\frac1n\ln\bcfr{\bink{n}{\alpha n}p_{\cC}(\alpha)}{\mu_{\cC}(1/2)}&\leq&
		\frac1n\ln\bcfr{\bink{n}{\alpha n}p_{\cC}(\alpha)}{\bink{n}{n/2}p_{\cC}(1/2)}\\
			&\leq&-(4-o_k(1))(\alpha-1/2)^2+\frac1n\ln\bcfr{p_{\cC}(\alpha)}{p_{\cC}(1/2)}\\
		&\leq&-\bc{4-O_k(\alpha-1/2)-O_k(k^4/2^k)}\cdot\bc{\alpha-\frac12}^2\\
		&=&-\bc{4-o_k(1)}\bc{\alpha-\frac12}^2,
	\end{eqnarray*}
whence the assertion follows for $k\geq k_0$ sufficiently large.
\qed

%\subsection{Proof of \Lem~\ref{Lemma_wild}}\label{Sec_wild}

\subsection{Proof of \Lem~\ref{Lemma_Taylor}}\label{Sec_Taylor}

A map $f:\brk m\times\brk k\ra\cbc{\red,\blue}$ is called a \emph{coloring}
	if for each $i\in\brk m$ there is at most one $j\in\brk k$ such that $f(i,j)=\red$.
Let $f_\sigma,f_\tau$ be colorings.
We say that the pair $f=(f_\sigma,f_\tau)$ is \emph{compatible} with a profile $\cC=(g_\sigma,g_\tau,\Gamma)$ if
	\begin{eqnarray*}
	\abs{g_\sigma^{-1}(c)\cap g_\tau^{-1}(c')}&=&\abs{f_\sigma^{-1}(c)\cap f_\tau^{-1}(c')}
		\qquad\mbox{ for any }c,c'\in\cbc{\red,\blue},\\
	\abs\Gamma&=&\abs{\cbc{i\in\brk m:\exists j\neq l:f_\sigma(i,j)=\red\wedge f_\tau(i,l)=\red}}.
	\end{eqnarray*}
%For any profile $\cC$ induces a canonical compatible pair $f^\cC=(f^\cC_\sigma,f^\cC_\tau)$ of colorings.

Let $f$ be a coloring and let $t:\brk m\times\brk k\ra\cbc{0,1}$ be a map.
%	think of $t$ as a truth assignment.
We call $(f,t)$ \emph{valid} for a signature $s$ if the following two conditions are satisfied:
\begin{enumerate}
\item[$\bullet$] for any $i\in\brk m$ there exist $j,l\in\brk k$ such that $s(i,j)(-1)^{t(i,j)}\neq s(i,l)(-1)^{t(i,l)}$.
\item[$\bullet$] if $f(i,j)=\red$, then for all $l\in\brk k\setminus\cbc j$ we have $s(i,j)(-1)^{t(i,j)}\neq s(i,l)(-1)^{t(i,l)}$.
\end{enumerate}
Intuitively, this means that any formula in which the signs are given by $s$ is NAE-satisfied if
for all $(i,j)\in\brk m\times\brk k$ the literal in position $(i,j)$ takes the value $t(i,j)$.
Furthermore, for each $(i,j)$ with $f(i,j)=\red$ the literal in position $(i,j)$ supports clause $i$
if the truth values are given by $t$.

Let $\vec\alpha\in\brk{0,1}^5$ be a vector.
Let $f=(f_\sigma,f_\tau)$ be a pair of colorings.
Let $t:\brk m\times\brk k\ra\cbc{0,1}$.
We call $(f,t)$ \emph{compatible} with $\vec\alpha$ if
	\begin{eqnarray*}
	\alpha_{c,c'}&=&\frac{\abs{t^{-1}(1)\cap f_\sigma^{-1}(c)\cap f_\tau^{-1}(c')}}{\abs{f_\sigma^{-1}(c)\cap f_\tau^{-1}(c')}}
		\qquad\mbox{for all $c,c'\in\cbc{\red,\blue}$, and}\\
	\alpha_\Gamma&=&\frac{
		\abs{t^{-1}(1)\cap\cbc{(i,l)\in\brk m\times\brk k:\exists j\neq l:f_\sigma(i,j)=\red\wedge f_\tau(i,l)=\red}}
		}{\abs{\cbc{i\in\brk m:\exists j\neq l:f_\sigma(i,j)=\red\wedge f_\tau(i,l)=\red}}}.
	\end{eqnarray*}
Let $\vec t:\brk m\times\brk k\ra\cbc{0,1}$ be uniformly distributed, and let
	$$q_{f}(\vec\alpha)=\pr_{\vec s,\vec t}\brk{(f,\vec t)\mbox{ is valid for }\vec s\,|\,\mbox{$(f,\vec t)$ is compatible with $\vec\alpha$}}.$$

\begin{fact}\label{Fact_q}
Suppose that $f$ is compatible with a profile $\cC$.
Then for any $\vec\alpha$ we have
	$p_\cC(\vec\alpha)=q_f(\vec\alpha).$
\end{fact}
\begin{proof}
Let $t:\brk m\times\brk k$ be be such that $(f,t)$ is compatible with $\vec\alpha$.
Let $\tau\in\cbc{0,1}^V$ be such that $\vec\alpha=\vec\alpha(\sigma,\tau,\cC)$.
Let $\Pi$ be the set of all $\pi:B\ra\brk m\times\brk k$ such that
	$t(\pi(x,i))=\tau(x)$ for all $x\in V$, $i\in\brk{d_x}$.
Then $\Pi$ consists of all $\pi$ that map the right ``type'' of ``ball'' to each position $(i,j)$.
Therefore,
	\begin{eqnarray*}
	\abs\Pi&=&((\alpha_\Gamma g_\Gamma n)!((1-\alpha_\Gamma)g_\Gamma n)!
			\cdot((\alpha_{\red,\red}g_{\red,\red} n)!((1-\alpha_{\red,\red})g_{\red,\red} n)!\\
			&&\quad\cdot
				((\alpha_{\red,\blue}g_{\red,\blue} n-\alpha_\Gamma g_\Gamma n)!((1-\alpha_{\red,\blue})g_{\red,\blue} n
						-(1-\alpha_\Gamma) g_\Gamma n)!\\
			&&\quad\cdot
				((\alpha_{\blue,\red}g_{\blue,\red} n-\alpha_\Gamma g_\Gamma n)!((1-\alpha_{\blue,\red})g_{\blue,\red} n
						-(1-\alpha_\Gamma) g_\Gamma n)!\\
			&&\quad\cdot
				((\alpha_{\blue,\blue}g_{\blue,\blue} n-\alpha_\Gamma g_\Gamma n)!((1-\alpha_{\blue,\blue})g_{\blue,\blue} n
						-(1-\alpha_\Gamma) g_\Gamma n)!\enspace.
	\end{eqnarray*}
Hence, $\abs\Pi$ is independent of the actual map $t$, which implies the assertion.
\qed\end{proof}

Thus, we are left to compute $q_{f}(\vec\alpha)$ for a fixed pair $f=(f_\sigma,f_\tau)$ of colorings that is compatible with the good profile $\cC$.
To facilitate this computation, we simplify the random experiment further.
Namely, let
	$$\cR=\cbc{(i,j)\in\brk m\times\brk k:f(i,j)\neq(\blue,\blue)},\qquad\cB=\brk m\times\brk k\setminus\cR.$$
For maps $t_{\red}:\cR\ra\cbc{0,1}$ and $t_{\blue}:\cB\ra\cbc{0,1}$ we let $t_{\red}\cup t_{\blue}:\brk m\times\brk k$ be the map defined by
	$$(i,j)\mapsto\left\{\begin{array}{cl}
		t_{\red}(i,j)&\mbox{ if }(i,j)\in\cR,\\
		t_{\blue}(i,j)&\mbox{ if }(i,j)\in\cB.
		\end{array}\right.$$
Furthermore, we say that $(f,t_{\red})$ is \emph{compatible with $\vec\alpha$} if there exists $t_{\blue}$
such that $(f,t_\red\cup t_\blue)$ is compatible with $\vec\alpha$.
%Moreover, we call $(f,t_\blue)$ \emph{compatible with $\alpha$} if
%	$$\abs{t_{\blue}^{-1}(1)}=\alpha_{\blue,\blue}\abs\cB.$$

Suppose that $(f,t_{\red})$ is compatible with $\vec\alpha$.
Let $\vec t_{\blue}:\cB\ra\cbc{0,1}$ be obtained by setting $\vec t_{\blue}(i,j)=1$ with probability $\alpha_{\blue,\blue}$
and $\vec t_{\blue}(i,j)=0$ with probability $1-\alpha_{\blue,\blue}$ independently for all $(i,j)\in\cB$.
Furthermore, let
	$$q_f(\vec\alpha,t_{\red})=\pr\brk{(f,t_\red\cup\vec t_{\blue})\mbox{ is valid for }\vec s|(f,t_\red\cup\vec t_{\blue})\mbox{ is compatible with }\vec \alpha}.$$

\begin{fact}\label{Fact_tred}
Suppose that $(f,t_{\red})$ is compatible with $\vec\alpha$.
Then
	$q_f(\vec\alpha)=q_f(\vec\alpha,t_{\red}).$
\end{fact}
%\begin{proof}
%This follows from the symmetry of the random experiment with repsect to permutations of clauses.
%\qed\end{proof}

\begin{lemma}\label{Lemma_magic}
Suppose that $(f,t_{\red})$ is compatible with $\vec\alpha$.
There is a number $C=C(k)>0$ such that
	$$q_f(\vec\alpha,t_{\red})\leq C\cdot \pr\brk{(f,t_\red\cup\vec t_{\blue})\mbox{ is valid for }\vec s}.$$
\end{lemma}
\begin{proof}
We have
	\begin{eqnarray}\nonumber
	q_f(\vec\alpha,t_{\red})&=&\pr\brk{(f,t_\red\cup\vec t_{\blue})\mbox{ is valid for }\vec s|(f,t_\red\cup\vec t_{\blue})\mbox{ is compatible with }\vec \alpha}\\
		&=&\pr\brk{(f,t_\red\cup\vec t_{\blue})\mbox{ is valid for }\vec s| \abs{t_{\blue}^{-1}(1)}=\alpha_{\blue,\blue}\abs\cB}\nonumber\\
		&=&\frac{\pr\brk{(f,t_\red\cup\vec t_{\blue})\mbox{ is valid for }\vec s\wedge\abs{t_{\blue}^{-1}(1)}=\alpha_{\blue,\blue}\abs\cB}}
				{\pr\brk{\abs{t_{\blue}^{-1}(1)}=\alpha_{\blue,\blue}\abs\cB}}\nonumber\\
		&=&\pr\brk{(f,t_\red\cup\vec t_{\blue})\mbox{ is valid for }\vec s}\cdot\nonumber\\
		&&\qquad		\frac{\pr\brk{\abs{t_{\blue}^{-1}(1)}=\alpha_{\blue,\blue}\abs\cB|(f,t_\red\cup\vec t_{\blue})\mbox{ is valid for }\vec s}}
						{\pr\brk{\abs{t_{\blue}^{-1}(1)}=\alpha_{\blue,\blue}\abs\cB}}.
							\label{eqmagic1}
	\end{eqnarray}

We claim that
	\begin{equation}\label{eqLLT}
	\pr\brk{\abs{t_{\blue}^{-1}(1)}=\alpha_{\blue,\blue}\abs\cB|(f,t_\red\cup\vec t_{\blue})\mbox{ is valid for }\vec s}=O(n^{-1/2}).
	\end{equation}
For given that $(f,t_\red\cup\vec t_{\blue})\mbox{ is valid for }\vec s$, $\abs{t_{\blue}^{-1}(1)}=\alpha_{\blue,\blue}\abs\cB$
is the sum of $m$ independent contributions, as the $t_\blue(i,j)$ are independent Bernoulli variables for all $(i,j)\in\cB$.
Furthermore, 
given $(f,t_\red\cup\vec t_{\blue})\mbox{ is valid for }\vec s$
for all $i$ such that $\red\not\in f_\sigma(i\times\brk k)\cup f_\tau(i\times\brk k)$ the random variable
	$\sum_{j\in\brk k}t_\blue(i,j)$
takes any value between $1$ and $k$ with non-zero probability.
Therefore, the conditional random variable $\abs{t_{\blue}^{-1}(1)}$ has a local limit theorem, see Lemma~\ref{lem:locallimit},
and~(\ref{eqLLT}) follows.

As the \emph{unconditional} distribution of $\abs{t_{\blue}^{-1}(1)}$ is just a binomial distribution with mean $\alpha_{\blue,\blue}\abs\cB$, we have
	$$\pr\brk{\abs{t_{\blue}^{-1}(1)}=\alpha_{\blue,\blue}\abs\cB}=\Omega(n^{-1/2}).$$
Combining this with~(\ref{eqmagic1}) and~(\ref{eqLLT}) yields the assertion.
\qed\end{proof}
Combining Facts~\ref{Fact_q} and~\ref{Fact_tred} with \Lem~\ref{Lemma_magic}, we obtain

\begin{corollary}\label{Cor_magic}
Suppose that $(f,t_{\red})$ is compatible with $\vec\alpha$.
Then
	$$p_\cC(\vec\alpha)\leq C\cdot \pr\brk{(f,t_\red\cup\vec t_{\blue})\mbox{ is valid for }\vec s}.$$
\end{corollary}
The crucial feature of the term
	$$\pr\brk{(f,t_\red\cup\vec t_{\blue})\mbox{ is valid for }\vec s}$$
is that in the underlying random experiment, the clauses are \emph{independent} objects,
although there are different ``types'' of clauses.
This independence property allows us to derive the following estimate.

\begin{figure}
	\begin{eqnarray*}
	\psi_{\sigma}&=&(1-k)\lambda\ln2+(r-\lambda)\ln(1-(k+1)2^{1-k}),\\
	\psi_{\red,\red}&=&g_{\red,\red}(k-1)\brk{\alpha_{\red,\red}\ln(a)+(1-\alpha_{\red,\red})\ln(1-a)},\\
	\psi_\Gamma&=&\gamma(k-2)\brk{\alpha_\gamma\ln(1-a)+(1-\alpha_\gamma)\ln a},\\
	\xi&=&g_{\red,\blue}-\gamma,\\
	\alpha_\xi&=&\frac{g_{\red,\blue}(\alpha_{\red,\blue}-\alpha_\Gamma\gamma)}{\xi},\\
	\psi_{\red,\blue}&=&-\xi\ln(2^{k-1}-k-1)+\alpha_\xi\xi\ln\bc{1-a^{k-1}-(1-a)^{k-1}-(k-1)a(1-a)^{k-2}}\\
		&&\qquad+(1-\alpha_\xi)\xi\ln\bc{1-a^{k-1}-(1-a)^{k-1}-(k-1)a^{k-2}(1-a)},\\
	\zeta&=&g_{\blue,\red}-\gamma,\\
	\alpha_\zeta&=&\frac{g_{\blue,\red}(\alpha_{\blue,\red}-\alpha_\Gamma\gamma)}{\zeta},\\
	\psi_{\blue,\red}&=&\alpha_\zeta\zeta\ln\bc{1-a^{k-1}-(1-a)^{k-1}-(k-1)a(1-a)^{k-2}}\\
		&&\qquad+(1-\alpha_\zeta)\zeta\ln\bc{1-a^{k-1}-(1-a)^{k-1}-(k-1)a^{k-2}(1-a)},\\
	\psi_{\blue,\blue}&=&(r-2\lambda+g_{\red,\red})\ln\brk{1-\frac{1+k-\eta(a)}{2^k-1}-k-1},\qquad\mbox{where}\\
	\eta(a)&=&a^k+(1-a)^k+ka(1-a)^{k-1}+ka^{k-1}(1-a)+\\
		&&\quad
			k(a(1-a)^{k-1}+(1-a)a^{k-1}+a^k+(1-a)^k+\\
		&&\qquad\qquad(k-1)a^{k-2}(1-a)^2+(k-1)a^2(1-a)^{k-2}).
	\end{eqnarray*}
\caption{The explicit expressions for \Prop~\ref{Prop_psis}.}\label{Fig_psis}
\end{figure}

\begin{proposition}\label{Prop_psis}
Suppose that $(f,t_{\red})$ is compatible with $\vec\alpha$.
Let $\cV$ be the event that $(f,t_\red\cup\vec t_{\blue})$ is valid for $\vec s$.
Let $a=\alpha_{\blue,\blue}$.
Then
	\begin{eqnarray}\label{eqpsis}
	\frac1n\ln\pr\brk\cV&=&\psi_\sigma+\psi_\Gamma+\sum_{c,c'\in\cbc{\red,\blue}}\psi_{c,c'},
	\end{eqnarray}
with the $\psi$s as shown in Figure~\ref{Fig_psis}.
%\aco{Sort out the various greek letters and double-check $a$ and $\psi_{\sigma}$!!!}
\end{proposition}
\begin{proof}
The first summand $\psi_\sigma$ accounts for the probability that $\sigma=\vecone$
is a NAE-solution and that preicsely the clauses $i$ such that $f(i,j)=\red$ for some $j\in\brk k$ are $\vecone$-critical.
There are precisely $\lambda n$ such clauses, and for each of them the probability of being critical with supporting
literal $(i,j)$ equals $2^{1-k}$.
Furthermore, for the $(r-\lambda)n$ other clauses the probability of being non-critical but NAE-satisfied equals $1-(k+1)2^{1-k}$.
Since these events depend on the signs of the literals only, they occur independently for all clauses, which explains $\psi_\sigma$.

The $\psi_{\red,\red}$ term is derived quite easily as well.
%Let $\rho n=f_{\sigma}^{-1}(red)\cap f_{tau}^{-1}(red)$ be the number of red/red `balls'.
The number of positions $(i,j)$ such that $f_\sigma(i,j)=f_\tau(i,j)=\red$ equals $g_{\red,\red}n$.
There are precisely $\alpha_{\red,\red}g_{\red,\red} n$ among these such that $t_\red(i,j)=1$.
Each such position $(i,j)$ supports its clause under $\vec t$ iff $\vec t(i,l)=1$ for all $l\in\brk k\setminus\cbc j$.
By the construction of $\vec t$, the probability of this event is $a^{k-1}$.
Similarly, the ``success probability'' is $(1-a)^{k-1}$ for all $(i,j)$ with $t_\red(i,j)=0$.

The next factor $\psi_{\Gamma}$ accounts for the number of $(i,j)\in f_\tau^{-1}(\red)\cap f_\sigma^{-1}(\blue)$
such that clause $i$ is $\sigma$-critical but supported by another literal $l\neq i$ under $\sigma$.
Each such clause contains precisely $k-2$ literals $h\in\brk k\setminus\cbc{j,l}$ such that $f_\tau(i,h)=f_\sigma(i,h)=\blue$.
If $t_\red(i,j)=1$, then $\vec t(i,h)=0$ for all $h$, which occurs with probability $(1-a)^{k-2}$.
Similarly, if $t_\red(i,j)=0$, then $\vec t(i,h)=1$ for all $h$,  the probability of which equals $a^{k-2}$.

The term $\psi_{\red,\blue}$ deals with clauses $i$ such that $(i,j)\in f_\tau^{-1}(\red)\cap f_\sigma^{-1}(\blue)\setminus\Gamma$ for some $j$.
The total number of such clauses is $\xi n$.
For each of these $\xi n$ indices $i$ we have $f_\sigma(i,l)=\blue$ for all $l\in\brk k$ (because $(i,j)\not\in\Gamma$).
Suppose that $t_\red(i,j)=1$.
Since clause $i$ is non-critical under $\sigma=\vecone$, it contains a total of $h\geq2$ literals whose signs agree with that of literal $j$.
In order for clause $i$ to be supported by literal $j$ under $\vec t$, the $h-1$ other literals $l$ whose signs agree with that of literal $j$
must take the value $\vec t(i,l)=0$, while the $k-h$ remaining literals $l$ must take value $t(i,l)=1$.
Summing over $h$ and taking into account the distribution of the signs, we obtain the overall probability in the case $t_\red(i,j)=1$:
	\begin{eqnarray*}
	\sum_{j=2}^{k-2}\frac{2\bink{k-1}{j-1}}{2^k-2k-2}(1-a)^{j-1}a^{k-j}=1-a^{k-1}-(1-a)^{k-1}-(k-1)a(1-a)^{k-2}.
	%%%%%\frac{2\bink{k-1}{j-1}}{2^k-2k-2}a^{j-1}(1-a)^{k-j}=1-a^{k-1}-(1-a)^{k-1}-(k-1)(1-a)a^{k-2}.
	%\frac{2\bink{k-1}{j-1}}{2^k-2k-2}(1-a)^{j-1}a^{k-j}
%			\brk{\sum_{j=2}^{k-2}\frac{2\bink{k-1}{j-1}}{2^k-2k-2}a^{j-1}(1-a)^{k-j}}^{(1-\alpha_\xi')\xi'n}\\
%		&=&(2^{k-1}-k-1)^{-\xi'n}\brk{1-a^{k-1}-(1-a)^{k-1}-(k-1)a(1-a)^{k-2}}^{\alpha_\xi'\xi'n}\\
%		&&\qquad	\cdot\brk{1-a^{k-1}-(1-a)^{k-1}-(k-1)(1-a)a^{k-2}}^{(1-\alpha_\xi')\xi'n}
	\end{eqnarray*}
%\aco{Problem: possibly we need to switch $\alpha$ and $1-\alpha$!!!}
The case $t_\red(i,j)=0$ is analogous to the above, and a similar argument yields $\psi_{\blue,\red}$.

Finally, %we have $Q_{\blue,\blue}$.
$\psi_{\blue,\blue}$ accounts for all clauses $i$ such that $f_\sigma(i,j)=f_\tau(i,j)=\blue$ for all $j\in\brk k$.
There are precisely $(r-2\lambda+g_{\red,\red})n$ such clauses.
Each of them is supposed to be assigned such that under both $\sigma=\vecone$ and $\vec t$
at least two literals evaluate to ``true'' and at least two evaluate to ``false''.
Given the distribution of the signature $\vec s$ and of $\vec t$, the probability of this event equals $\eta(a)$.
However, we are already conditioning on the event that each clause contains at least one literal of either sign
	(this probability is accounted for by $\psi_\sigma$).
Hence, the conditional probability of the desired outcome equals
	$\frac{\eta(a)}{1-(k+1)2^{1-k}}$.
Since the clauses are independent, the overall probability is given by $\psi_{\blue,\blue}$.
%
%This is essentially the probability that two random assignments with overlap $a$
%are satisfying in a random CSP where we forbid the $2(k+1)$ possible assignments
%that have at least $k-1$ literals of the same value.
%This probability is $F(a)^m$, where
%	\begin{eqnarray*}
%	F(a)&=&1-2(2^{1-k}+k2^{1-k})
%		+2^{1-k}(a^k+(1-a)^k+ka(1-a)^{k-1}+ka^{k-1}(1-a))\\
%		&&\qquad+k2^{1-k}(a(1-a)^{k-1}+(1-a)a^{k-1}+a^k+(1-a)^k\\
%		&&\qquad\qquad\qquad\qquad+(k-1)a^{k-2}(1-a)^2+
%			(k-1)a^{2}(1-a)^{k-2}).
%	\end{eqnarray*}
%Hence, letting $\mu n=(r-2\lambda+\rho)n$ denote the remaining number of clauses, we get
%	\begin{eqnarray*}
%	Q_{\red,\blue}&=&\bcfr{F(a)}{1-(k+1)2^{1-k}}^{\mu n}.
%	\end{eqnarray*}
%(This equals the conditional probability that $\tau$ satisfies the blue/blue clauses given that $\sigma$ does.)
%
%Thus, $Q$ provides an explicit formula for $q_{f}(a,\vec\alpha,t_\red)$.
%Indeed, this formula is independent of $t_\red$, and thus
%$p_{\cC}(\vec\alpha)=Q$.
%Hence, we need to study $Q=Q(\vec\alpha)$ near $\vec\alpha=\frac12\vecone$.
\qed\end{proof}

\noindent\emph{Proof of \Lem~\ref{Lemma_Taylor}.}
The assertion simply follows from \Prop~\ref{Prop_psis} by Taylor expanding 
the right hand side of~(\ref{eqpsis}) around $\frac12\vecone$.
\qed

\subsection{Proof of \Cor~\ref{Cor_Taylor}}\label{Sec_CorTaylor}

We begin with the following observation, which hinges upon the assumption that we work with a good profile.

\begin{proposition}\label{Prop_goodAzuma}
There is an absolute constant $c>0$ such that for a random $\vec d$ chosen from $\vec D$ the following is true \whp\
Let $\cC$ be a good profile, let $(\frac12-2^{-k/3})\leq\alpha\leq\frac12$, and let $\vec\tau$ be chosen uniformly
at random from all assignments such that $\dist(\sigma,\tau)=\alpha n$.
Then for any $\delta>0$ we have
	\begin{eqnarray*}
	\pr\brk{|\alpha_{\blue,\blue}-\alpha|>\delta}&\leq&\exp(-c\delta^2n),\\
	\pr\brk{|\alpha_{\red,\blue}-\alpha|>\delta}&\leq&\exp(-c\delta^2n),\\
	\pr\brk{|\alpha_{\blue,\red}-\alpha|>\delta}&\leq&\exp(-c\delta^2n),\\
	\pr\brk{|\alpha_{\red,\red}-\alpha|>\delta}&\leq&\exp(-g_{\red,\red}\delta^2 n/k^2),\\
	\pr\brk{|\alpha_{\Gamma}-\alpha|>\delta}&\leq&\exp(-\gamma\delta^2 n/k^2).
	\end{eqnarray*}
\end{proposition}
\begin{proof}
Recall that $\sigma=\vecone$.
%We are only going to show the upper tail bounds.
%The lower tail bounds follow analogously.
By standard monotonicity arguments, we may assume that $\vec\tau$ is obtained by letting $\vec\tau(x)=0$ with probability $\alpha$ and $\vec\tau(x)=1$ with probability $1-\alpha$ for all $x\in V$ independently.	
	% More precisely, to make this work consider the upper/lower tails for each $\alpha_{c,c'}$ separately! Then monotonicity can be used.
Furthermore, since by standard arguments the degrees $d_x$ are asymptotically independently Poisson, \whp\ the degree sequence $\vec d$ is such that
	\begin{equation}\label{eqDegSquares}
	\sum_{x\in V}d_x^2\leq10\bc{\frac1n\sum_{x\in V}d_x}^2n\leq10(kr)^2n.
	\end{equation}
Hence, we are going to assume that~(\ref{eqDegSquares}) is satisfied.

We begin by analyzing $\alpha_{\blue,\blue}$.
Switching the value $\vec\tau(x)$ of a single variable $x\in V$ can only alter the random variable $\alpha_{\blue,\blue}$ by $d_v/(g_{\blue,\blue}n)$.
Therefore, by Azuma's inequality and~(\ref{eqDegSquares}), for any $t>0$
	\begin{eqnarray}\label{eqgoodAzuma1}
	\pr\brk{\abs{\alpha_{\blue,\blue}-\Erw\brk{\alpha_{\blue,\blue}}}>t/(g_{\blue,\blue}n)}\leq\exp\brk{-\frac{t^2}{\sum_{x\in V}d_x^2}}
		\leq\exp\brk{-\frac{t^2}{10n(kr)^2}}.
	\end{eqnarray}
Since $g_{\blue,\blue}\leq\frac12 krn$ for any good profile, (\ref{eqgoodAzuma1}) yields the first inequality.

With respect to $\alpha_{\red,\blue}$, recall that in a good profile each $x\in V$ satisfies $\red_\tau(x)\leq k$
	(recall that $\red_\tau$ depends on the profile $\cC$ only).
Therefore, Azuma's inequality yields
	\begin{eqnarray}\label{eqgoodAzuma2}
	\pr\brk{\abs{\alpha_{\red,\blue}-\Erw\brk{\alpha_{\red,\blue}}}>t/(g_{\red,\blue}n)}\leq\exp\brk{-\frac{t^2}{k^2n}}.
	\end{eqnarray}
Since $g_{\red,\blue}\geq ckn$ for a certain constant $c>0$, the second claim follows from~(\ref{eqgoodAzuma2}).
A similar argument yields the third inequality.

Regarding $\alpha_{\red,\red}$, we recall that given $\cC$ we know how many ``red/red balls'' each variable has.
Since $\cC$ is good, their total number is $g_{\red,\red}n\leq k^22^{-k}n$.
In particular, there are no more than $g_{\red,\red}n$ variables that have a ``red/red ball'' in the first place.
Furthermore, switching $\vec\tau(x)$ for a single variable $x$ can alter $\alpha_{\red,\red}$ by at most
$k/(g_{\red,\red}n)$, because $\red_\tau(x),\red_\sigma(x)\leq k$ for all $x$ as $\cC$ is good.
Therefore, by Azuma's inequality
	\begin{eqnarray}\label{eqgoodAzuma3}
	\pr\brk{\abs{\alpha_{\red,\red}-\Erw\brk{\alpha_{\red,\red}}}>t/(g_{\red,\red}n)}\leq\exp\brk{-\frac{t^2}{k^2 g_{\red,\red}n}}.
	\end{eqnarray}
(The $g_{\red,\red}$ in the denominator mirrors the fact that no more than $g_{\red,\red}n$ variables have a ``red/red ball''.)
Setting $t=\delta g_{\red,\red}n$ yields the fourth inequality.
The last inequality follows from a similar argument.
\qed\end{proof}

%\noindent\emph{Proof of \Cor~\ref{Cor_Taylor}.}
%Let us first deal with the contribution $p_{\cC}^*(\alpha)$ of $\vec\tau$ with $\vec\alpha(\vec\tau)=\alpha\vecone$ and $\dist(\sigma,\tau)=\alpha n$.
%Let $\delta=\alpha-\frac12$ so that $\abs\delta\leq2^{-k/3}$.
%The total number of such $\vec\tau$ equals
%	\begin{eqnarray*}
%	O(n^{-2})\bink n{\alpha n}&=&O(n^{-5/2})\exp\brk{-\alpha\ln(\alpha)-(1-\alpha)\ln(1-\alpha)}\\
%		&\leq&O(n^{-5/2})\exp\brk{-4\delta^2+O_k(\delta^3)}.
%	\end{eqnarray*}
%Furthermore, for each such $\vec\tau$ the ``success probability'' satisfies
%	\begin{eqnarray*}
%	\frac1n\ln\bcfr{p_{\cC}(\alpha\vecone)}{p_{\cC}(\frac12\vecone)}&\leq&
%		%O_k(k)\cdot\brk{g_{\red,\red}}
%		\alpha^2\cdot\brk{O_k\bc{k}\cdot(g_{\red,\red}	+\gamma)+O_k\bcfr{k^4}{2^k}}
%			=\alpha^2O_k\bcfr{k^4}{2^k}.
%	\end{eqnarray*}

Finally, \Cor~\ref{Cor_Taylor} follows by comparing the bounds on the deviations of the individual components
of $\vec\alpha$ from \Prop~\ref{Prop_goodAzuma} with \Lem~\ref{Lemma_Taylor} and \Lem~\ref{Lemma_wild}.
\qed

%\aco{Problem: polynomial terms for $\vec\alpha$ \emph{very} close to $\alpha\vecone$?!}

%\subsection{Proof of \Prop~\ref{Prop_exactHalf}}\label{Sec_exactHalf}

\end{appendix}

\end{document}